%% file: main.tex
\documentclass[% reprint,
superscriptaddress,
frontmatterverbose, 
 amsmath,amssymb,
 pra,aps,
notitlepage,9pt]{revtex4-1}
\usepackage[english]{babel}
\usepackage{amsmath, amssymb, amsfonts, mathrsfs}
\usepackage{ntheorem}
\usepackage{subcaption}
\usepackage{graphicx}% Include figure files
\usepackage{dcolumn}% Align table columns on decimal point
\usepackage{bm}% bold math
\usepackage[shortlabels]{enumitem} 
\usepackage{dsfont}
\usepackage{chngcntr}
\usepackage{apptools}
\usepackage{mathtools}
\AtAppendix{\counterwithin{cor}{section}}
\usepackage{xcolor,soul}
\usepackage{url} 
\usepackage{multirow}

\DeclarePairedDelimiter\abs{\lvert}{\rvert}%
\DeclarePairedDelimiter\norm{\lVert}{\rVert}%

\DeclareMathOperator{\Span}{span}

\DeclareMathOperator{\Tr}{Tr}

\DeclareMathOperator{\diag}{diag}

\newtheorem{theorem}{Theorem}
\newtheorem{lemma}{Lemma}

\newtheorem{definition}{Definition}

\newtheorem{remark}{Remark}
\theoremstyle{nonumberbreak}
\newtheorem{proof}{Proof}

\usepackage[colorlinks = true,
            urlcolor=blue
            ]{hyperref}

\usepackage[capitalise]{cleveref}
\crefname{cor}{Corollary}{Corollaries}
\crefname{obs}{Observation}{Observations}
\crefname{remark}{Remark}{Remarks}
\crefname{proof}{Proof}{Proofs}
\creflabelformat{proof}{#2proof#3}
\usepackage{epstopdf}
\usepackage{tikz}
\begin{document}

%\preprint{APS/123-QED}

\title{Monogamy of correlations and entropy inequalities in the Bloch picture}
\author{Paul Appel}
\email{paul.jonas.appel@gmail.com}
\affiliation{
Institute for Quantum Optics and Quantum Information - IQOQI Vienna, Austrian Academy of Sciences, Boltzmanngasse 3, 1090 Vienna, Austria}
\author{Marcus Huber}%
\email{entangledanarchist@gmail.com}
\affiliation{
Institute for Quantum Optics and Quantum Information - IQOQI Vienna, Austrian Academy of Sciences, Boltzmanngasse 3, 1090 Vienna, Austria}
\author{Claude Kl\"ockl}%
 \email{claudio.kloeckl@reflex.at}
\affiliation{
Institute for Quantum Optics and Quantum Information - IQOQI Vienna, Austrian Academy of Sciences, Boltzmanngasse 3, 1090 Vienna, Austria}%
\affiliation{Institute of Computer Science, Masaryk University, Botanick\'{a} 68a, 60200 Brno, Czech Republic}

\date{\today}% It is always \today, today,
             %  but any date may be explicitly specified

\begin{abstract}
\noindent
We investigate monogamy of correlations and entropy inequalities in the Bloch representation. Here, both can be understood as direct relations between different correlation tensor elements and thus appear intimately related. 
To that end we introduce the \emph{split Bloch basis}, that is particularly useful for representing quantum states with low dimensional support and thus amenable to purification arguments.
Furthermore, we find dimension dependent entropy inequalities for the Tsallis 2-entropy.
In particular, we present an analogue of the strong subadditivity and a quadratic entropy inequality. These relations are shown to be stronger than subadditivity for finite dimensional cases.
\end{abstract}

\pacs{Valid PACS appear here}% PACS, the Physics and Astronomy
                             % Classification Scheme.
%\keywords{Suggested keywords}%Use showkeys class option if keyword
                              %display desired
\maketitle

%\tableofcontents
\section{Introduction}
\noindent
This articles covers two important themes of quantum information: \emph{entropy inequalities} and \emph{monogamy relations}.
The first focus of our article are entropy inequalities.
Entropy has been described by many differing mathematical definitions, but essentially the entropy of a quantum state always describes the lack of knowledge about the respective system.
Entropy inequalities govern the way information about some constituents of the system determines our knowledge about the remainder of the system.
They could be broadly equated to the ``natural laws'' of information.
Typically entropy can be either formulated in a classical or quantum version. 
Throughout this article, we will be discussing the quantum case.
How well these laws are understood varies greatly depending on the exact mathematical formulation of entropy. 
The standard formulation of the Shannon entropy \cite{ShannonEntropy} as well as its quantum analogue the von Neumann entropy have been thoroughly classified \cite{YeungCone,PippengerCone}.
Many further parametrized families of entropies are known. The two most well-known single-parameter examples are the  R\'enyi $\alpha$-entropy \cite{RenyiEntropy}  and the Tsallis $q$-entropy \cite{TsallisEntropy}.
Recently, important progress was made on the family of the  R\'enyi $\alpha$-entropies \cite{MilanRenyi,rankineq}. 
This almost completed the description of the better-known entropies. 
For the Tsallis $q$-entropy there are, to the best of our knowledge, no complete classification results known.
Even though the whole family of $q$-entropies (sometimes also referred to as q-logarithms) is  well studied \cite{Yamanoqlog}, of relevance for the field of complex systems \cite{TsallisBook,HanelThurner,ThurnerTsallis} and in case of the Tsallis $2$-entropy frequently employed in quantum mechanics under the name \emph{linear entropy}.

Note that  a number of powerful no-go results are known, specifically regarding the impossibility of linear inequalities for R\'enyi entropies for $\alpha\neq 0,1$ \cite{MilanRenyi} and the unachievability of strong subadditivity for Tsallis $q$-entropies with $q \neq 1$ \cite{PetzVirosztek}.
The Tsallis $q$-entropies directly translate to R\'enyi $\alpha$-entropies, but the impossibility of linear inequalities for the R\'enyi $2$-entropy does not imply the same for the Tsallis $2$-entropy. 
Indeed, the Tsallis $2$-entropy is subadditive \cite{AudenaertQSA}. 
Nonetheless we are able to circumvent these no-go results in the form of dimension-dependent analogues to strong subadditivity and pseudo-additivity for the linear entropy.
These analogues hold where their dimension-independent counterparts do not. 
We can hence trade dimension dependency for a broader range of applicability.
In particular we introduce a linear and another non-linear entropy inequality for the Tsallis 2-entropy.

The second focus of this article concerns correlations. 
In quantum physics some correlations exhibit a property commonly referred to as monogamy \cite{Terhal}.
Monogamy intuitively means that whenever two parties share a sufficient amount of monogamous correlations it prohibits a third party from also being correlated in a non-trivial way to the former.
In other words we could see them as a ``natural law'' limiting correlated information in analogy to the way we introduced entropy inequalities before.
This simple idea is not only of interest when trying to understand the fundamental structure of quantum correlations, but it is also the base for the security proofs of quantum key distribution (see the security proof of either Lo-Chau's protocol \cite{LoChau} or the security proof \cite{ShorPreskill} of the BB84 protocol \cite{BB84}). 
Therefore monogamy of entanglement is an essential tool for one of the most mature practical applications from the field of quantum information.
Formalizing the intuitive idea of monogamy has been the subject of a long-standing debate within the quantum information community \cite{CKW,OsborneVerstraete,Def:SquashedEntanglement,KoashiWinter}.

Let us review a few well-known facts about monogamy and motivate the introduction of dimension-dependent factors.
It is noteworthy that the \emph{strict} classical definition of monogamy, e.g. $\mu_{AB|C}\left(\rho_{ABC}\right)\geq \mu_{A|C}\left(\rho_{AC}\right)+\mu_{B|C}\left(\rho_{BC}\right)$ where $\mu$ is an entanglement measure in arbitrary dimension,  is intimately connected to the notion of entanglement:\\
It has been shown, that \emph{only} entanglement measures, in contrast to measures of other correlations, can be \emph{strictly} monogamous in arbitrary dimension \cite{AdessoPiani}.
On the other hand, even though many relations have been found to fulfill this definition for \emph{qubits}, e.g. the Coffman-Kundu-Wootters (CKW) inequality \cite{CKW}, and even though a generalization for $n$-qubit case \cite{OsborneVerstraete} is available, there exists only one known entanglement measure which fulfills this strict notion of monogamy in \emph{arbitrary dimension}: Squashed entanglement \cite{KoashiWinter}.
Unfortunately squashed entanglement is hardly computable; even numerically.
Many other entanglement measures can not be readily generalized to arbitrary dimension, e.g. the CKW inequality \cite{CounterExampleOu,CounterExampleMarcus}.
Even though common monogamy inequalities do not explicitly feature dimension dependent factors, they are indeed implicitly dimension dependent, i.e. assuming particular local dimensions.
Many arguments concerning certain aspects of monogamy are dependent on these assumptions:
For example, the maximal entanglement of two parties excludes a third party to be entangled with either of the former \emph{only} if the systems of the former have the \emph{same} local dimension. 
Furthermore, it has been shown that in general entanglement measures that quantify maximal entanglement geometrically faithful cannot be monogamous in asymptotic dimensions \cite{CounterExampleMarcus};
this is true also for squashed entanglement.

Considering all these problems in finding feasible entanglement measures which fulfill the strict monogamy relation in arbitrary dimension, it is sensible to relax this strict condition, e.g. \cite{Gour2017} or by introducing dimension-dependent factors to find analogue relations, as suggested in \cite{CounterExampleMarcus}.
These analogues might not fulfill the \emph{strict} notion of monogamy but capture the main idea of monogamy: That the correlations within one marginal restrict the correlations of all other possible marginals with the former.
With this in mind, we define a correlation monotone which, even though it does not fulfill the \emph{strict} monogamy relation, still complies with useful (quasi)-monogamy relations.
Even though this (quasi)-monogamy relation is weaker than then CKW inequality for qubits, it provides a dimension-dependent inequality for \emph{arbitrary} dimension.
Furthermore we can reproduce a well known result for qubit systems in arbitrary dimension: We show that if two parties with systems of \emph{equal} dimension are maximally entangled, no other party can be entangled with the composite system.
Additionally if two parties are \emph{not maximally} entangled but share some entanglement we can still provide non-trivial bounds on the shared correlations with any other party.

That the introduction of dimensional factors enlarges the applicability of monogamy and entropy relations is in itself a noteworthy observation already made in \cite{CounterExampleMarcus}, in this article we are however able to contribute entropy and (quasi-)\allowbreak monogamy relations.
Let us stress another important point:
Even though both areas of research, entropy inequalities and monogamy relations, are usually considered as two distinct subfields, 
there  exists an intimate connection, which becomes clear once we identify them as particular instances of the marginal problem.
The marginal problem is solved if for a given set of marginals all complementary sets of marginals, which make up a physical state, are identified \cite{klyatchko}.
Since both entropy inequalities as well as monogamy relations depend on functions of the marginals, it is not surprising that a connection between them can be made.

This connection becomes particularly obvious by using the correlation tensor of the generalized Bloch representation \cite{Bloch46,KimuraBloch,KimuraBloch2,BertlmannKrammer}.
Another important advantage of using the Bloch vector representation is the fact, that it naturally introduces dimension-dependent factors, which we show to be useful to circumvent certain no-go-theorems.
Even though the Bloch picture is an extremely well-established subject, it attracts continued interest up to this day. 
Most notably the proof of the long standing non-existence conjecture of the absolutely maximally  entangled (AME) states for seven qubit systems \cite{JensAME1} including a subsequent connection to coding theory and a classification of higher dimensional AME states \cite{JensAME2} but also applications in  entanglement detection \cite{HassanJoag,VicenteHuber,ClaudeMarcusCortensor} and quantum thermodynamics \cite{NicoMartiMarci}. 
In this work we are able to contribute to the Bloch formalism with the introduction of the \emph{Split Bloch basis}, an orthogonal operator basis, which allows a concise description of states with low dimensional support in high-dimensional Hilbert spaces.
The use of this basis is helpful in the handling of multipartite systems with different local dimensions, since it allows to easily define functions which are invariant under isometric transformations.
It is our hope that this article  sparks further discussions over the merits of the generalized Bloch decomposition. 

This article is structured as follows: 
In \cref{sec:corr_ten} we repeat the basic concepts of the generalized Bloch decomposition and the correlation tensor formalism.
Building on these concepts we  introduce the split Bloch basis in \cref{sec:split_bloch}.
The split Bloch basis is an orthogonal operator basis, similar to the generalized Bloch basis.
Using this novel tool we define a correlation monotone, which is convex and invariant under isometry transformations.
Subsequently several (quasi-)monogamy relations are found (\cref{Sec:Monogamy}).
Then we turn our attention towards the second part of the paper: entropy inequalities:
In \cref{Sec:EntropyInequalities} we explain why the Bloch decomposition relates the linear entropy of a system to the correlation tensor of the aforementioned system.
Following up we briefly review the state of the art on entropy inequalities of the R\'eyni and Tsallis type
and go on to introduce a novel linear inequality for the linear (or Tsallis 2-)entropy. 
In \cref{Sec:NonlinEntropyInequalities} we introduce a new quadratic entropy relation for the linear entropy: 
the \emph{generalized} pseudo-additivity. 
It originates from the pseudo-additivity of the q-entropies, but is applicable to \emph{all} states instead of only product states as in the case of pseudo-additivity. 
In the remainder of the section we show that this relation is independent of the known subadditivity \cite{AudenaertQSA} and visualize and discuss the body of allowed entropies for tripartite systems.
Finally, we sum up our results and remark on open questions in \cref{sec:Conclusion}.

\section{Intro: The Correlation Tensor Formalism\label{sec:corr_ten}}
\noindent
In this section we will shortly introduce the concepts and notions  of the traditional (generalized-) Bloch decomposition \cite{Bloch46},\cite{BertlmannKrammer} and the correlation tensor formalism and continue to introduce the \emph{split Bloch basis}, which tackles the problem of expressing low-dimensional quantum states in high-dimensional Hilbert spaces.

Let us start by stating the traditional Bloch decomposition:
\begin{definition}
\label{def:single_bloch}
A single partite (qudit) quantum state $\rho \in \mathcal{H}^d$ can always be written in the (generalized-)Bloch decomposition as defined in \cite{BertlmannKrammer} by:
\begin{align}
   \rho
   &=
   \frac{1}{d}
    \left(
        \mathds{1}_{d}
        + \sum_{i=1}^{d^2-1} 
            \left\langle 
                        \lambda_{i}
            \right\rangle   
            \lambda_{i}
    \right)
    \end{align}
    with $\lambda_{i}$ being orthogonal, traceless  (and canonically assumed hermitian $\lambda_i=\lambda_i^\dagger$) matrices and $d$ being the dimension of the Hilbert space of $\rho$.
    \end{definition}
\begin{remark}
Note that the Hilbert space we consider is a Hilbert-Schmidt space $\mathcal{H}^d = (\mathds{C}^d)^*\otimes\mathds{C}^d$ with the associated scalar product $\left\langle A|B \right\rangle=\Tr\left( B^\ast A\right)$.
\end{remark}
In large dimensions there is some freedom in the choice of $\lambda_{i}$ with the canonical choices being either the generalized Gell-Mann \cite{GGM,BertlmannKrammer} matrices or the (non-hermitian) Heisenberg-Weyl \cite{Weyl27} matrices, alongside more unusual choices like the Heisenberg-Weyl observables \cite{HWO}. 
The single party case can be naturally extended to multi-partite systems by a tensor product construction.
\begin{definition}
\label{def:n_bloch}
Any $n$-partite quantum state $\rho_\Sigma\in \mathcal{H}^{d_\Sigma}=\otimes_{i=1}^n\mathcal{H}^{d_{i}}$ where $\Sigma:=\{\Sigma_1,\Sigma_2,\cdots,\Sigma_n\}$ is the set of all parties can always be represented as
    \begin{align}
    \rho_\Sigma
    &=
    \frac{1}{d_\Sigma}
    \left(
        \sum_{i_1=0}^{d_1^{2}-1} \dots \sum_{i_n=0}^{d_n^{2}-1} 
            \left\langle 
                        \lambda_{i_1}^{\Sigma_1}\otimes\dots \otimes\lambda_{i_n}^{\Sigma_n}
            \right\rangle   
            \lambda_{i_1}^{\Sigma_1}\otimes\dots \otimes\lambda_{i_n}^{\Sigma_n}
    \right)
\end{align}
Here one canonically uses a basis comprised of tensor products of local hermitian Bloch bases, which then allow for a local Bloch vector decomposition.
This means the $\lambda_{i_j}^{\Sigma_j}$ are again orthogonal, e.g. $\Tr\left(\lambda_{m_j}^{\Sigma_j}\lambda_{n_j}^{\Sigma_j}\right)=d_j \delta_{mn}$, 
where $d_j$ is the dimension of the local Hilbert space $\mathcal{H}^{d_j}$; note that $d_\Sigma=\prod_{j=1}^n d_j$.
Furthermore all $\lambda_{i_j}^{\Sigma_j}$ are traceless, except for $\lambda_{0_j}^{\Sigma_j}=\mathds{1}_{d_j}$.
\end{definition}

The advantage is that the Bloch components divide into an intuitive set of correlation tensors.
For example in tripartite systems this leads to three local Bloch vectors, three correlation matrices and a single  three-body correlation tensor.
The Bloch representation furthermore has the advantage of giving an economical form for $\Tr(\rho^{2})$ through the tracelessness of $\lambda_{i}$.

\begin{definition}[Correlation Tensor]
\label{def:corr_ten}
Let $\rho_\Sigma$ be an $n$-partite state:
\begin{enumerate}[(i)]
    \item 
        The correlation tensor of a state $T(\rho_\Sigma)$ is the generalization of the Bloch vector in the qubit Bloch decomposition, it collects all coordinates of the operator basis:
        \begin{align}
            [T(\rho_\Sigma)]_{i_1,i_2,\cdots,i_n}
            &:=
             \left\langle 
                                \lambda_{i_1}^{\Sigma_1}\otimes\dots \otimes\lambda_{i_n}^{\Sigma_n}
                    \right\rangle
                    &
                    1 \leq i_j < d^2_j
\\
   \shortintertext{ \item
        We can now define lower order correlation tensors as:}
        T(\rho_v)
         :=
        T^{v}
        \end{align}
        Where $v\subseteq\Sigma $ runs over all non-empty subsets of $\Sigma$, i.e. its powerset $P(\Sigma)\setminus \emptyset$ and $\rho_v=\Tr_{\bar{v}}\left(\rho_\Sigma\right)$ is the state of the subsystem of $v$ given by taking the partial trace over its complement $\bar{v}$, i.e. $v$ and $\bar{v}$ are a partition of $\Sigma$.
\end{enumerate}
\end{definition}

Let us state now a very useful lemma:
\begin{lemma}
\label{lem:trace_squared}
If $\rho_{\Sigma}$ denotes a n-partite system owned by a set of parties $\Sigma=\{\Sigma_1,\dots,\Sigma_n\}$ with dimension $d_\Sigma=\prod_i d_{i}$,
we can write $\Tr\left(\rho_\Sigma^{2}\right)$ as :
\begin{equation}
\Tr(\rho_\Sigma^{2})=\frac{1}{d_\Sigma}\left(1+\sum_{\forall v \in P(\Sigma)} \left\| T^{v} \right\|^{2} \right).
\end{equation}
Where $v\subseteq\Sigma $ runs over all non-empty subsets of $\Sigma$, e.g. its powerset $P(\Sigma)\setminus \emptyset$. 
$\sum_v\norm{T^{v}}^2$ is the sum over correlation tensors as defined above (\cref{def:corr_ten}).
\end{lemma}
\begin{proof}
The lemma follows directly from the definition of the basis, i.e. since 
$Tr\left(
    (\lambda_{i_1}^{\Sigma_1}\otimes\dots \otimes\lambda_{i_n}^{\Sigma_n})
    (\lambda_{i_1'}^{\Sigma_1}\otimes\dots \otimes\lambda_{i_n'}^{\Sigma_n})\right)=\delta_{i_1,i_1'}\dots\delta_{i_{n},i_n'}d_\Sigma $ 
we find 
\begin{align}
\Tr\left(\rho_\Sigma^2\right)
=\frac{d_\Sigma}{d_\Sigma^2}
\left( 
    \sum_{i_1=0}^{d_1^{2}-1} \dots \sum_{i_n=0}^{d_n^{2}-1} 
        \left\langle 
                \lambda_{i_1}^{\Sigma_1}
                \otimes\dots \otimes
                \lambda_{i_n}^{\Sigma_n}
        \right\rangle^2
\right)
=
\frac{1}{d_\Sigma}
\left( 
    1
    +\sum_v \left\| T^{v} \right\|^{2}
\right)
\end{align}
\end{proof}

\section{The Split Bloch Basis\label{sec:split_bloch}}
\noindent
After having repeated the basic concept of the traditional Bloch decomposition, we proceed now to extend this notion to the \emph{split Bloch decomposition} and motivate its introduction:
Let us consider the scenario where a low dimensional state is represented in a Bloch basis of a Hilbert space with much larger dimension.
In this case the representation is unnecessarily involved.
For example, consider the non-normalized state $\rho=\diag(1,1,0,0)\in  \mathcal{H}^4$, where $\diag$ is the function which maps a vector to a matrix with the vectors entries on the diagonal.
Recall the set of operators, that are typically associated with the generalized Bloch decomposition \cite{BertlmannKrammer}:
$
\mathcal{H}^d=   \Span\left(  \mathds{1}_d,\lambda_{1},...,\lambda_{d^2-1} \right) 
$
where $\mathds{1}_d$ is the identity on the space and $\lambda_{i}$ are required to be traceless and normalized.

The identity is independently of the choice of basis always a part of the set of operators. 
Returning to the example, we see that the use of the identity is sub-optimal, since the additional ones on the diagonal have to be compensated by additional diagonal basis elements.
This simple example already demonstrate a deeper problem concerning canonical Bloch decompositions, i.e. the difficulty to describe \emph{subspaces} of a Hilbert space in a concise way.
This problem becomes more apparent for multi-partite states which have different local dimensions:
We will later consider a cryptographic scenario where Alice and Bob want to communicate securely, meaning they want to exclude a malicious third party (Eve) to extract any information about the content of their communication.
While we can assume the knowledge about the local dimension of Alice and Bob, the local dimension of Eve is arbitrary.
To treat this problem, a separation into different subspaces is very helpful (compare \cref{thm:2-part-less-3-part}).

Now we can ask ourselves: \emph{Is there a Bloch basis which allows a concise description of subspaces? }
The answer is the \emph{split Bloch basis}.
The idea of the basis is to find a complete set of basis elements of a Hilbert space, which can be divided into subsets which span the subspaces.
In the following we will define an orthogonal operator basis similar to the generalized Gellmann matrices.
Recall, that the generalized Gellmann basis for a $d$-dimensional Hilbert space consists of the identity, $d^2-d$ elements with only \emph{off-diagonal} entries and $d-1$ traceless elements with only \emph{diagonal} entries.
The main difference between the split Bloch basis and the Gellmann basis is that we will replace the identity $\mathds{1}$  with two sub-identities $\bar{\mathds{1}} $ and the diagonal elements with a set of $d-2$ diagonal elements.
The off-diagonal elements need not to be replaced since they already divide into two subsets which live in either subspace.

Let us first describe a operation $\bm{\odot}$ which describes the split into different subspaces:
\begin{definition}
Given a Hilbert-Schmidt space $\mathcal{H}^d=\Span\left(\lambda_0,\dots,\lambda_{d^2-1}\right)$ with the customary Hilbert-Schmidt scalar product, i.e. $\left\langle A|B\right\rangle=\Tr(B^{\ast}A)$ can always be divide into subspaces such that:
\begin{align}
    \mathcal{H}^d 
    =\mathcal{H}^c\bm{\odot} \mathcal{H}^{d-c}
    :=&\left(\mathds{C}^c\oplus\mathds{C}^{d-c}\right)^{\ast}\otimes\left(\mathds{C}^c\oplus\mathds{C}^{d-c}\right)
    \\
    =& \;
        \mathcal{H}^c
        \oplus
        \mathcal{H}^{d-c}
        \oplus
         \left(
            \left(\mathds{C}^{d-c}\right)^*
            \otimes
            \mathds{C}^c
        \oplus
            \left(\mathds{C}^{c}\right)^*
            \otimes
            \mathds{C}^{d-c}
        \right)
    \\
    =&\,
        \Span\left(\Lambda_c\right)
        \oplus
        \Span\left(\Lambda_{d-c}\right)
        \oplus
        \Span\left(\Sigma\right)
\end{align}
Where $\Lambda_c$, $\Lambda_{d-c}$ and $\Sigma$ first have to be a partition of $\left\{\lambda_0,\dots,\lambda_{d^2-1}\right\}$ and second $\abs{\Lambda_c}=c^2$, $\abs{\Lambda_{d-c}}=(d-c)^2$  and $\abs{\Sigma}=2dc-2c^2$
Note that the direct sum in the third line is the \emph{internal} direct sum, since $\lambda_i \in \mathcal{H}^d \;\forall i$.
The internal is however isomorphic to the external direct sum. 
\end{definition}

\begin{definition}[The Split Bloch Basis]
Given a Hilbert-Schmidt space $\mathcal{H}^d$ with the customary Hilbert-Schmidt scalar product, i.e. $\left\langle A|B\right\rangle=\Tr(B^{\ast}A)$, 
we want to find a basis $\mathcal{B}_{split}=\{\lambda_i \,|\,\forall i\}\cup\{\omega_j \,|\,\forall j\} \cup\{\nu_k \,|\,\forall k\} $, such that $\mathcal{H}^d=\mathcal{H}^c\bm{\odot} \mathcal{H}^{d-c}$ and
$\lambda_i,\omega_j,\nu_k \in \mathcal{H}^d$ while $\mathcal{H}^c=\Span\left(\{\lambda_i \,|\,\forall i\}\right)$ and $\mathcal{H}^{d-c}=\Span\left(\{\omega_j \,|\,\forall j\}\right)$.
First we will define non-canonical operators, similar to the canonical generalized Gellmann matrices, then we use them to define a split Bloch basis: $\mathcal{B}_{split}$: 
\begin{enumerate}[(i)]
    \item 
Assuming the standard computational basis $\mathds{C}^d=\Span\left(\left|0\right\rangle,\dots,\left|d-1\right\rangle\right)$,
let us define a set of operators, which span the Hilbert space. 
It consist of: 
\\ 
 two sub-identities
\begin{align}
    \lambda_{00}
    &=
    \bar{\mathds{1}}_{c\phantom{-c}}
    = \sum_{l=0}^{c-1} 
        \left|
            l
        \right\rangle
        \left\langle
            l
        \right|
    \\
    \omega_{00}
    &=
    \bar{\mathds{1}}_{d-c}
    = \sum_{l=c}^{d-1} 
        \left|
            l
        \right\rangle
        \left\langle
            l
        \right|\; ,
    \\
    \shortintertext{$d-2$ diagonal elements:}
    \lambda_{kk}
    &=\sqrt{\frac{c}{k+k^2}}\left(
    \sum_{l=0}^{k-1} 
        \left|
            l
        \right\rangle
        \left\langle 
            l 
        \right|
    - 
    k \left|
            k
        \right\rangle
        \left\langle 
            k 
        \right|\right)
    & 1\leq k\leq c-1
    \\
 \omega_{kk}
    &=\sqrt{\frac{d-c}{k+k^2}}\left(
    \sum_{l=0}^{k-1} 
        \left|
            l+c
        \right\rangle
        \left\langle 
            l +c
        \right|
    - 
    k \left|
            k+c
        \right\rangle
        \left\langle 
            k+c
        \right|\right)
    & 1\leq k\leq d-c-1 
    \; ,
    \\
    \shortintertext{$\frac{d(d-1)}{2}$ symmetric off-diagonal elements }
    \lambda_{kl}
    &=   \sqrt{\frac{c}{2}}\left(      
        \left|
            k
        \right\rangle
        \left\langle 
            l 
        \right|
        +
        \left|
            l
        \right\rangle
        \left\langle 
            k
        \right|\right)
        & 0\leq k<l\leq c-1
    \\
    \omega_{kl}
    &=   \sqrt{\frac{d-c}{2}}\left(      
        \left|
            k
        \right\rangle
        \left\langle 
            l 
        \right|
        +
        \left|
            l
        \right\rangle
        \left\langle 
            k
        \right|\right)
         & c-1\leq k<l\leq d-1
    \\
    \nu_{kl}
    &= \sqrt{\frac{d-c}{2}}\left(        
        \left|
            k
        \right\rangle
        \left\langle 
            l 
        \right|
        +
        \left|
            l
        \right\rangle
        \left\langle 
            k
        \right|\right)
         &0\leq k\leq c-1<l\leq d-1
    \\
    \shortintertext{and $\frac{d(d-1)}{2}$ anti-symmetric off-diagonal elements}
        \hat{\lambda}_{kl}
    &=   \sqrt{\frac{c}{2}}\left(-\mathbf{i}      
        \left|
            k
        \right\rangle
        \left\langle 
            l 
        \right|
        + \mathbf{i}
        \left|
            l
        \right\rangle
        \left\langle 
            k
        \right|\right)
        & 0\leq k<l\leq c-1
    \\
        \hat{\omega}_{kl}
    &=   \sqrt{\frac{d-c}{2}}\left(-\mathbf{i}      
        \left|
            k
        \right\rangle
        \left\langle 
            l
        \right|
        + \mathbf{i}
        \left|
            l
        \right\rangle
        \left\langle 
            k
        \right|\right)
        & c-1\leq k<l\leq d-1
    \\
        \hat{\nu}_{kl}
    &=   \sqrt{\frac{d-c}{2}}\left(-\mathbf{i}      
        \left|
            k
        \right\rangle
        \left\langle 
            l
        \right|
        + \mathbf{i}
        \left|
            l
        \right\rangle
        \left\langle 
            k
        \right|\right)
        &0\leq k\leq c-1<l\leq d-1
\end{align}
Note that $\mathbf{i}$ does not stand for an index but rather $\mathbf{i}^2=-1$.
\item 
    With these elements we are able to find a basis $\mathcal{B}_{split}$ for $\mathcal{H}^d$:
    \begin{align}
        \mathcal{B}_{split} 
        &:=
        \left(
            \{\lambda_{kl}\,|\, \forall k,l \}\cup\{\hat{\lambda}_{mn}\,|\,\forall m,n  \}
        \right)
        \cup
         \left(
            \{\omega_{kl}\,|\, \forall k,l\}\cup\{\hat{\omega}_{mn} \,|\,\forall m,n \}
        \vphantom{\hat{\lambda}_{kl}}
        \right)
        \cup
         \left(
            \{\nu_{kl}\,|\, \forall k,l\}\cup\{\hat{\nu}_{mn} \,|\,\forall m,n \}
        \vphantom{\hat{\lambda}_{kl}}
        \right)
        \\
        &:= \{\lambda_i^\prime \,|\, 0 \leq i \leq c^2-1\} 
        \cup
        \{\omega_i^\prime \,|\, 0 \leq i \leq  (d-c)^2-1\}
        \cup 
        \{\nu_i^\prime \,|\, 0 \leq i \leq  2dc-2c^2-1\}
    \end{align}
    In the second line of the above equation we simplified the notation by concatenating the indices, we only require $\lambda^\prime_0=\lambda_{00}$ and $\omega^\prime_0=\omega_{00}$ all other indices may be assigned freely. 
    Note that the basis is split naturally into three parts. 
    Two of them span two Hilbert-Schmidt subspaces, i.e.
    $\mathcal{H}^{c}
        =
           \Span\left( \{\lambda^\prime_i \,|\, \forall i\}\right)
    $ 
    and
    $\mathcal{H}^{d-c}
        = 
          \Span\left(  \{\omega^\prime_i \,|\, \forall i \}\right)
        $.
    The third part spans a subspace which is however \emph{not} a Hilbert-Schmidt space since \emph{in general} it can not be written as a tensor product of $\mathds{C}^d$ with its dual, i.e. $\Span(\nu_i^\prime \,|\, 0 \leq i \leq  2dc-2c^2-1) \neq (\mathds{C}^d)^\ast \otimes \mathds{C}^d \;\forall d,c $.
    From now on we will use 
    ${\{\mu^i \,|\, 0 \leq i \leq  d^2-c^2-1\}}
    =  \{\omega_i^\prime \,|\, 0 \leq i \leq  (d-c)^2-1\}
        \cup 
        \{\nu_i^\prime \,|\, 0 \leq i \leq  2dc-2c^2-1\} $ with $\mu_0=\omega_o=\bar{\mathds{1}}_{d-c}$ and take $\lambda'_i \rightarrow \lambda_i \;\forall i$ for simplicity of notation.
    We will however mention when we use the split Bloch basis although it can be understood from context. 
    In fact, whenever $\mu_i$'s appear it should be understood that we are using the split Bloch basis.
    \end{enumerate}
\end{definition}
Note that we have chosen the normalization of the subspaces differently, i.e. $\Tr\left(\lambda_{i}\lambda_{j}\right)=\delta_{ij}\,c $ and $\Tr\left(\mu_{i}\mu_{j}\right)=\delta_{ij}\left(d-c\right)$. This ensures that the normalization of the subspace $\mathcal{H}^c$ is independent of the total dimension.

It is still necessary to show that we indeed constructed a orthogonal operator basis, e.g. that
all elements of the basis are orthogonal $\left\langle \lambda_i \right|\left.\lambda_{j\neq i}\right\rangle=0$, $\left\langle \mu_i \right|\left.\mu_{j\neq i}\right\rangle=0$, $\left\langle \lambda_i \right|\left.\mu_{j}\right\rangle=0$ and that they actually span $\mathcal{H}^d$.
This is however easy to see: 
    Any scalar product of pairs of operators with only off-diagonal elements is trivially zero, the same is true for any scalar product of off-diagonal and diagonal operators.
    Furthermore, any scalar product of diagonal operators in the separate subspaces is zero due to the tracelessness of all operators, except the (sub-)identity.
    Finally, any scalar product of any diagonal operators of the different subspaces is zero, due to the fact that they are non-overlapping block matrices.
    The fact that we have $d^2$ orthogonal elements already suffice to span $\mathcal{H}^d$, i.e. $c$ diagonal elements in $\mathcal{H}^c$, $d-c$ diagonal elements in $\mathcal{H}^{d-c}$ add up to the regular $d$ diagonal elements of the canonical Bloch basis of $\mathcal{H}^d$ and the off-diagonal are, apart from normalization, the same.
    
We now see that the split Bloch basis allows for a concise representation states with low-dimensional support in a $d$-dimensional system. For example, take $|\psi\rangle=\frac{1}{\sqrt{2}}(|0\rangle+|1\rangle)\in\mathds{C}^d$. This state has a split Bloch representation of $\langle\lambda_{01}\rangle=1$ being the only traceless Bloch matrix with  non-zero expectation value. The canonical Bloch representation would have to compensate the first entry $\lambda_0$ by additional diagonal Bloch matrices. In addition, keeping the squared trace constant it would also imply an increasing value of $\langle\lambda_{01}\rangle=\sqrt{2d}$ using the above normalization.

After having defined the split Bloch basis, we can repeat the definition of the correlation tensor (\cref{def:corr_ten}) and the lemma connecting the trace of the squared state with the Euclidean norm of the correlation tensor (\cref{lem:trace_squared}). 
First we describe only bipartite states in the split Bloch basis, the generalization is however straight-forward (see \cref{re:generalization}).
\begin{remark}
Given a bipartite state $\rho_{AB}\in \mathcal{H}^{d_A d_B}$ where the parts have \emph{different} local dimension; w.l.o.g. $d_A<d_B$.
We will use the split basis $\mathcal{B}_{split}$ to divide the bigger Hilbert space 
$\mathcal{H}^{d_{B}}=\mathcal{H}^{d_A}\odot \mathcal{H}^{d_B-d_A}$, such that the local dimension of one of the Hilbert-Schmidt subspaces coincides with the local dimension of the smaller space, i.e. $\mathcal{H}^{d_A d_B}= \mathcal{H}^{d_A}\otimes \left(\mathcal{H}^{d_A}\odot \mathcal{H}^{d_B-d_A }\right) $.
\begin{equation}
\rho_{AB }
=
    \sum_{i=0}^{d_A^{2}-1}
        \left(
            \frac{1}{d_A^{2}}
            \left(
                \sum_{j=0}^{d_A^{2}-1}
                    \left\langle 
                        \lambda^A_{i}\otimes\lambda^B_j\right
                    \rangle  
                    \lambda^A_{i}\otimes\lambda^B_j
            \right)
            +\frac{1}{\left(d_B-d_A\right)d_A}
            \left(
                \sum_{j=0}^{d_B^{2}-d_A^2-1}
                \left\langle 
                    \lambda^A_{i}\otimes\mu^B_j\right
                \rangle  
                \lambda^A_{i}\otimes\mu^B_j
            \right)
        \right)
        \label{obs:split_state}
\end{equation}
Note that only $\lambda^{A}_0$, $\lambda^{B}_0$ and $\mu^B_{0}$ have finite trace, all other elements are traceless.
Any state is properly normalized, i.e. $\Tr\left(\rho_{AB}\right)=1$.
Since the only elements with a finite trace are $\mathds{1}_{d_A}\otimes\bar{\mathds{1}}_{d_A}$ and $\mathds{1}_{d_A}\otimes\bar{\mathds{1}}_{d_B-d_A}$ we find:
\begin{align}
    \Tr\left(\rho_{AB}\right) 
    &=
    \frac{\left\langle
    		\mathds{1}_{d_A}
    		\otimes
    		\bar{\mathds{1}}_{d_A}
        	\right\rangle}
        {d_A^2}
    \Tr\left(
        \mathds{1}_{d_A}
        \otimes
        \bar{\mathds{1}}_{d_A}
        \right)
    +   \frac{\left\langle
        \mathds{1}_{d_A}
        \otimes \bar{\mathds{1}}_{d_B-d_A}
        \right\rangle}
        {\left(d_B-d_A\right)d_A}
    \Tr\left(
        \mathds{1}_{d_A}
        \otimes \bar{\mathds{1}}_{d_B-d_A}
        \right)
    \\
    &=
    \left\langle
        \mathds{1}_{d_A}
        \otimes
        \bar{\mathds{1}}_{d_A}
    \right\rangle 
    +
    \left\langle
        \mathds{1}_{d_A}
        \otimes
        \bar{\mathds{1}}_{d_B-d_A}
    \right\rangle
    =
    1
\end{align}
Note that in the standard Bloch decomposition the Bloch coefficient of the identity is fixed to one by normalization of the density matrix, in our case \emph{only the sum} of the Bloch coefficients of the partial identities adds up to one.
\end{remark}

Now let us investigate the correlation tensor formalism in the split Bloch basis; again we will use a bipartite state as an example:
\begin{definition}
\label{def:split_corr}
Given a bipartite state $\rho_{AB}\in \mathcal{H}^{d_A}\otimes \mathcal{H}^{d_B}$ with $d_B>d_A$ expressed in the split Bloch basis.
Now we split the correlation tensor  into parts :
        \begin{align}
            \left[T_{SD}(\rho_{AB})\vphantom{\tilde{T}}\right]_{ij} 
                &:=\left\langle 
                        \lambda_{i}^{A}\otimes\lambda_{j}^{B}
                  \right\rangle & \, i,j\in \{1,\dots,d_A^2-1\} 
            \\
            \left[\tilde{T}(\rho_{AB})\right]_{ij} 
                &:=\left\langle 
                        \lambda_{i}^{A}\otimes\mu_{j}^{B}
                  \right\rangle 
                & \, i\in \{1,\dots,d_A^2-1\} 
                \,;\,    j\in \{1,\dots,d_B^2-d_A^2-1\}
        \end{align}
$T^{SD}\left(\rho_{AB}\right)$ describes the correlations in $\mathcal{H}^{d_A}\otimes\mathcal{H}^{d_A} $ and $\tilde{T}\left(\rho_{AB}\right)$ describes the remaining correlations in $\mathcal{H}^{d_A}\otimes\mathcal{H}^{d_B}$.
The lower order correlation tensors are, in this case, just the local Bloch vectors.
\end{definition}
Finally we can find an expression for \cref{lem:trace_squared} in the split Bloch basis:

\begin{remark}
We can write the trace of a squared bipartite state  $\rho_{AB}\in \mathcal{H}^{d_A}\otimes\mathcal{H}^{d_B}$ of different local dimension, e.g. $d_{A}< d_B$:
\begin{align}
\Tr\left(\rho_{AB}^{2}\right)
=\frac{1}{d_A^2}
    \left(
        \left\langle
        \mathds{1}_{d_A}
        \otimes \bar{\mathds{1}}_{d_A}
        \right\rangle
        +\sum_{v \in P(\{A,B\})} \left\| T^{v}_{SD} \right\|^{2} 
    \right)
+\frac{1}{\left(d_B-d_A\right)d_A}
    \left(
        \left\langle
        \mathds{1}_{d_A}
        \otimes \bar{\mathds{1}}_{d_B-d_A}
        \right\rangle
        +\sum_{v \in \{B,AB\}} \norm{ \tilde{T}^{v}}^{2} 
        \right)
\end{align}
\end{remark}
As we have seen the split Bloch basis preserves the advantages of the Bloch basis.
However, it has additional advantages one of which is the possibility to express the (operator) Schmidt basis in terms of the split Bloch basis:
\begin{lemma}[The Schmidt decomposition]
\label{lem:Schmidt_decom}
For pure states $\left|\psi\right\rangle_{AB}\in \mathds{C}^{d_A}\otimes\left(\mathds{C}^{d_A}\oplus\mathds{C}^{d_B-d_A}\right)$ with $d_B>d_A$, there exists \emph{some} split Bloch basis such that we find the Schmidt decomposition as:
\begin{align}
 \left(
    \left|\psi\right\rangle
    \left\langle\psi\right|
 \right)_{AB}
    =  
    \sum_{i,j=0}^{\mathclap{\min(d_A,d_B)-1}}
        \left\langle
            \lambda^A_i
                \otimes
            \lambda^B_j 
        \right\rangle
        \lambda^A_i
        \otimes
        \lambda^B_j
\end{align}
where $\left\langle\lambda^A_0\otimes\lambda^B_0\right\rangle=1$. 
The \labelcref{pr:Schmidt_decom} of this lemma can be found in the appendix.
\end{lemma}
Let us finish this chapter with a short note on the $n$-partite generalization:

\begin{remark}[Generalization]
\label{re:generalization}
Until now, we have only considered bipartite states.
However, this choice is only due to the fact that the notation becomes very cumbersome in the $n$-partite scenario.
The generalization to multipartite states is straightforward:
Formally we can just repeat the step from \cref{def:single_bloch} to \cref{def:n_bloch}, i.e. replacing the Bloch basis by tensor products of \emph{local} Bloch bases, which can be split or not.
Note that the powerset construction we have chosen allows to scale the definitions concerning the split Bloch basis to $n$-partite systems easily by simply using the set of the $n$-parties.
\end{remark}

\section{Monogamy of Correlations from the Bloch Picture\label{Sec:Monogamy}}
\noindent
We proceed in this section by utilizing the correlation tensor formalism to derive some quasi-monogamy relations for (quantum-)correlations. 
The key idea is to focus on the relevant parts of the correlation tensor and drop all non-essential terms. 
This choice of a subset of the correlation tensor will lead to a dimension-dependent quasi-monogamy relation. 
In order to capture the correlations between two subsystems $A,B$ of a bipartite system  we define the following quantity to measure these correlations:
\begin{definition}[Correlation Monotone]
\label{def:mono}
Let $\rho_{AB}$ be an arbitrary bi-partite system and w.l.o.g. $d_A\leq d_B$, we define a correlation monotone $\mathcal{T}_{A|B}\left(\rho_{AB}\right)$:
        \begin{align}
        \mathcal{T}_{A|B}(\rho_{AB}) 
        :=
        \frac{1}{g_{A|B}} 
        \max_{\{\lambda_j^B\}}
        \left( \sum_{i=1}^{d_{min}^{2}-1}\sum_{j=1}^{d_{min}^{2}-1} 
                    \left\langle 
                            \lambda_{i}^{A}\otimes\lambda_{j}^{B}
                    \right\rangle^2 
        \right)
        \end{align}
Where $d_{min}=\min\left(d_A,d_B\right)$ is the minimum of the dimensions of the subsystems $A$ and $B$.
Furthermore, $\Tr\left(\lambda_{i}^{A}\lambda_{j}^{A}\right)=\Tr\left(\lambda_{i}^{B}\lambda_{j}^{B}\right)=\delta_{ij} \, d_{min}$.
Note that the maximization is obsolete if the local dimension of $\rho_A$ and $\rho_B$ is the same.
If the local dimensions are different the maximization is over arbitrary basis changes on the bigger subspace.
If $d_B> d_A$ then $\rho_B \in \mathcal{H}^{d_A}\odot \mathcal{H}^{d_B-d_A} $ will be described by the split Bloch basis such that one part of the split matches the local dimension of $\rho_A$.
The optimization is such that the correlations are maximal between the same dimensional spaces of $\rho_A$ and $\rho_B$. 
\end{definition}
$\mathcal{T}_{A|B}$ should be read as ``the correlations between the systems $A$ and $B$''. 
This monotone has a natural operational interpretation: For any given complete set of observables on $A$, it is the maximum amount of correlations achievable over any $d_A$-dimensional subspace of $B$. 
Due to the symmetric nature of the Schmidt decomposition, for pure states this coincides with the squared 2-norm of the correlation tensor, i.e. $\mathcal{T}_{A|B}\left(|\psi_{AB}\rangle\langle\psi_{AB}|\right)=\frac{1}{g_{A|B}}||T^{AB}||^2$.
The particular choice of normalization is left unspecified by intention. We will write in the following $g_{A|B}$ without further specification.
In principle the exact choice is a matter of taste.

Let us now state a useful lemma, which gathers some important properties of our monotone:
\begin{lemma}
\label{lem:traceoutlemma}
For a  general (possibly mixed) state $\rho_{A B E}$ be the following relations for the correlation monotone are true: 
\begin{enumerate}[label=(\roman*)]
    \item
    \begin{align}
       \mathcal{T}_{A|B} &= \mathcal{T}_{B|A}
       \\
       \shortintertext{\item and: }
    \frac{g_{A|B}}{g_{A|BE}} \mathcal{T}_{A|B}(\rho_{ABE}) &\leq \mathcal{T}_{A|BE}(\rho_{ABE})
 \end{align}
 \end{enumerate}
\end{lemma}

\begin{proof}
\begin{enumerate}[label={\text{add} (\roman*):}, wide, labelwidth=!, labelindent=0pt]
\item 
    Follows directly from the definition.
\item
    Simply using the fact that $\norm{T^{AB}}^2\leq \norm{T^{AB}}^2+\norm{T^{AE}}^2 + \norm{T^{ABE}}^2$ together with the definition of $\mathcal{T}$ recovers the above equation.
\end{enumerate}
\end{proof}
The main result of this section is the following theorem that demonstrates how correlations within a tripartite system are constrained by a quasi-monogamy relation. 
\begin{theorem}
\label{thm:2-part-less-3-part}
Let $\rho_{ABE} \in \mathcal{H}=\mathcal{H}^{d^2}\otimes\mathcal{H}^{d_E}$ be an arbitrary tripartite state owned by $A,B$ and $E$ with local dimensions $d_A=d_B=d$ and $d_E$:
\begin{enumerate}[(i):]
    \item  The correlation of a composite system $\rho_{AB}$  with an arbitrary system $\rho_E$ limits the correlation of its marginals with the same:
        \begin{align}
        \mathcal{T}_{A|E}\left(\rho_{ABE}\right) 
        + \mathcal{T}_{B|E}\left(\rho_{ABE}\right) 
        &\leq \frac{g_{AB|E}}{\min\left(g_{A|E},g_{B|E}\right)} \mathcal{T}_{AB|E}\left(\rho_{ABE}\right)   \label{eq:2-part-less-3-part}
        \\
        \shortintertext{\item The correlation of any state $\rho_{AB} \in \mathcal{H}^{d^2}$ with an arbitrary state $\rho_E$ is restricted by:
        }
            \mathcal{T}_{AB|E}\left(\rho_{ABE}\right)
            &\leq
            \frac{
                d^{4}  
                -1
                -2
                \left(
                    \left\|T^{A}\right\|^2+\left\|T^{B}\right\|^2
                \right)
                -2 g_{A|B} \mathcal{T}_{A|B}
            }
            {g_{AB|E}}
        \end{align}

\end{enumerate}
The \labelcref{pr:2-part-less-3-part} of this theorem can be found in the appendix.
\end{theorem}

We point out that \cref{thm:2-part-less-3-part} (i) is reminiscent of the Coffman-Kundu-Wootters inequality \cite{CKW}.
It was the first example of a monogamy relation.
For three qubits the CKW inequality is used almost synonymous with the term monogamy. 
The form of \cref{thm:2-part-less-3-part} (i) is the same, except for a dimensional correction factor, hidden in the normalization, on the right hand side and our quantity $\mathcal{T}_{\Omega|\Sigma}$ replacing the concurrence. 
Even though this result was inspired by the notion of monogamy of entanglement (\cite{Terhal,CKW,OsborneVerstraete,Def:SquashedEntanglement,KoashiWinter}) we have to differentiate from these well known relations. 

First, our results primarily discuss \emph{correlations} instead of entanglement. 
Our quantity  $\mathcal{T}_{\Omega|\Sigma}$ can easily be rewritten into an entanglement \emph{monotone} by subtracting the maximal value allowed by a separable state, however we have not proven the LOCC non-increasingness of $\mathcal{T}_{\Omega|\Sigma}$. 
Therefore we can not assume it to be an entanglement \emph{measure}, unlike the concurrence \cite{Def:Concurrence} used in Coffman-Kundu-Wootters \cite{CKW}.

Second, the functional form of our relation is different since it involves dimension-dependent constants. 
We consider this to be acceptable, since the most well known monogamy inequalities are also inherently dimension-dependent. 
They may not contain explicit dimension factors, but in fact they are dimension-dependent since they only hold for the qubit case. 
Already in the case of the slightly larger qutrit system a counter example is known \cite{CounterExampleOu}, that can be extended to the case of arbitrary dimension \cite{CounterExampleMarcus}. 
Our dimension-dependent factors, hidden in the normalization constant $g_{\ast|\ast}$, in turn allow our relations to hold in any dimension.

We call \cref{thm:2-part-less-3-part} (ii) a quasi-monogamy relation, since it captures the fundamental idea which monogamy relations aim to describe: 
That the correlations of Alice and Bob limit the amount of possible shared correlations of Eve with either of the two remaining systems.
In fact, \cref{thm:2-part-less-3-part} (ii) evaluated for maximally entangled Alice and Bob turns out to coincide with the upper bound of $\mathcal{T}^{sep}$ for product states.
Thus generalising a known result for the CKW inequality, that two maximally entangled parties with \emph{equal} dimension can not be entangled with any other party.
We point out that our result has a functional form in \emph{arbitrary} dimension.

But even if Alice and Bob are not maximally entangled, the bound (\cref{thm:2-part-less-3-part} (ii)) is still limiting the correlations Eve can have with the system.
We assume that Alice and Bob have the same local dimension, e.g. $d_A=d_B=d$, normalize by the bound for separable states $g_{AB|E}=\mathcal{T}_{AB|E}^{sep}$ with  $\mathcal{T}_{AB|E}^{sep}=\left(d^2-1\right)\left(d_E-1\right)\geq \mathcal{T}_{AB|E}\left(\rho_{AB}\otimes\rho_{E}\right)$  and define 
$\mathcal{T}_{AB|E}^{excess}:=d\left(\mathcal{T}_{AB|E}-1\right)$.
$\mathcal{T}_{AB|E}^{excess}$ measures how much entanglement Eve shares with the composite system of Alice and Bob scaled by the local dimension $d$.
Note that we estimate entanglement rather than correlations due to the fact that we subtract the bound for separable states 
$1=\mathcal{T}_{AB|E}^{sep}/ g_{AB|E}$.
The scaling is only used to make $\mathcal{T}_{AB|E}^{excess}$ comparable for different dimensions.

Considering the worst case scenario, that Eve holds a purification to Alice and Bob's system (see the proof of \cref{thm:2-part-less-3-part} in the Appendix), i.e. $d_E=d^2$, we plot (\cref{fig:restricted_entanglement})  the maximal value of $\mathcal{T}_{AB|E}^{excess}$ for $0\leq\mathcal{T}_{A|B}\leq 1$ for different dimensions.
Note that the range of $\mathcal{T}_{A|B}$ corresponds to a normalization of $g_{A|B}=d^2-1$.
We consider two different scenarios, i.e. the local correlation tensors $\norm{T^{A}}^{2}$ and $\norm{T^{B}}^{2}$ are either minimal or maximal.
If they are minimal (\cref{fig:resticted_entanglement_worst}), either due to ignorance of the local states or because the state is proportional to the identity, we see that only maximal entanglement between Alice and Bob, i.e. $\mathcal{T}_{A|B}=1$  excludes entanglement with Eve.
This is however true independent of the local dimension $d$ of Alice's and Bob's system.
The other scenario is also interesting (\cref{fig:resticted_entanglement_best}): 
If it is possible to certify that $\norm{T^{A}}^{2}+\norm{T^{B}}^{2}$ is maximal it suffices that Alice and Bob share some correlations  to exclude any entanglement with Eve.
This is however highly dependent on the local dimension $d$:
Unsurprisingly the lower the local dimension the greater the advantage of knowing the correlation tensor norms of the local systems.
In fact we can see that the advantage for $d=100$ is minute even for maximal local correlation tensor norms. 
The factors influencing the bound $\norm{T^{A}}^{2}$, $\norm{T^{B}}^{2}$ and $\mathcal{T}_{A|B}$ are however related, i.e.
$0\leq\norm{T^{A}}^{2}+\norm{T^{B}}^{2}\leq\min\left( 2d -2,\left(d^{2}-1\right)\left(1- \mathcal{T}_{A|B}\right)\right) $.
For the proof of these bounds consult the appendix (\cref{lem:lower-tensor_bound}).

Let us finish this section again with a short note on the $n$-partite generalization:
\begin{remark}
Note that this correlation monotone can be used to describe multi-partite systems as well:
Let $\rho_{A_1,\dots,A_n}$ be a multi-partite system with $\Omega \subset \left\{A_{1} \ldots A_{n}\right\}, \Sigma \subset \left\{A_{1} \ldots A_{n}\right\}$ and $\Omega \cap \Sigma \neq \emptyset$.
If $\Sigma,\Omega$ are composite systems, i.e. it is possible to describe them as tensor product of local Bloch bases, then we want to consider the correlations of all possible subsystems of $\Sigma$ with all subsystems of $\Omega$.
Therefore, if the index is replaced with a multi-index $\lambda_i^\Sigma\rightarrow \lambda^\Sigma_{i_1,\dots,i_{\abs{\Sigma}}}$ and $\lambda_j^\Omega\rightarrow \lambda^\Omega_{j_1,\dots,j_{\abs{\Omega}}}$, we need the cross sum over the new indices in the parts both to be greater then one: $\sum_k i_k \geq 1 \wedge \sum_l j_l \geq 1$.
In the language of correlation tensors this is of course equal to the sum of the squared correlation tensor norms $\sum_{v \in \Gamma}\norm{T^v}^2$ with $\Gamma:=P\left(\Omega \cup \Sigma\right)\setminus \left(P\left(\Omega\right)\cup P\left(\Sigma\right) \right)$,
since this set contains exactly the correlation tensor elements that have at least one element in each of the two subsets $\Omega$ and $\Sigma$.
Finally, note that if the subsystems of the bi-partition have the \emph{same} dimension $d_\Sigma=d_\Omega$ the correlation tensors are defined as \cref{def:corr_ten} and the maximization is obsolete,
if they have \emph{different} dimension $d_\Sigma\neq d_\Omega$ the split Bloch basis will be used and the sum will be over $\norm{T_{SD}^v}^2$ as defined in \cref{def:split_corr}.
The maximization can be found by arbitrary basis change in the bigger space. 
\end{remark}

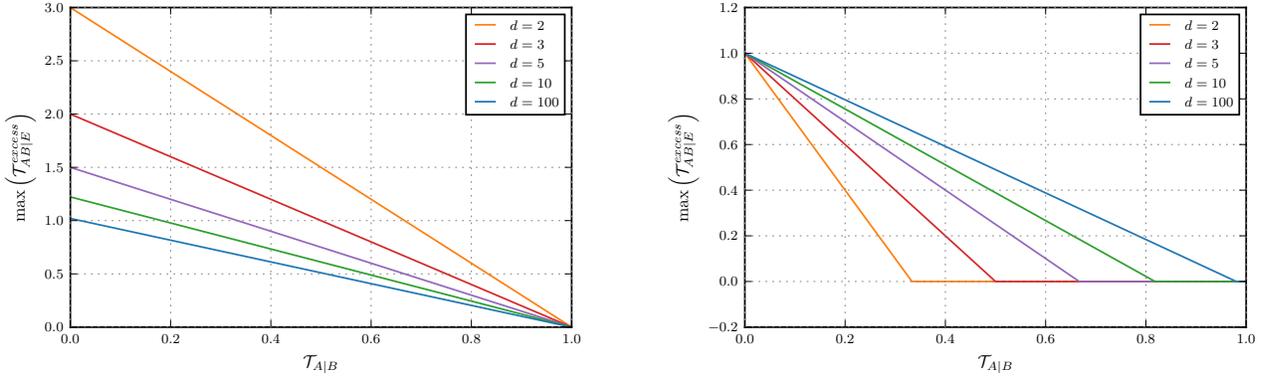
\begin{figure*}[tbp]
\begin{subfigure}[b]{0.49\textwidth}
\centering
\scalebox{0.65}{\input{worst_case.pgf}}
\caption{Locally maximally mixed states, i.e. only correlations contributing to the purity $\norm{T^{A}}^{2}+\norm{T^{B}}^{2}=0$.
\label{fig:resticted_entanglement_worst}}
\end{subfigure}
\begin{subfigure}[b]{0.49\textwidth}
\centering
\scalebox{0.65}{\input{best_case.pgf}}
\caption{
Local Bloch vectors at maximal possible local norm, i.e. $\norm{T^{A}}^{2}+\norm{T^{B}}^{2}=$ $\min\left( 2d -2,\left(d^{2}-1\right)\left(1- \mathcal{T}_{A|B}\right)\right) $ \label{fig:resticted_entanglement_best}
}
\end{subfigure}
\caption{ The (scaled) shared entanglement of Eve with Alice and Bob $\mathcal{T}_{AB|E}^{excess}=
d\left(
\mathcal{T}_{AB|E}
-1\right)$ is restricted by the shared correlations between the latter $\mathcal{T}_{A|B}$ even if Eve holds a purification for the composite system of Alice and Bob, e.g. $d_E=d^2$. 
Note that $\mathcal{T}_{AB|E}^{excess}$ is normalized by the bound of separable states $g_{AB|E}=\left(d^2-1\right)\left(d_E-1\right)$ and weighed with $d$. \label{fig:restricted_entanglement}}
\end{figure*}

\section{Linear Entropy Inequalities from the Bloch Picture}
\label{Sec:EntropyInequalities}
\noindent
In this section we will demonstrate that correlation tensor norm constraints (often also called Bloch vector length constraints) can be reinterpreted as entropy inequalities.
This equivalence is best seen in the  Bloch picture:
The natural relation of $\Tr(\rho^{2})$ to the correlation tensors of the Bloch representation allows us to rewrite existing relations between marginals in the Bloch picture into entropy inequalities of suitable members of several parametrized entropic families. 
 The two most famous of these are the additive  R\'enyi $\alpha$-family \cite{RenyiEntropy}, given by 
\begin{align}
\mathcal{S}_{\alpha}(\rho)&:=\frac{1}{1-\alpha}\log\Tr(\rho^{\alpha}),
\\
\shortintertext{
 alongside the non-additive Tsallis $q$-family \cite{TsallisEntropy} defined by}
 \mathcal{S}^{q}(\rho)&:=\frac{1}{q-1}\left( 1-\Tr(\rho^{q}) \right).
 \end{align}
 Both families retrieve the von Neumann entropy in the case of $\alpha/q\rightarrow1$. 
In the case of $\alpha=q=2$ we can use the simple representation of $\Tr\left(\rho^2\right)$ in form of correlation tensor norms (compare \cref{lem:trace_squared}) to find simple forms of these entropies:
 \begin{remark}
Given $n$-partite quantum state $\rho_\Sigma\in \mathcal{H}^{d_\Sigma}=\otimes_{i=1}^n\mathcal{H}^{d_{i}}$ where $\Sigma:=\{\Sigma_1,\Sigma_2,\cdots,\Sigma_n\}$ is the set of all parties we find 
the Tsallis $2$ or linear entropy as:
    \begin{align}
  \label{Def:SL}
\mathcal{S}_{L}(\rho_{\Sigma})
    &:=
    \left(1-\Tr(\rho^{2}_\Sigma)\right)
    \\
    &=
    1
    -
    \frac{1}{d_\Sigma}
    \left(
        1
        +\sum_{\forall v \in P(\Sigma)} \left\| T^{v} \right\|^{2} 
    \right)
    \end{align}
Note that we simply used \cref{lem:trace_squared} to replace $\Tr(\rho^{2}_\Sigma)$.
 \end{remark}
 
 Entropy inequalities are an extensively studied subject already in the context of classical information theory, where entropy relations are readily available \cite{YeungCone}.
 The first attempts to recreate the classical results in quantum information  have been achieved by Lieb and Ruskai \cite{LiebRuskai}, who showed the famous strong subadditivity relation (SSA)
 \begin{equation}
 \mathcal{S}_{1}(\rho_{ABC})+\mathcal{S}_{1}(\rho_{C}) \leq \mathcal{S}_{1}(\rho_{AC}) + \mathcal{S}_{1}(\rho_{BC})\,,
 \end{equation}
for the von Neumann entropy. The hope to extend this result to the remaining parameter space of any of these two families has been impaired by broadly applicable no-go results. 
It is known that the Tsallis $q$-familiy can not satisfy SSA except for $q=1$ \cite{FuruichiTsallisSSA}. 
In contrast the weaker notion of subadditivity holds for quantum Tsallis entropies with $q \geq 1$ \cite{AudenaertQSA}. 
Similarly, the case of the  R\'enyi entropy has been mostly settled by \cite{MilanRenyi}, who demonstrated that for $\alpha \in (0,1)\cup(1,\infty)$ SSA can not hold. 
The results of \cite{MilanRenyi} go far beyond SSA.
They limit the parameter space where SSA or even a similar relation may hold severely. 
In the interval $(0,1)$ non-negativity is the only possible relation and in $(1,\infty)$ no homogeneous, thus linear, entropy inequality may exist. For the R\'enyi $0$-entropy linear inequalities do exist \cite{rankineq}, however SSA is not among them. 
As a reaction to this setback many authors consider alternative or generalized versions of SSA \cite{PetzVirosztek}. 
In a similar spirit the Bloch picture allows us to construct relations that resemble dimension-dependent versions of SSA.

\begin{theorem}
For a tripartite quantum system  $\rho_{ABC}$ we find the following entropy  inequality for the linear entropy $\mathcal{S}_{L}\left(\rho_{ABC}\right)=1- \Tr\left(\rho_{ABC}^{2}\right)$:
\begin{align}
\mathcal{S}_{L}\left(\rho_{ABC}\right)
 + \frac{1}{d_Ad_B} \mathcal{S}_{L}\left(\rho_{C}\right) 
\leq
\frac{1}{d_B}\mathcal{S}_{L}\left(\rho_{AC}\right)
+ \frac{1}{d_A} \mathcal{S}_{L}\left(\rho_{BC}\right)
+\frac{d_A d_B+1-d_A-d_B}{d_A d_B} \label{eq:dim_sub_add}
\end{align}
The \labelcref{pr:dim_sub_add} of the theorem can be found in the appendix.
\label{thm:dim_sub_add}
\end{theorem}

This theorem can be considered as providing a substitute for SSA in the $q=2$ case, where SSA is known not to hold \cite{FuruichiTsallisSSA}. Alternatively, we could have chosen to express $\Tr(\rho^{2})$ in terms of $2^{-S_{2}}$ instead of $S_{L}$. This formulation leads to a non-linear entropy inequality for the  R\'enyi entropy with $\alpha=2$. We point out that such a reformulation is not in conflict with the established no-go result about homogeneous  R\'enyi $\alpha$-entropy relations \cite{MilanRenyi} for $\alpha \in (1,\infty)$, since $2^{-S_{2}}$ is non-linear.
The proof of \cref{thm:dim_sub_add} in the appendix however employs the very same techniques as the proofs of the monogamy relations in \cref{Sec:Monogamy}.
Generally speaking all proofs are inspired by the Bloch formalism. 
Essentially, we set some elements of the correlation tensor to zero. 
The resulting statement depends on the choice which correlation tensors are kept and which are discarded.
Interestingly, these very different types of results seem to be complementary to each other in the Bloch picture. 
The correlation tensor formalism allows us to harness tools like purity and combine it with the Schmidt decomposition in a straightforward way. 
Additionally, once we know that a certain state does not exhibit correlations in a particular sector of the correlation tensor, we can develop tailor-made entropy relations that are expected to be tight with respect to the chosen state. 

Before closing this section, we need to asses the utility of our newly derived linear entropy inequalities. 
It is natural to ask how they compare to other well-known entropy inequalities. 
It was possible to classify all possible classical entropy inequalities by a well known result, that showed that all classical entropy inequalities are representable as a convex cone of properly chosen elementary entropy vectors \cite{YeungCone}.
Similar arguments were used to study the von Neumann entropy \cite{PippengerCone}.
In contrast, benchmarking q-entropy inequalities is not a straightforward endeavour, since to the best of our knowledge no comparable complete classifications exist.
Audenaert's (weak) subadditivity of q-entropies \cite{AudenaertQSA} seems to be one of the natural competitors
\begin{equation}
\label{eq:AudenaertQSA}
\mathcal{S}^{q}(\rho_{AB})\leq \mathcal{S}^{q}(\rho_{A})+\mathcal{S}^{q}(\rho_{B}).
\end{equation}
In the following we want to demonstrate that \cref{thm:dim_sub_add,eq:dim_sub_add} does not follow trivially from the above \cref{eq:AudenaertQSA}. Since \cref{eq:dim_sub_add} involves tripartite terms and \cref{eq:AudenaertQSA} does not, our first step is to pad Audenaert's subadditivity by an extra system and set $q=2$ to make both relations comparable, thus we rewrite both into
\begin{align}
\label{eq:AudenaertQSA2}
\mathcal{S}_{L}(\rho_{ABC})&\leq \mathcal{S}_{L}(\rho_{AC})+\mathcal{S}_{L}(\rho_{B}) \\
\mathcal{S}_{L}\left(\rho_{ABC}\right) &\leq
\frac{1}{d_B}\mathcal{S}_{L}\left(\rho_{AC}\right)
+ \frac{1}{d_A} \mathcal{S}_{L}\left(\rho_{BC}\right)
- \frac{1}{d_Ad_B} \mathcal{S}_{L}\left(\rho_{C}\right)
+\frac{d_A d_B+1-d_A-d_B}{d_A d_B}.
\label{eq:dim_sub_add2}
\end{align}
The above are bounds on the same quantity $\mathcal{S}_{L}(\rho_{ABC})$. 
Now the only remaining question whether one bound is sharper than the other. This question depends on the state and the dimension. 
Let us set for example $d_{A}=d_{B}=d_{C}=2$, then the question whether \cref{eq:dim_sub_add2} is sharper than \cref{eq:AudenaertQSA2} is equivalent to
\begin{equation}
\frac{1}{2}\mathcal{S}_{L}(\rho_{AC})+\mathcal{S}_{L}(\rho_{B}) 
- \frac{1}{2} \mathcal{S}_{L}(\rho_{BC})
+ \frac{1}{4} \mathcal{S}_{L}\left(\rho_{C}\right)
 > \frac{1}{4}.
\end{equation}
Phrased like this it is obvious, that this is true for all highly mixed states.
For the maximally mixed state $\rho=\frac{1}{d_{A}d_{B}d_{C}}\mathds{1}$ each single party marginal fulfills $S_{L}(\Tr_{ij}(\frac{1}{8}\mathds{1}))=\frac{1}{2}$, while all two party marginals $S_{L}(\Tr_{j}(\frac{1}{8}\mathds{1}))=\frac{3}{4}$. This means that for the maximally mixed state the above evaluates to the true statement $\frac{5}{8}>\frac{1}{4}$. 
This gives one example where our dimension-dependent SSA is sharper than (weak) subadditivity and thus independent.
Similar results will hold for an $\varepsilon$-ball of states surrounding the maximally mixed state.
\begin{figure*}[t]
\begin{subfigure}{.33\textwidth}
  \includegraphics[width=.9\textwidth]{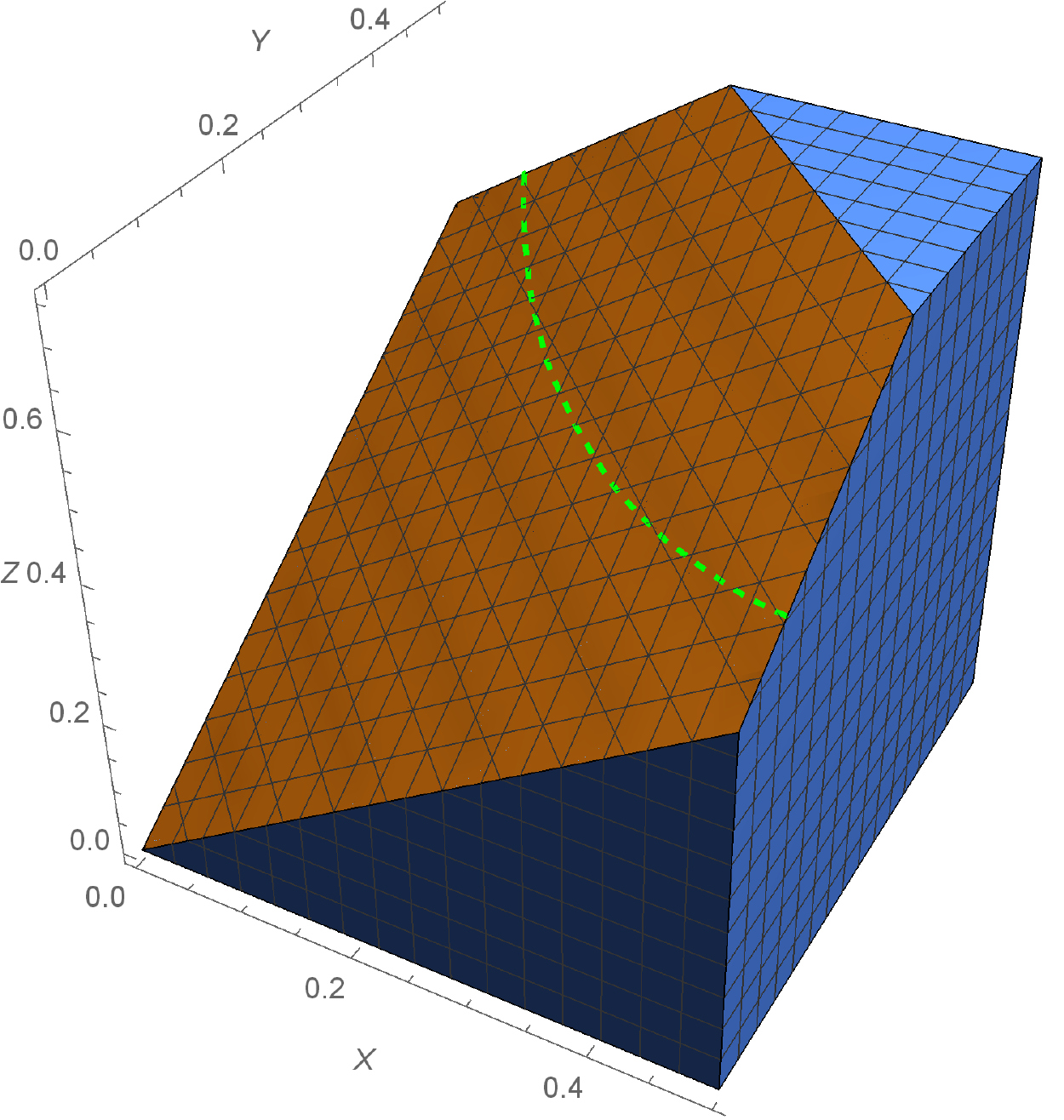}
  \caption{(Weak) Subadditivity}
  \label{fig:A1}
\end{subfigure}% 
\begin{subfigure}{.33\textwidth}
  \centering
  \includegraphics[width=.9\linewidth]{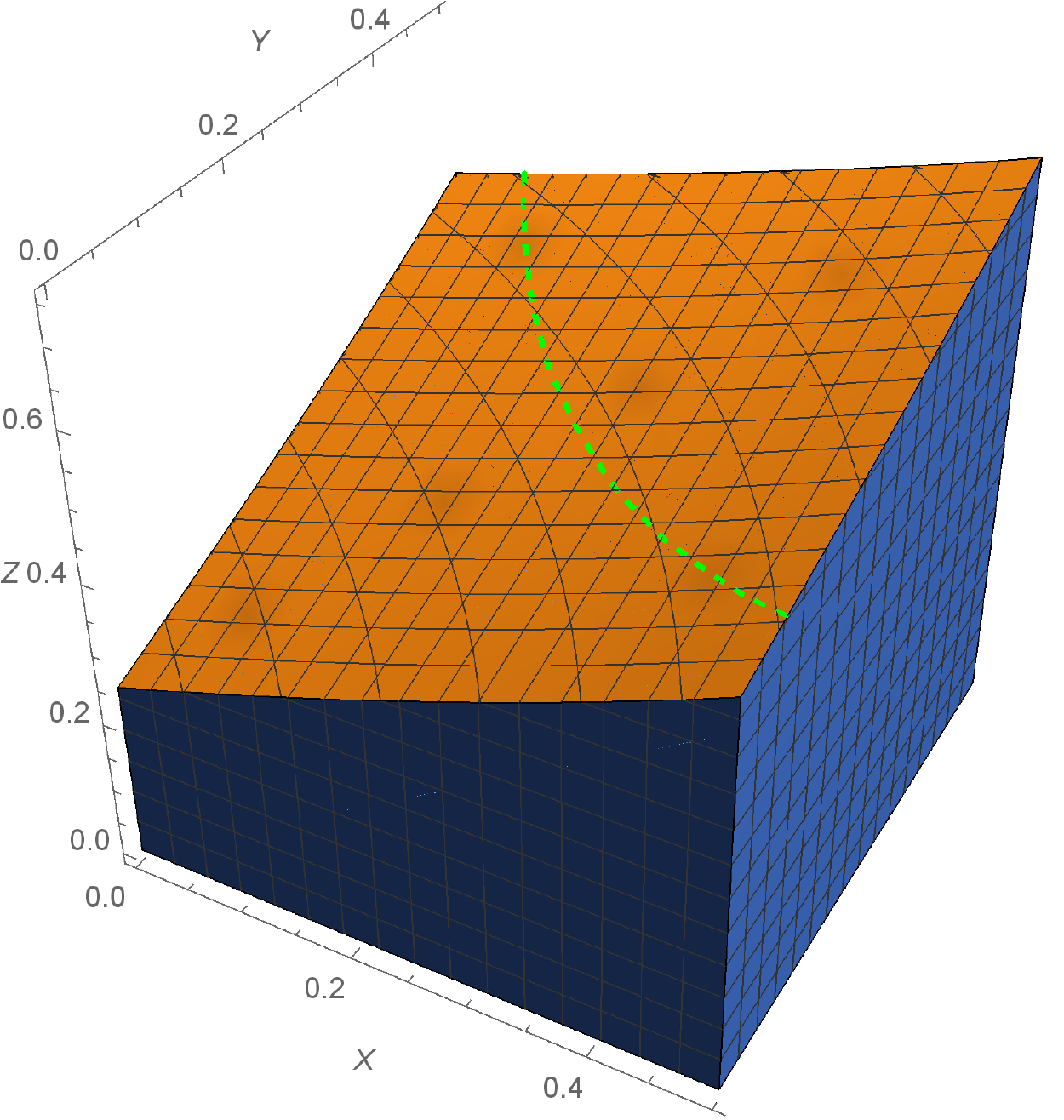}
  \caption{Generalized Pseudo-Additivity}
  \label{fig:A3}
\end{subfigure}

\caption{Comparison of different entropy inequalities for $d_{A}=d_{B}=2$.
On the X-Y plane we have plotted the values for $\mathcal{S}_{L}(A)$ and $\mathcal{S}_{L}(B)$, the Z-axis describes the maximal attainable value of $\mathcal{S}_{L}(AB)$ allowed by the corresponding entropy inequality.
This means that a lower  Z-value in the plots corresponds to the the respective inequality being sharper for the corresponding marginal values $\mathcal{S}_{L}(A)$ and $\mathcal{S}_{L}(B)$.
Above the green dashed line in \cref{fig:A1} and \cref{fig:A3} our generalized pseudo-additivity is sharper than subadditivity, below the line the subadditivity is the sharper entropy inequality.
}
\label{fig:A}
\end{figure*}

\section{A quadratic entropy inequality for the linear entropy from the Bloch Picture}
\label{Sec:NonlinEntropyInequalities}
\noindent
We know that in the special case of $q=\alpha=1$, the Shannon \cite{YeungCone} or von Neumann \cite{PippengerCone} entropy inequalities all describe a convex cone. All possible inequalities are described by linear combinations of certain elementary inequalities.
This is a remarkably easy and elegant situation.
Unfortunately similar results do not seem available for the remainder of the Tsallis family. 
Since the Tsallis 2-entropy features values between $0$ and $1$, we can nonetheless represent physical realizations in a hypercube of dimension $2^{n}-1$, which allows a particularly instructive visualization for $n=2$.

In the discussion of \cref{Sec:EntropyInequalities}, we were asking what Tsallis 2-entropy or linear entropy relations were possible for a tripartite system. In absence of complete classifications, we had to resort to a direct comparison with the most famous equation known to us: subadditivity. In this discussion we briefly consider the question of how q-entropies may be classified and find a new non-linear entropy inequality for the linear entropy.

Generally speaking, not many linear inequalities are available for $q$-entropies. Apart from the subadditivity \cite{AudenaertQSA}, the triangle inequality \cite{RasteginUnifiedEntropies} and the trivial choice of non-negativity not a lot is known. 
We know for a fact, that SSA does never hold for $q \neq 1$ \cite{PetzVirosztek}. 
There are some further information theoretic results such as Fannes type bounds and Lesche stability \cite{RasteginUnifiedEntropies}, but we are not aware of any further linear inequalities discussed in literature. 

However, in the special case of $q=2$ we can say more. 
Our \cref{thm:dim_sub_add}  delivers a new, albeit dimension-dependent, linear entropy inequality for $q=2$ and there may very well be more independent linear inequalities.

Even though we are not aware of any well-known non-linear entropy inequality we can construct a non-trivial example.
As an Ansatz we can consider the most simple non-linear situation: the quadratic case.
For tensor products we have a clear candidate. It is known in the literature that a q-entropy fulfills the so-called pseudo-additivity \cite{AbePseudoadditivity},\cite{PetzVirosztek}
\begin{equation}
\label{eq:Pseudo-Additivity}
S_{q}(\rho_{A} \otimes \rho_{B})= S_{q}(\rho_{A})+S_{q}(\rho_{B})+(1-q)S_{q}(\rho_{A})S_{q}(\rho_{B}).
\end{equation}
This gives $q$ a clear interpretation as a parametrization of the corresponding Tsallis entropies non-additivity.
Furthermore, it is a quadratic relation between a composite system and its marginals. Quadratic functions are non-linear but still simple enough to work with. 
The caveat is that pseudo-additivity only holds for product states $\rho_{A} \otimes \rho_{B}$.
In the next theorem we show, that with simple Bloch picture techniques we can remedy this drawback for $q=2$. 
\begin{theorem}
For all $\rho_{AB}$ the the linear entropy or Tsallis 2-entropy can be bounded as:
\begin{align}
\label{eq:nonlinqentropy}
1-\frac{d_{A}d_{B}}{4}
\left(1-\mathcal{S}_{L}(\rho_{AB})+\frac{1}{d_{A}d_{B}}\right)^{2}
\leq
 \mathcal{S}_{L}(\rho_{A})+\mathcal{S}_{L}(\rho_{B})-\mathcal{S}_{L}(\rho_{A})\mathcal{S}_{L}(\rho_{B}).
\end{align}
The \labelcref{pr:nonlinqentropy} can be found in the Appendix.
\end{theorem}
Using the purity bound we are able to turn the equality in \cref{eq:Pseudo-Additivity} into an inequality  and thus managed to also cover entangled or correlated composite systems $\rho_{AB}$ instead of the product state appearing on the left hand side of \cref{eq:Pseudo-Additivity}. 
This turns the pseudo-additivity into a general non-linear inequality for linear entropy that is applicable to all states. 

That this inequality is independent of subadditivity can be seen through the example of the maximally mixed state $\frac{1}{4}\mathds{1}_{AB}$ for $d_{A}=d_{B}=2$. 
Clearly, in this case $\mathcal{S}_{L}(\frac{1}{4}\mathds{1}_{AB})=\frac{3}{4}$, while its marginals attain the values $\mathcal{S}_{L}(\frac{1}{2}\mathds{1}_{A})=\mathcal{S}_{L}(\frac{1}{2}\mathds{1}_{B})=\frac{1}{2}$.
Thus $\frac{3}{4} \leq \frac{1}{2} + \frac{1}{2}$ fulfills subadditivity $\mathcal{S}^{q}(\rho_{AB})\leq \mathcal{S}^{q}(\rho_{A})+\mathcal{S}^{q}(\rho_{B})$, but is not sharp.
On the other hand \cref{eq:nonlinqentropy} evaluates to
$\frac{3}{4}=1-\left(\frac{5}{4}-\frac{3}{4} \right)^{2} \leq \frac{1}{2} + \frac{1}{2} - \frac{1}{4}=\frac{3}{4}$.
Therefore, \cref{eq:nonlinqentropy} does \emph{not only hold}, it is \emph{even tight}. 
We can conclude that it is an independent equation.
A purely linear description of the entropy inequalities for the linear entropy seems impossible.
Therefore we have to forfeit the hope to achieve a linear description in analogy to \cite{YeungCone} of the linear entropy. 

\begin{figure*}[t]
\begin{subfigure}{.33\textwidth}
  \includegraphics[width=.9\textwidth]{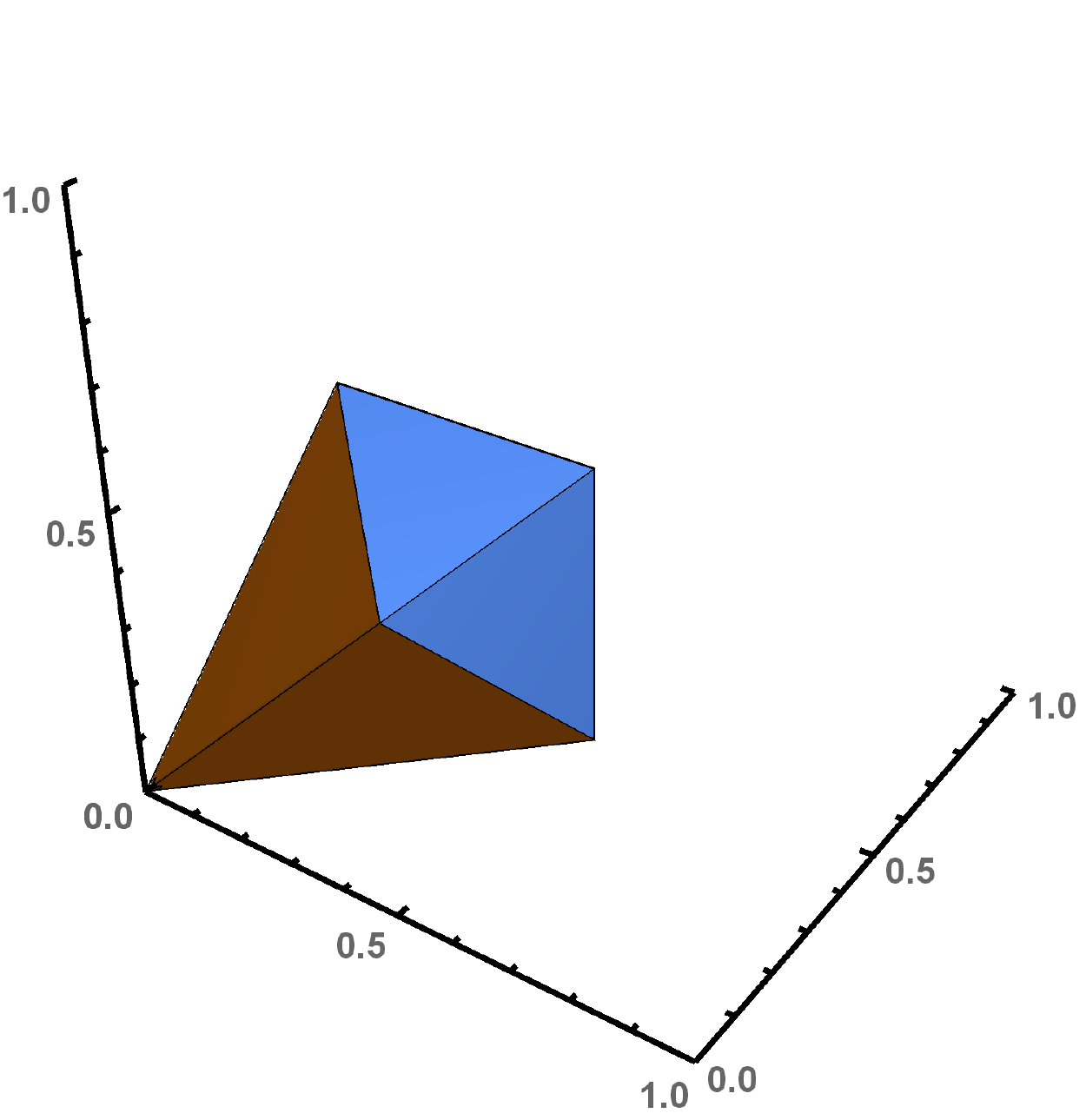}
  \caption{$d_{A}=d_{B}=2$ and $d_{C}=2$}
  \label{fig:B1}
\end{subfigure}% 
\begin{subfigure}{.33\textwidth}
  \includegraphics[width=.9\textwidth]{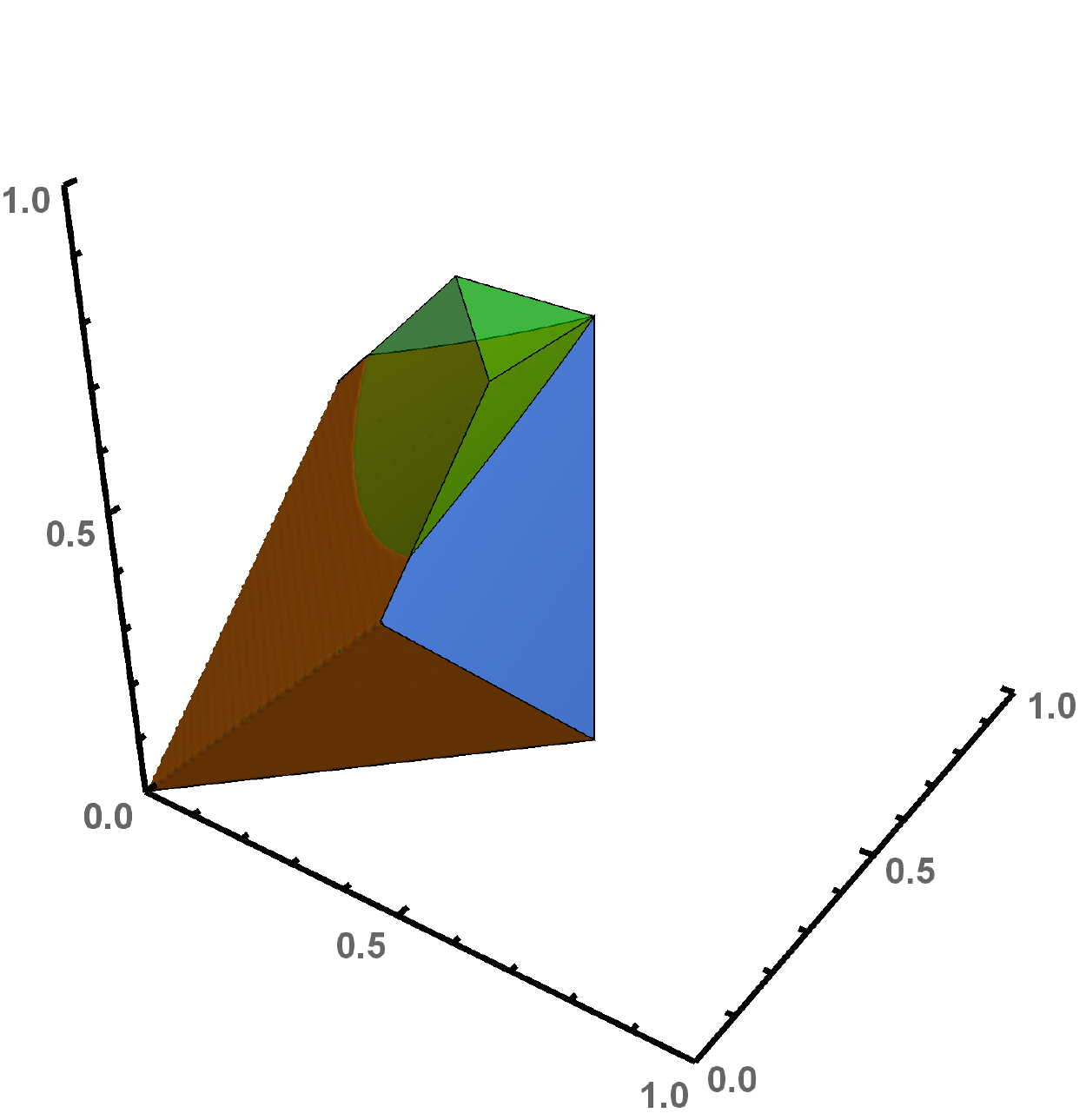}
  \caption{$d_{A}=d_{B}=2$ and $d_{C}=4$}
  \label{fig:B2}
\end{subfigure}
\begin{subfigure}{.33\textwidth}
  \centering
  \includegraphics[width=.9\linewidth]{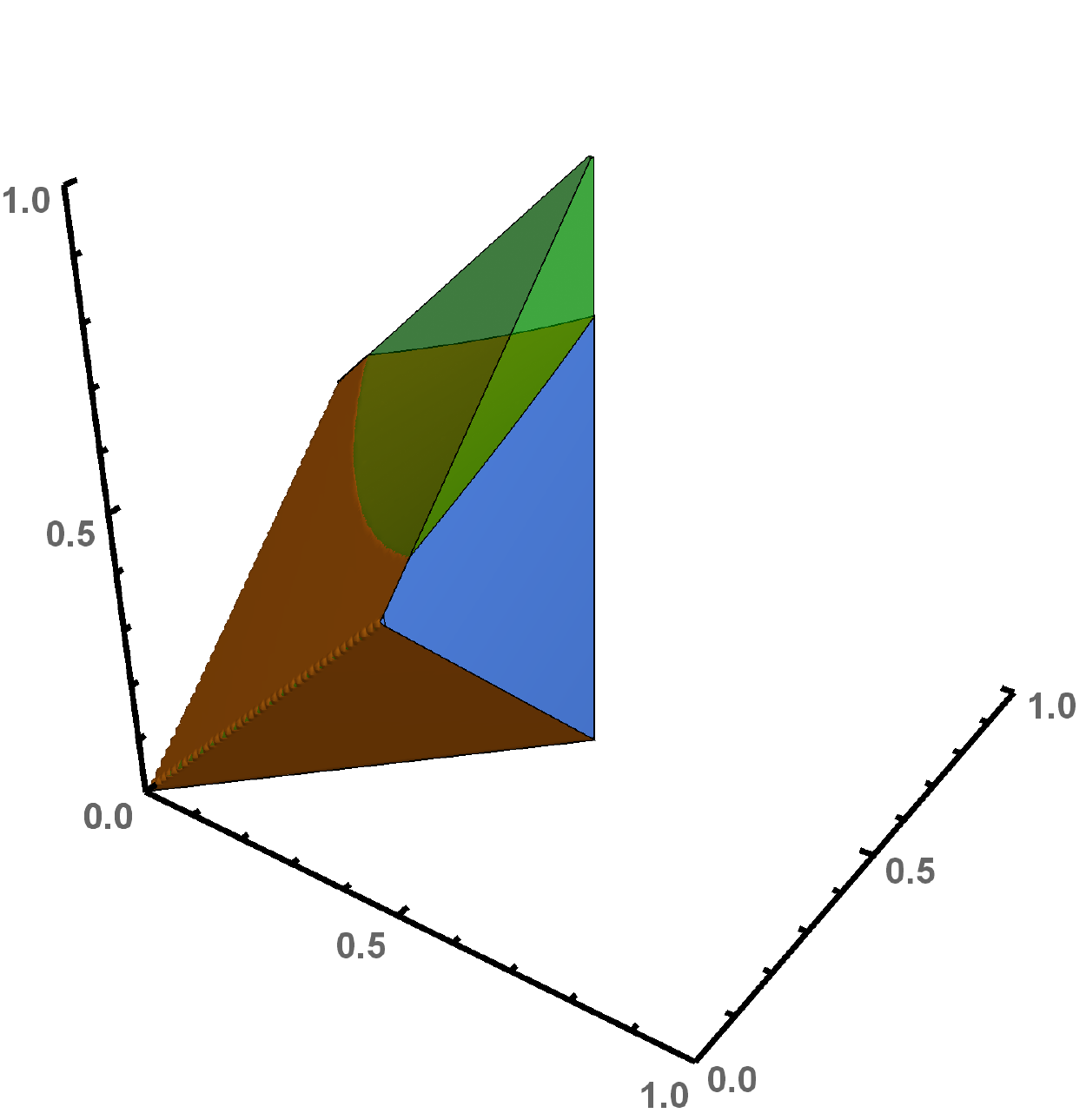}
  \caption{$d_{A}=d_{B}=2$ and $d_{C}=100$}
  \label{fig:B3}
\end{subfigure}
\caption{Every plotted point in the orange-blue body corresponds to a triple $\left(\mathcal{S}_{L}(A),\mathcal{S}_{L}(B),\mathcal{S}_{L}(C)\right)$ of a tripartite systems marginal entropies that are admissible in both subadditivity and generalized pseudo-additivity.
The green region  is depicted transparent and contains the points that are only admissible in subadditivity but not generalized pseudo-additivity. 
Thus the green region marks the region where generalized pseudo-additivity is sharper than (weak) subadditivity.}
\label{fig:B}
\end{figure*}
For a more geometric picture of the entropy space in the tripartite scenario refer to (see \cref{fig:A}).
There the relation of the marginals $\mathcal{S_{L}}(A)$ and $\mathcal{S_{L}}(B)$ to $\mathcal{S_{L}}(AB)$ is plotted.
We have depicted the attainable value of $\mathcal{S_{L}}(AB)$ for subadditivity and our generalized pseudo-additivity. 
A lower value on the $Z$-axis corresponds to the inequality being sharper with respect to the corresponding pair of marginal entropies $\mathcal{S_{L}}(A)$ and $\mathcal{S_{L}}(B)$.

The region of validity for subadditivity resembles a cube cut apart diagonally (see \cref{fig:A1}). 
The cut reveals a facet. 
In contrast the upper surface delimiting the generalized pseudo-additivity curves inward (see \cref{fig:A3}).
Together with the fact that edges are explicitly realizable it furthermore proves the non-convexity of the admissible entropy manifold.
With respect to standard subadditivity, we find a non-trivial behaviour that depends on dimensional factors and on the exact values of the involved entropic quantities. 
We have plotted as a dashed green line the intersection of both inequalities. 
The regions above the dashed line are stronger than subadditivity, while those towards the bottom and close to the origin are weaker (see \cref{fig:A1,fig:A3}). 
This shows that depending on the involved state one or the other relation may be sharper.
The exact size of these regions depends on the exact interrelations of the involved dimensions. 

To analyze the dimensional dependence of standard subadditivity and generalized pseudo-additivity refer to \cref{fig:B}.
In \cref{fig:B} we have plotted the body of admissible marginal entropies that fulfill all three equations given by subadditivity (green) and those that as well satisfy the three equations given by permutations of \cref{eq:nonlinqentropy} (orange-blue).
If all dimensions are equal, generalized pseudo-additivity (see \cref{eq:nonlinqentropy}) is not stronger than subadditivity. 
This is why \cref{fig:B1} only shows the cone defined by subadditivity. In fact \cref{fig:B1} represents a zoomed out version of \cref{fig:A1}.

Already for $d_{A}=d_{B}=2$ and $d_{C}=3$ generalized pseudo-additivity is sharper for some regions. 
The upper blue facet of the body in \cref{fig:B1} curves inward with increasing $d_{C}$ in total analogy to \cref{fig:A3}. 
In \cref{fig:B2,fig:B3} the upper part of the body turns orange indicating that it has become a non-linear surface.

The region where \cref{eq:nonlinqentropy} is stronger than subadditivity grows with $|d_{i}-d_{j}|$. 
In fact the cone described by only subadditivity grows with $|d_{i}-d_{j}|$, while the one defined by pseudo-additivity seems to saturate at $\min(d_{i},d_{j})^{2}$.
This is visible in \cref{fig:B2,fig:B3}, where the orange-blue body stopped growing. 
 
One can sum up, that generalized pseudo-additivity is relevant for systems with asymmetric dimensions. 
It clearly rules out the possibility to rely on purely linear descriptions, as the extremal points between the non-linear inequality are realizable and the non-linear surface in \cref{fig:B2,fig:B3} is slightly non-convex.

\section{Conclusion}
\label{sec:Conclusion}
\noindent
We have demonstrated that the correlation tensor formalism is a powerful tool to derive dimension-dependent statements about two seemingly different areas of research: monogamy relations and entropy inequalities.

First, we have introduced the \emph{split Bloch basis}, a sparse Bloch representation for quantum states with low-dimensional support. While all physical quantum states are of course expected to be full rank and thus an empty kernel, this tool can nonetheless be very useful for theoretical techniques that make use of purification and Schmidt decompositions. 

Using this representation we have shown that a natural monotone for quantifying correlations indeed exhibits monogamy in arbitrary dimensions. In particular, any amount of bipartite correlation between them non-trivially restricts the correlations of any external party with the system and maximal correlation implies decoupling from any external party.

Complementary, by using the very same techniques with some slightly different choices, we derived a number of inequalities for the well-known linear (or Tsallis 2-)entropy. 
First, we have a new linear inequality, a dimension-dependent analogue to strong subadditivity.
Second and maybe more interesting, we provide a non-linear but simple inequality for the linear entropy. 
 
We find that, while in the case of symmetric dimension subadditivity is strictly stronger than our inequality, it is nonetheless sharper for asymmetric dimensions.

An open question is whether our results could be made sharper. 
The key technique in all our theorems is the relation of the purity $\Tr(\rho^{2})$ to correlation tensor norms. Typically we employ purity here, however a sharper bound may lead to better results.
Possible candidates could be found in either \cite{JensMonogamy} or \cite{JensStateInverter}, where more involved relations for $\Tr(\rho^{2})$ are given. 
Similarly, in \eqref{TAB0} we simply set $\norm{T^{AB}}^2$ to zero for the derivation of the generalized pseudo-additivity. It may be that a sharper bound is possible here. 
Furthermore, the reliance on the connection between $\Tr(\rho^{2})$ and the Bloch parameterization during the proof of the generalized subadditivity limits the relation to the linear entropy. 
It would be certainly desirable to obtain a similarly general result for all $q$-entropies. Another interesting direction may be the derivation of non-trivial restrictions for correlation tensor elements from q-entropy inequalities.

It is still unclear if a complete geometric classification is possible, but it will be certainly more complicated than the entropy cones of the von Neumann entropy.
Still, one could hope to at least find a more complex body that contains all entropic relations at once.
\section{Acknowledgments}
\noindent
We especially thank Felix Huber for pointing out an error in the first arxiv version.
We thank Nicolai Friis, Milan Mosonyi and Matej Pivoluska for their comments.
PA, CK and MH acknowledge  support of funding from the Austrian Science Fund (FWF) through the START project Y879-N27 and the joint Czech-Austrian project MultiQUEST (I 3053-N27 and GF17-33780L).
CK especially thanks Babsi for being the best and dearest and excuses for not having had the opportunity to thank Barbara Weberndorfer accordingly in his thesis. He does sincerely hope time travel may one day allow to correct this error.

\newpage

\appendix
\section{Monogamy of Correlation Relations}
\label{Appendix:Monogamy}
\setcounter{lemma}{3}

\begin{lemma}[The Schmidt decomposition]
For pure states $\left|\psi\right\rangle_{AB}\in \mathcal{H}^{d_A}\otimes\left(\mathcal{H}^{d_A}\oplus\mathcal{H}^{d_B-d_A}\right)$ with $d_B>d_A$ we find the Schmidt decomposition as:
\begin{align}
 \left(
    \left|\psi\right\rangle
    \left\langle\psi\right|
 \right)_{AB}
    =  
    \sum_{i,j=0}^{\mathclap{\min(d_A,d_B)-1}}
        \left\langle
            \lambda^A_i
                \otimes
            \lambda^B_j 
        \right\rangle
        \lambda^A_i
        \otimes
        \lambda^B_j
\end{align}

\end{lemma}
\begin{proof}
\label{pr:Schmidt_decom}
We will simply prove the above statement by explicitly constructing the standard Schmidt basis from our Split Bloch basis:
First let us state the standard Schmidt decomposition for pure states $\left|\psi\right\rangle_{AB}=\sum_{i=0}^{\min(d_A,d_B)-1}c_i|ii\rangle$ in operator form:
\begin{align}
    \left(
    \left|\psi\right\rangle
    \left\langle\psi\right|
 \right)_{AB}
 =
 \sum_{i,j=0}^{\mathclap{\min(d_A,d_B)-1}} c_i c_j \left|ii\right\rangle\left\langle jj\right|
\end{align}
Let us construct first the diagonal $\left|i\right\rangle\left\langle i\right|$  in the local Bloch bases.
Note that the construction in the split Bloch basis and the canonical Gellmann matrix basis is analogue, since the $\lambda_{ij}$ and $\hat{\lambda}_{ij}$ of the split Bloch basis are simply the embedded Gellmann matrices in the bigger space. 
For simplicity we will concatenate the index of the diagonal Bloch basis elements, e.g. $\lambda_{ii}=\lambda_i$ 
\begin{align}
   \sqrt{d} \left|0\right\rangle\left\langle 0\right|
    &= 
    \lambda_0
    +
    \frac{1}{d-1}\lambda_{d-1}
    +
    \sum_{k=2}^{d-1} \alpha_{d-k} \lambda_{d-k}
    + \alpha_{d-i-1}\lambda_i 
    &
    \\
    \sqrt{d}
    \left|i\right\rangle\left\langle i\right|
    &= 
    \frac{1}{i}
    \left(
    \lambda_0
    +
    \frac{1}{d-1}\lambda_{d-1}
    +
    \sum_{k=2}^{d-i-1} \alpha_{d-k} \lambda_{d-k}
    - \alpha_{d-i-1}\lambda_i 
    \right)
    &\,|\, i\in \{1,\dots,d-2\}
    \\
    \sqrt{d}
    \left|d-1\right\rangle\left\langle d-1\right|
    &= 
    \frac{1}{d-1}
    \left(
    \lambda_0
    +
    \lambda_{d-1}
    \right)
\end{align}
with $\alpha_{d-k}=\Pi_{j=1}^{k-1}(1+\frac{1}{d-j})/d-k$ and $d$ the local dimension of the Hilbert space.
The diagonal elements of the bipartite basis are simply found as the tensor product $\left|ii\right\rangle\left\langle ii\right|=\left|i\right\rangle\left\langle i\right|_A\otimes\left|i\right\rangle\left\langle i\right|_B$ for all $i\in\{0,\dots,d_{min}-1\}$.
Where $d_{min}=\min(d_A,d_B)$.
For the off-diagonal elements we find:
\begin{align}
     \left|ii\right\rangle\left\langle jj\right|
     +\left|jj\right\rangle\left\langle ii\right|
    &=
    \frac{1}{2d_{min}}
        \left(
            \lambda_{ij}^A\otimes\lambda_{ij}^B
            +
            \hat{\lambda}_{ij}^A\otimes\hat{\lambda}_{ij}^B
        \right)
        &\,|\, i\neq j
    \\
    \shortintertext{this is true since:}
    \lambda_{ij}^A\otimes\lambda_{ij}^B
    &=
    \left|ij\right\rangle\left\langle ji\right|
    +\left|ji\right\rangle\left\langle ij\right|
    +\left|ii\right\rangle\left\langle jj\right|
    +\left|jj\right\rangle\left\langle ii\right|
    \\
    \shortintertext{and:}
    \hat{\lambda}_{ij}^A\otimes\hat{\lambda}_{ij}^B
    &=
    \left|ij\right\rangle\left\langle ji\right|
    +\left|ji\right\rangle\left\langle ij\right|
    -\left|ii\right\rangle\left\langle jj\right|
    -\left|jj\right\rangle\left\langle ii\right|
\end{align}
Note that we already assumed that the computational bases in which we expressed the Schmidt basis and the split Bloch basis are the same. 
This is not necessarily the case, however, any two computational bases are connected by unitary transformations.
\end{proof}
\setcounter{theorem}{0}
\begin{theorem}
Let $\rho_{ABE} \in \mathcal{H}=\mathcal{H}^{d^2}\otimes\mathcal{H}^{d_E}$ be an arbitrary tripartite state owned by $A,B$ and $E$ with local dimensions $d_A=d_B=d$ and $d_E$:
\begin{enumerate}[(i):]
    \item  The correlation of a composite system $\rho_{AB}$  with an arbitrary system $\rho_E$ limits the correlation of its marginals with the same:
        \begin{align}
        \mathcal{T}_{A|E}\left(\rho_{ABE}\right) 
        + \mathcal{T}_{B|E}\left(\rho_{ABE}\right) 
        &\leq \frac{g_{AB|E}}{\min\left(g_{A|E},g_{B|E}\right)} \mathcal{T}_{AB|E}\left(\rho_{ABE}\right)  
        \\
        \shortintertext{
    \item The correlation of any state $\rho_{AB} \in \mathcal{H}^{d^2}$ with an arbitrary state $\rho_E$ is restricted by:
        }
            \mathcal{T}_{AB|E}
            &\leq
            \frac{
                d^{4}  
                -1
                -2
                \left(
                    \left\|T^{A}\right\|^2+\left\|T^{B}\right\|^2
                \right)
                -2 g_{A|B} \mathcal{T}_{A|B}
            }
            {g_{AB|E}}
        \end{align}

\end{enumerate}
\end{theorem}
\begin{proof}
\label{pr:2-part-less-3-part}
\begin{enumerate}[label={\text{add} (\roman*):},wide, labelwidth=!, labelindent=0pt]
\item
By the definition of $\mathcal{T}_{A|E}$, $\mathcal{T}_{B|E}$ and $\mathcal{T}_{AB|E}$ we have
\begin{align}
\min\left(g_{A|E},g_{B|E}\right)\left( \mathcal{T}_{A|E} +  \mathcal{T}_{B|E}\right) 
& \leq
g_{A|E}\; \mathcal{T}_{A|E} + g_{B|E}\; \mathcal{T}_{B|E} \\
& =
\left\|T^{AE}\right\|^2+\left\|T^{BE}\right\|^2 
\leq  
\left\|T^{AE}\right\|^2+\left\|T^{BE}\right\|^2 +
\left\|T^{ABE}\right\|^2 
=  g_{AB|E} \mathcal{T}_{AB|E} 
\end{align}
\item This proof makes use of the split Bloch decomposition. First we need to state that for every $\rho_{ABE}$ with a fixed marginal $\rho_{AB}=\sigma$ there exists a purification $|\psi_{ABEE'}\rangle$.  A purification is however not unique, i.e. all purifications are connected by an isometry: $\left|\psi^\prime\right\rangle_{ABEE'} = \mathds{1}\otimes V \left|\psi\right\rangle_{ABEE'} \in \mathcal{H}^{d^2}\otimes \mathcal{H}^{d^\prime>d_{min}}$ with the isometry 
$V:\mathcal{B}\left(\mathcal{H}^{d_{min}},\mathcal{H}^{d^\prime}\right)$.

Using these considerations allows us to use the Schmidt decomposition for pure states as well as the fact that $\Tr\left(\left(\left|\psi\right\rangle\left\langle\psi\right|\right)_{ABEE'}^2\right)=1$.
Furthermore it is very useful to use the split Bloch basis. 
For sake of brevity let us introduce the shorthand subscript $EE'\rightarrow \tilde{E}$ and thus $c_0=\left\langle\lambda_0^{AB}\otimes\lambda^{\tilde{E}}_0\right\rangle^2$ and $c_0^\prime=\left\langle\lambda_0^{AB}\otimes\mu^{\tilde{E}}_0\right\rangle^2$.
\begin{align}
\Tr
\left(
    \left(
    \left|\psi\right\rangle\left\langle\psi\right|
    \right)_{AB\tilde{E}}^2
\right)
    =
    1 
    &=
    \frac{1}{d^4}
        \left(
            c_0
            +\norm{T^A}^2
            +\norm{T^B}^2
            +\norm{T^{\tilde{E}}_{SD}}^2
            +\norm{T^{AB}}^2
            +\norm{T^{A\tilde{E}}_{SD}}^2
            +\norm{T^{B\tilde{E}}_{SD}}^2
            +\norm{T^{AB\tilde{E}}_{SD}}^2
                \vphantom{\norm{\tilde{T}^{\tilde{E}}}^2}
        \right)
    \nonumber
    \\
    & \phantom{=}\, +
        \frac{1}{\left(d_{\tilde{E}}-d^2\right)d^2}
        \left(
            c_0^\prime
            +\norm{\tilde{T}^{\tilde{E}}}^2
            +\norm{\tilde{T}^{A\tilde{E}}}^2
            +\norm{\tilde{T}^{B\tilde{E}}}^2
            +\norm{\tilde{T}^{AB\tilde{E}}}^2
        \right)
    \end{align}
    By simply rewriting this equation we find:
\begin{align}
    \norm{T^{A\tilde{E}}_{SD}}^2
    +\norm{T^{B\tilde{E}}_{SD}}^2
    +\norm{T^{AB\tilde{E}}_{SD}}^2
    &=
    d^4
    -
    \left(
        c_0
        +\norm{T^A}^2
        +\norm{T^B}^2
        +\norm{T^{\tilde{E}}_{SD}}^2
        +\norm{T^{AB}}^2
            \vphantom{\norm{\tilde{T}^{\tilde{E}}}^2}
    \right)
    \nonumber
    \\
    & \phantom{=}\, -
        \frac{d^2}{\left(d_{\tilde{E}}-d^2\right)}
        \left(
            c_0^\prime
            +\norm{\tilde{T}^{\tilde{E}}}^2
            +\norm{\tilde{T}^{A\tilde{E}}}^2
            +\norm{\tilde{T}^{B\tilde{E}}}^2
            +\norm{\tilde{T}^{AB\tilde{E}}}^2
        \right)
    \\
    \shortintertext{
Now we can use the Schmidt decomposition, i.e.:}
    \frac{1}{d^2}
        \left(
            1
            +\norm{T^A}^2
            +\norm{T^B}^2
            +\norm{T^{AB}}^2
                \vphantom{\norm{\tilde{T}^{\tilde{E}}}^2}
        \right)
    &=
    \frac{1}{d^2}
    \left(
        c_0
        +\norm{T^{\tilde{E}}_{SD}}^2
            \vphantom{\norm{\tilde{T}^{\tilde{E}}}^2}
    \right)
    +\frac{1}{\left(d_{\tilde{E}}-d^2\right)d^2}
        \left(
            c_0^\prime
            +\norm{\tilde{T}^{\tilde{E}}}^2
        \right)
    \\
    \norm{T^{\tilde{E}}_{SD}}^2
    &=
    1
    +\norm{T^A}^2
    +\norm{T^B}^2
    +\norm{T^{AB}}^2
    - c_0
    -\frac{1}{\left(d_{\tilde{E}}-d^2\right)}
        \left(
            c_0^\prime
            +\norm{\tilde{T}^{\tilde{E}}}^2
        \right)
    \\
    \shortintertext{Plugging this in above we find:}
    \norm{T^{A\tilde{E}}_{SD}}^2
    +\norm{T^{B\tilde{E}}_{SD}}^2
    +\norm{T^{AB\tilde{E}}_{SD}}^2
    &=
    d^4
    -
    \left(
        c_0
        +2\norm{T^A}^2
        +2\norm{T^B}^2
        +2\norm{T^{AB}}^2
        +1
        -c_0
            \vphantom{\norm{\tilde{T}^{\tilde{E}}}^2}
    \right)
    \nonumber
    \\
    & \phantom{=}\, -
    \underbrace{
    \frac{d^2-1}{\left(d_{\tilde{E}}-d^2\right)}
    \left(
        c_0^\prime
        +\norm{\tilde{T}^{\tilde{E}}}^2
    \right)
    -\frac{d^2}{\left(d_{\tilde{E}}-d^2\right)}
        \left(
            \norm{\tilde{T}^{A\tilde{E}}}^2
            +\norm{\tilde{T}^{B\tilde{E}}}^2
            +\norm{\tilde{T}^{AB\tilde{E}}}^2
        \right)}_{\Delta}
    \\
    \shortintertext{Due to \cref{lem:Schmidt_decom} we know there exists a basis such that $\Delta=0$, since the $\Tr\left(\rho_{AB}^2\right)$ is invariant under changing the local Bloch basis, this is achieved by maximizing the left hand side: }
    \max_{\{\lambda_i^{\tilde{E}}\}}\left(
    \norm{T^{A\tilde{E}}_{SD}}^2
    +\norm{T^{B\tilde{E}}_{SD}}^2
    +\norm{T^{AB\tilde{E}}_{SD}}^2
    \right)
    &=
    d^4
    - 1
    - 2
        \left(
            \norm{T^A}^2
            +\norm{T^B}^2
            +g_{A|B} \mathcal{T}_{A|B}
        \right)
    \\
    g_{AB|\tilde{E}}\mathcal{T}_{AB|\tilde{E}}\left(|\psi_{AB\tilde{E}}\rangle\langle\psi_{AB\tilde{E}}|\right)
    &=
    d^4
    - 1
    - 2
        \left(
            \norm{T^A}^2
            +\norm{T^B}^2
            +g_{A|B} \mathcal{T}_{A|B}
        \right)
\end{align}
Finally due to Lemma \ref{lem:traceoutlemma} (ii) we find:
\begin{align}
    g_{AB|E}\mathcal{T}_{AB|E}\left(\left(\left|\psi\right\rangle\left\langle\psi\right|
    \right)_{AB\tilde{E}}\right)
    &\leq
        g_{AB|\tilde{E}}\mathcal{T}_{AB|\tilde{E}}
        \left(\left(\left|\psi\right\rangle\left\langle\psi\right|
    \right)_{AB\tilde{E}}\right)
    =
    d^4
    - 1
    - 2
        \left(
            \norm{T^A}^2
            +\norm{T^B}^2
            +g_{A|B} \mathcal{T}_{A|B}
        \right)
    \\
    \shortintertext{and since $\mathcal{T}_{AB|E}$ is independent of $E'$ we can take the partial trace in the argument over $E'$, i.e. $\mathcal{T}_{AB|E}\left(\left(\left|\psi\right\rangle\left\langle\psi\right|
    \right)_{AB\tilde{E}}\right)=\mathcal{T}_{AB|E}\left(\Tr_E'\left(\left(\left|\psi\right\rangle\left\langle\psi\right|
    \right)_{AB\tilde{E}}\right)\right)=\mathcal{T}_{AB|E}\left(\rho_{ABE}\right)$.}
    g_{AB|E}\mathcal{T}_{AB|E}\left(\rho_{ABE}\right)
    &\leq
        d^4
        - 1
        - 2
            \left(
                \norm{T^A}^2
                +\norm{T^B}^2
                +g_{A|B} \mathcal{T}_{A|B}
            \right)
\end{align}
\end{enumerate}
\end{proof}

\setcounter{lemma}{5}
\begin{lemma}
\label{lem:lower-tensor_bound}
The sum of correlation tensor norms in \cref{thm:2-part-less-3-part} (ii) is bounded from above by: 
\begin{align*}
\norm{T^{A}}^{2} + \norm{T^{B}}^{2} \leq \min\left( 2d -2,\left(d^{2}-1\right)\left(1- \mathcal{T}_{A|B}\right)\right) 
\end{align*}
and from below by:
\begin{align*}
\norm{T^{A}}^{2} + \norm{T^{B}}^{2} \geq \max\left(0,\frac{d^{2}}{d_E}-1-\left(d^{2}-1\right) \mathcal{T}_{A|B}\right)
\end{align*}
\end{lemma}
\begin{proof}
We will start by proving the upper bound for $\norm{T^{A}}^{2}+\norm{T^{B}}^{2}$.
A very simple connection can be found by using purity:
\begin{align}
d^{2} - 1 & \geq \norm{T^{A}}^{2} + \norm{T^{B}}^{2} + \norm{T^{AB}}^{2} \\
d^{2}-1- g_{A|B}\mathcal{T}_{A|B} & \geq \norm{T^{A}}^{2} + \norm{T^{B}}^{2}\\
\left(d^{2}-1\right)\left(1- \mathcal{T}_{A|B}\right) & = \norm{T^{A}}^{2} + \norm{T^{B}}^{2}
\end{align}
Now by varying $0\leq\mathcal{T}_{A|B}\leq 1$ we can estimate the entanglement Eve possesses. 
However due to the fact that the squared Euclidean norm of a correlation tensor is realized by a pure state,  another bound has to be considered:
\begin{align*}
2d -2 \geq\norm{T^{A}}^{2} + \norm{T^{B}}^{2}
\end{align*}
Thus we found an upper bound for:
\begin{align*}
\min\left( 2d -2,\left(d^{2}-1\right)\left(1- \mathcal{T}_{A|B}\right)\right) \geq\norm{T^{A}}^{2} + \norm{T^{B}}^{2}
\end{align*}
For the lower bound we will use:
\begin{align*}
    \frac{1}{d_E}\left(1+\norm{T^{E}}^2\right)&= 
    \frac{1}{d^{2}}\left(1+\norm{T^{A}}^2+\norm{T^{B}}^2+
    \norm{T^{AB}}^2\right)
    \\
    \frac{d^{2}}{d_E}-1&\leq\norm{T^{A}}^2+\norm{T^{B}}^2+
    \norm{T^{AB}}^2 
    \\
    \frac{d^{2}}{d_E}-1-g_{A|B}\mathcal{T}_{A|B} &\leq\norm{T^{A}}^2+\norm{T^{B}}^2 
    \\
    \frac{d^{2}}{d_E}-1-g\left(d^{2}-1\right)\mathcal{T}_{A|B} &\leq\norm{T^{A}}^2+\norm{T^{B}}^2 
\end{align*}
We know however that $\norm{T^{A}}^2+\norm{T^{B}}^2>0$ thus we find the lower bound as :
\begin{align*}
    \max\left(0,\frac{d^{2}}{d_E}-1-\left(d^{2}-1\right) \mathcal{T}_{A|B}\right)
\end{align*}
Which concludes our proof.
\end{proof}

\section{Entropy Inequalities}

\begin{theorem}
For a tripartite quantum system  $\rho_{ABC}$ we find the following entropy  inequality for the linear entropy $\mathcal{S}_{L}\left(\rho_{ABC}\right)=1- \Tr\left(\rho_{ABC}^{2}\right)$:
\begin{align}
\mathcal{S}_{L}\left(\rho_{ABC}\right)
+ \frac{1}{d_Ad_B} \mathcal{S}_{L}\left(\rho_{C}\right)
&\leq
\frac{1}{d_B}\mathcal{S}_{L}\left(\rho_{AC}\right)
+ \frac{1}{d_A} \mathcal{S}_{L}\left(\rho_{BC}\right)
+\frac{d_A d_B+1-d_A-d_B}{d_A d_B} 
\end{align}
\begin{proof}
\label{pr:dim_sub_add}
We rewrite the entire system into the relevant subsystems
\begin{align}
d_Ad_Bd_C \Tr\left(\rho_{ABC}^{2}\right)
&= 
		1 
		+ \norm{{T}^{A}}^{2}
		+ \norm{{T}^{B}}^{2}
		+ \norm{{T}^{C}}^{2}
		+ \norm{{T}^{AB}}^{2}
		+ \norm{{T}^{AC}}^{2}
		+ \norm{{T}^{BC}}^{2}
		+ \norm{{T}^{ABC}}^{2}
\\
&= \left(
		1 
		+ \norm{{T}^{A}}^{2}	
		+ \norm{{T}^{C}}^{2}
		+ \norm{{T}^{AC}}^{2}
	\right)
+  \left(
		1 
		+ \norm{{T}^{B}}^{2}	
		+ \norm{{T}^{C}}^{2}
		+ \norm{{T}^{BC}}^{2}
	\right)
-	\left(
		1 + \norm{{T}^{C}}^{2}
	\right)
\\
&=
d_A d_C \Tr\left(\rho_{AC}^{2}\right)
+ d_B d_C \Tr\left(\rho_{BC}^{2}\right)
- d_C\; \Tr\left(\rho_{C}^{2}\right). 
\end{align}
We can rewrite $\Tr(\rho^{2})$ into $S_L$ by its definition (see \cref{Def:SL}) and obtain 
\begin{align}
1-\mathcal{S}_{L}\left(\rho_{ABC}\right)
+ \frac{1}{d_Ad_B}\left(1-\mathcal{S}_{L}\left(\rho_{C}\right)\right) 
&\geq
\frac{1}{d_B} \left(1-\mathcal{S}_{L}\left(\rho_{AC}\right)\right)
+ \frac{1}{d_A} \left(1-\mathcal{S}_{L}\left(\rho_{BC}\right)\right),
\end{align}
or equivalently our claim
\begin{align}
\mathcal{S}_{L}\left(\rho_{ABC}\right)
+ \frac{1}{d_Ad_B} \mathcal{S}_{L}\left(\rho_{C}\right)
&\leq
\frac{1}{d_B}\mathcal{S}_{L}\left(\rho_{AC}\right)
+ \frac{1}{d_A} \mathcal{S}_{L}\left(\rho_{BC}\right)
+\frac{d_A d_B+1-d_A-d_B}{d_A d_B}.
\end{align}
\end{proof}
\end{theorem}
\begin{theorem}
For all $\rho_{AB}$ and the linear entropy or Tsallis 2-entropy we have

\begin{equation}
1-\frac{d_{A}d_{B}}{4}
\left(1-\mathcal{S}_{L}(\rho_{AB})+\frac{1}{d_{A}d_{B}}\right)^{2}
\leq
 S_{L}(\rho_{A})+S_{L}(\rho_{B})-S_{L}(\rho_{A})S_{L}(\rho_{B}).
\end{equation}
\end{theorem}

\begin{proof}
\label{pr:nonlinqentropy}
We prove the desired equation by showing a kind of  ``sub-multiplicativity'' $f(\mathcal{S}_{L}(\Tr(\rho_{AB})^{2}))\leq\mathcal{S}_{L}(\Tr(\rho_{A}\otimes\rho_{B})^{2})$ of $\mathcal{S}_{L}$ under the tensor product.
For this we express $\Tr(\rho_{AB})^{2}$ in the Bloch picture. We make use of the relation
$\Tr\left(\rho_{A}^{2}\right)=\frac{1}{d_{A}}\left( 1 + \norm{T^A}^2\right)$ to obtain
\begin{align}
\Tr(\rho_{AB})^{2}
&=
\frac{1}{d_{A}d_{B}}
\left( 
    1 
    + \norm{T^A}^2 
    + \norm{T^B}^2 
    + \norm{T^{AB}}^2 
\right)
\\
\label{TAB0}
&= 
\frac{1}{d_{A}d_{B}}
\left( 
    1 
    + \left(d_{A}\Tr\left(\rho_{A}^{2}\right) 
        - 1 
        \right) 
    + \left(d_{B}\Tr\left(\rho_{B}^{2}\right) 
        - 1
        \right)
    + \norm{T^{AB}}^2 
\right)
\\
&\geq
\frac{1}{d_{A}d_{B}}
\left(
    d_{A}\Tr\left(\rho_{A}^{2}\right) 
    + d_{B}\Tr\left(\rho_{B}^{2}\right) 
    - 1   
\right) 
\\
&=
\frac{1}{d_{A}d_{B}}
\left(
    d_{A}d_{B}\Tr\left(\left(\rho_{A}\otimes\rho_{B}\right)^{2}\right) 
    \left(
        \frac{1}{d_{B}\Tr\left(\rho_{B}^{2}\right)} 
        +\frac{1}{d_{A}\Tr\left(\rho_{A}^{2}\right)} 
    \right) -1
\right)
\end{align}
Finally by the Arithmetic-Geometric mean inequality $a+b \geq 2 \sqrt{ab}$ we got
\begin{align}
\label{eq:GeneralizedPseudoAdditivity1}
\Tr(\rho_{AB}^{2}) 
&\geq
\frac{1}{d_{A}d_{B}}
\left(
    d_{A}d_{B}\Tr\left(\left(\rho_{A}\otimes\rho_{B}\right)^{2}\right)
    \left( 
        \frac{1}{d_{B}\Tr\left(\rho_{B}^{2}\right)} 
        + \frac{1}{d_{A}\Tr\left(\rho_{A}^{2}\right)}
    \right) 
    - 1
\right)
\\
&=
\frac{1}{d_{A}d_{B}}
\left(
    d_{A}d_{B}\Tr\left(\left(\rho_{A}\otimes\rho_{B}\right)^{2}\right)
    \left( 
        \frac{2}{\sqrt{d_{A}d_{B}\Tr\left(\rho_{A}^{2}\right)\Tr\left(\rho_{B}^{2}\right)}}
    \right) 
    - 1
\right)
\\
&= 
\frac{2}{\sqrt{d_{A}d_{B}}}
\sqrt{\Tr\left(\left(\rho_{A}\otimes\rho_{B}\right)^{2}\right)}
- \frac{1}{d_{A}d_{B}}
\end{align}
or alternatively $\Tr\left(\left(\rho_{A}\otimes\rho_{B}\right)^{2}\right)
\leq
\frac{d_{A}d_{B}}{4}
\left(\Tr\left(\rho_{AB}^{2}\right)
+\frac{1}{d_{A}d_{B}}\right)^{2}$.

The pseudo-additivity \cite{AbePseudoadditivity},\cite{PetzVirosztek} for $q=2$ is
\begin{align}
 S_{L}(\rho_{A})
 +S_{L}(\rho_{B})
 -S_{L}(\rho_{A})S_{L}(\rho_{B})
&=
 S_{L}(\rho_{A} \otimes \rho_{B})
=
 1-\Tr\left(\left(\rho_{A} \otimes \rho_{B}\right)^{2}\right)
\\
&
\geq
1
-\frac{d_{A}d_{B}}{4}
\left(
    \Tr\left(\rho_{AB}^{2}\right)
    +\frac{1}{d_{A}d_{B}}
\right)^{2},
\\
\shortintertext{leaving us with}
\\
S_{L}(\rho_{A})
+S_{L}(\rho_{B})
-S_{L}(\rho_{A})S_{L}(\rho_{B})
&\geq
1-\frac{d_{A}d_{B}}{4}
\left(1-\mathcal{S}_{L}(\rho_{AB})+\frac{1}{d_{A}d_{B}}\right)^{2} 
\end{align}
\end{proof}

\end{document}

%% file: worst_case.pgf
%% Creator: Matplotlib, PGF backend
%%
%% To include the figure in your LaTeX document, write
%%   \input{<filename>.pgf}
%%
%% Make sure the required packages are loaded in your preamble
%%   \usepackage{pgf}
%%
%% Figures using additional raster images can only be included by \input if
%% they are in the same directory as the main LaTeX file. For loading figures
%% from other directories you can use the `import` package
%%   \usepackage{import}
%% and then include the figures with
%%   \import{<path to file>}{<filename>.pgf}
%%
%% Matplotlib used the following preamble
%%   \usepackage[utf8x]{inputenc}
%%   \usepackage[T1]{fontenc}
%%
\begingroup%
\makeatletter%
\begin{pgfpicture}%
\pgfpathrectangle{\pgfpointorigin}{\pgfqpoint{4.837986in}{3.191777in}}%
\pgfusepath{use as bounding box, clip}%
\begin{pgfscope}%
\pgfsetbuttcap%
\pgfsetmiterjoin%
\definecolor{currentfill}{rgb}{1.000000,1.000000,1.000000}%
\pgfsetfillcolor{currentfill}%
\pgfsetlinewidth{0.000000pt}%
\definecolor{currentstroke}{rgb}{1.000000,1.000000,1.000000}%
\pgfsetstrokecolor{currentstroke}%
\pgfsetdash{}{0pt}%
\pgfpathmoveto{\pgfqpoint{0.000000in}{-0.000000in}}%
\pgfpathlineto{\pgfqpoint{4.837986in}{-0.000000in}}%
\pgfpathlineto{\pgfqpoint{4.837986in}{3.191777in}}%
\pgfpathlineto{\pgfqpoint{0.000000in}{3.191777in}}%
\pgfpathclose%
\pgfusepath{fill}%
\end{pgfscope}%
\begin{pgfscope}%
\pgfsetbuttcap%
\pgfsetmiterjoin%
\definecolor{currentfill}{rgb}{1.000000,1.000000,1.000000}%
\pgfsetfillcolor{currentfill}%
\pgfsetlinewidth{0.000000pt}%
\definecolor{currentstroke}{rgb}{0.000000,0.000000,0.000000}%
\pgfsetstrokecolor{currentstroke}%
\pgfsetstrokeopacity{0.000000}%
\pgfsetdash{}{0pt}%
\pgfpathmoveto{\pgfqpoint{0.625853in}{0.467414in}}%
\pgfpathlineto{\pgfqpoint{4.662561in}{0.467414in}}%
\pgfpathlineto{\pgfqpoint{4.662561in}{3.042715in}}%
\pgfpathlineto{\pgfqpoint{0.625853in}{3.042715in}}%
\pgfpathclose%
\pgfusepath{fill}%
\end{pgfscope}%
\begin{pgfscope}%
\pgfpathrectangle{\pgfqpoint{0.625853in}{0.467414in}}{\pgfqpoint{4.036708in}{2.575301in}} %
\pgfusepath{clip}%
\pgfsetrectcap%
\pgfsetroundjoin%
\pgfsetlinewidth{1.003750pt}%
\definecolor{currentstroke}{rgb}{1.000000,0.498039,0.054902}%
\pgfsetstrokecolor{currentstroke}%
\pgfsetdash{}{0pt}%
\pgfpathmoveto{\pgfqpoint{0.625853in}{3.042715in}}%
\pgfpathlineto{\pgfqpoint{4.662561in}{0.467414in}}%
\pgfpathlineto{\pgfqpoint{4.662561in}{0.467414in}}%
\pgfusepath{stroke}%
\end{pgfscope}%
\begin{pgfscope}%
\pgfpathrectangle{\pgfqpoint{0.625853in}{0.467414in}}{\pgfqpoint{4.036708in}{2.575301in}} %
\pgfusepath{clip}%
\pgfsetrectcap%
\pgfsetroundjoin%
\pgfsetlinewidth{1.003750pt}%
\definecolor{currentstroke}{rgb}{0.839216,0.152941,0.156863}%
\pgfsetstrokecolor{currentstroke}%
\pgfsetdash{}{0pt}%
\pgfpathmoveto{\pgfqpoint{0.625853in}{2.184281in}}%
\pgfpathlineto{\pgfqpoint{4.662561in}{0.467414in}}%
\pgfpathlineto{\pgfqpoint{4.662561in}{0.467414in}}%
\pgfusepath{stroke}%
\end{pgfscope}%
\begin{pgfscope}%
\pgfpathrectangle{\pgfqpoint{0.625853in}{0.467414in}}{\pgfqpoint{4.036708in}{2.575301in}} %
\pgfusepath{clip}%
\pgfsetrectcap%
\pgfsetroundjoin%
\pgfsetlinewidth{1.003750pt}%
\definecolor{currentstroke}{rgb}{0.580392,0.403922,0.741176}%
\pgfsetstrokecolor{currentstroke}%
\pgfsetdash{}{0pt}%
\pgfpathmoveto{\pgfqpoint{0.625853in}{1.755065in}}%
\pgfpathlineto{\pgfqpoint{4.662561in}{0.467414in}}%
\pgfpathlineto{\pgfqpoint{4.662561in}{0.467414in}}%
\pgfusepath{stroke}%
\end{pgfscope}%
\begin{pgfscope}%
\pgfpathrectangle{\pgfqpoint{0.625853in}{0.467414in}}{\pgfqpoint{4.036708in}{2.575301in}} %
\pgfusepath{clip}%
\pgfsetrectcap%
\pgfsetroundjoin%
\pgfsetlinewidth{1.003750pt}%
\definecolor{currentstroke}{rgb}{0.172549,0.627451,0.172549}%
\pgfsetstrokecolor{currentstroke}%
\pgfsetdash{}{0pt}%
\pgfpathmoveto{\pgfqpoint{0.625853in}{1.516611in}}%
\pgfpathlineto{\pgfqpoint{4.662561in}{0.467414in}}%
\pgfpathlineto{\pgfqpoint{4.662561in}{0.467414in}}%
\pgfusepath{stroke}%
\end{pgfscope}%
\begin{pgfscope}%
\pgfpathrectangle{\pgfqpoint{0.625853in}{0.467414in}}{\pgfqpoint{4.036708in}{2.575301in}} %
\pgfusepath{clip}%
\pgfsetrectcap%
\pgfsetroundjoin%
\pgfsetlinewidth{1.003750pt}%
\definecolor{currentstroke}{rgb}{0.121569,0.466667,0.705882}%
\pgfsetstrokecolor{currentstroke}%
\pgfsetdash{}{0pt}%
\pgfpathmoveto{\pgfqpoint{0.625853in}{1.343190in}}%
\pgfpathlineto{\pgfqpoint{4.662561in}{0.467414in}}%
\pgfpathlineto{\pgfqpoint{4.662561in}{0.467414in}}%
\pgfusepath{stroke}%
\end{pgfscope}%
\begin{pgfscope}%
\pgfsetrectcap%
\pgfsetmiterjoin%
\pgfsetlinewidth{1.003750pt}%
\definecolor{currentstroke}{rgb}{0.000000,0.000000,0.000000}%
\pgfsetstrokecolor{currentstroke}%
\pgfsetdash{}{0pt}%
\pgfpathmoveto{\pgfqpoint{0.625853in}{3.042715in}}%
\pgfpathlineto{\pgfqpoint{4.662561in}{3.042715in}}%
\pgfusepath{stroke}%
\end{pgfscope}%
\begin{pgfscope}%
\pgfsetrectcap%
\pgfsetmiterjoin%
\pgfsetlinewidth{1.003750pt}%
\definecolor{currentstroke}{rgb}{0.000000,0.000000,0.000000}%
\pgfsetstrokecolor{currentstroke}%
\pgfsetdash{}{0pt}%
\pgfpathmoveto{\pgfqpoint{4.662561in}{0.467414in}}%
\pgfpathlineto{\pgfqpoint{4.662561in}{3.042715in}}%
\pgfusepath{stroke}%
\end{pgfscope}%
\begin{pgfscope}%
\pgfsetrectcap%
\pgfsetmiterjoin%
\pgfsetlinewidth{1.003750pt}%
\definecolor{currentstroke}{rgb}{0.000000,0.000000,0.000000}%
\pgfsetstrokecolor{currentstroke}%
\pgfsetdash{}{0pt}%
\pgfpathmoveto{\pgfqpoint{0.625853in}{0.467414in}}%
\pgfpathlineto{\pgfqpoint{4.662561in}{0.467414in}}%
\pgfusepath{stroke}%
\end{pgfscope}%
\begin{pgfscope}%
\pgfsetrectcap%
\pgfsetmiterjoin%
\pgfsetlinewidth{1.003750pt}%
\definecolor{currentstroke}{rgb}{0.000000,0.000000,0.000000}%
\pgfsetstrokecolor{currentstroke}%
\pgfsetdash{}{0pt}%
\pgfpathmoveto{\pgfqpoint{0.625853in}{0.467414in}}%
\pgfpathlineto{\pgfqpoint{0.625853in}{3.042715in}}%
\pgfusepath{stroke}%
\end{pgfscope}%
\begin{pgfscope}%
\pgfpathrectangle{\pgfqpoint{0.625853in}{0.467414in}}{\pgfqpoint{4.036708in}{2.575301in}} %
\pgfusepath{clip}%
\pgfsetbuttcap%
\pgfsetroundjoin%
\pgfsetlinewidth{0.501875pt}%
\definecolor{currentstroke}{rgb}{0.501961,0.501961,0.501961}%
\pgfsetstrokecolor{currentstroke}%
\pgfsetdash{{1.000000pt}{3.000000pt}}{0.000000pt}%
\pgfpathmoveto{\pgfqpoint{0.625853in}{0.467414in}}%
\pgfpathlineto{\pgfqpoint{0.625853in}{3.042715in}}%
\pgfusepath{stroke}%
\end{pgfscope}%
\begin{pgfscope}%
\pgfsetbuttcap%
\pgfsetroundjoin%
\definecolor{currentfill}{rgb}{0.000000,0.000000,0.000000}%
\pgfsetfillcolor{currentfill}%
\pgfsetlinewidth{0.501875pt}%
\definecolor{currentstroke}{rgb}{0.000000,0.000000,0.000000}%
\pgfsetstrokecolor{currentstroke}%
\pgfsetdash{}{0pt}%
\pgfsys@defobject{currentmarker}{\pgfqpoint{0.000000in}{0.000000in}}{\pgfqpoint{0.000000in}{0.055556in}}{%
\pgfpathmoveto{\pgfqpoint{0.000000in}{0.000000in}}%
\pgfpathlineto{\pgfqpoint{0.000000in}{0.055556in}}%
\pgfusepath{stroke,fill}%
}%
\begin{pgfscope}%
\pgfsys@transformshift{0.625853in}{0.467414in}%
\pgfsys@useobject{currentmarker}{}%
\end{pgfscope}%
\end{pgfscope}%
\begin{pgfscope}%
\pgfsetbuttcap%
\pgfsetroundjoin%
\definecolor{currentfill}{rgb}{0.000000,0.000000,0.000000}%
\pgfsetfillcolor{currentfill}%
\pgfsetlinewidth{0.501875pt}%
\definecolor{currentstroke}{rgb}{0.000000,0.000000,0.000000}%
\pgfsetstrokecolor{currentstroke}%
\pgfsetdash{}{0pt}%
\pgfsys@defobject{currentmarker}{\pgfqpoint{0.000000in}{-0.055556in}}{\pgfqpoint{0.000000in}{0.000000in}}{%
\pgfpathmoveto{\pgfqpoint{0.000000in}{0.000000in}}%
\pgfpathlineto{\pgfqpoint{0.000000in}{-0.055556in}}%
\pgfusepath{stroke,fill}%
}%
\begin{pgfscope}%
\pgfsys@transformshift{0.625853in}{3.042715in}%
\pgfsys@useobject{currentmarker}{}%
\end{pgfscope}%
\end{pgfscope}%
\begin{pgfscope}%
\pgftext[x=0.625853in,y=0.411859in,,top]{\rmfamily\fontsize{8.000000}{9.600000}\selectfont \(\displaystyle 0.0\)}%
\end{pgfscope}%
\begin{pgfscope}%
\pgfpathrectangle{\pgfqpoint{0.625853in}{0.467414in}}{\pgfqpoint{4.036708in}{2.575301in}} %
\pgfusepath{clip}%
\pgfsetbuttcap%
\pgfsetroundjoin%
\pgfsetlinewidth{0.501875pt}%
\definecolor{currentstroke}{rgb}{0.501961,0.501961,0.501961}%
\pgfsetstrokecolor{currentstroke}%
\pgfsetdash{{1.000000pt}{3.000000pt}}{0.000000pt}%
\pgfpathmoveto{\pgfqpoint{1.433195in}{0.467414in}}%
\pgfpathlineto{\pgfqpoint{1.433195in}{3.042715in}}%
\pgfusepath{stroke}%
\end{pgfscope}%
\begin{pgfscope}%
\pgfsetbuttcap%
\pgfsetroundjoin%
\definecolor{currentfill}{rgb}{0.000000,0.000000,0.000000}%
\pgfsetfillcolor{currentfill}%
\pgfsetlinewidth{0.501875pt}%
\definecolor{currentstroke}{rgb}{0.000000,0.000000,0.000000}%
\pgfsetstrokecolor{currentstroke}%
\pgfsetdash{}{0pt}%
\pgfsys@defobject{currentmarker}{\pgfqpoint{0.000000in}{0.000000in}}{\pgfqpoint{0.000000in}{0.055556in}}{%
\pgfpathmoveto{\pgfqpoint{0.000000in}{0.000000in}}%
\pgfpathlineto{\pgfqpoint{0.000000in}{0.055556in}}%
\pgfusepath{stroke,fill}%
}%
\begin{pgfscope}%
\pgfsys@transformshift{1.433195in}{0.467414in}%
\pgfsys@useobject{currentmarker}{}%
\end{pgfscope}%
\end{pgfscope}%
\begin{pgfscope}%
\pgfsetbuttcap%
\pgfsetroundjoin%
\definecolor{currentfill}{rgb}{0.000000,0.000000,0.000000}%
\pgfsetfillcolor{currentfill}%
\pgfsetlinewidth{0.501875pt}%
\definecolor{currentstroke}{rgb}{0.000000,0.000000,0.000000}%
\pgfsetstrokecolor{currentstroke}%
\pgfsetdash{}{0pt}%
\pgfsys@defobject{currentmarker}{\pgfqpoint{0.000000in}{-0.055556in}}{\pgfqpoint{0.000000in}{0.000000in}}{%
\pgfpathmoveto{\pgfqpoint{0.000000in}{0.000000in}}%
\pgfpathlineto{\pgfqpoint{0.000000in}{-0.055556in}}%
\pgfusepath{stroke,fill}%
}%
\begin{pgfscope}%
\pgfsys@transformshift{1.433195in}{3.042715in}%
\pgfsys@useobject{currentmarker}{}%
\end{pgfscope}%
\end{pgfscope}%
\begin{pgfscope}%
\pgftext[x=1.433195in,y=0.411859in,,top]{\rmfamily\fontsize{8.000000}{9.600000}\selectfont \(\displaystyle 0.2\)}%
\end{pgfscope}%
\begin{pgfscope}%
\pgfpathrectangle{\pgfqpoint{0.625853in}{0.467414in}}{\pgfqpoint{4.036708in}{2.575301in}} %
\pgfusepath{clip}%
\pgfsetbuttcap%
\pgfsetroundjoin%
\pgfsetlinewidth{0.501875pt}%
\definecolor{currentstroke}{rgb}{0.501961,0.501961,0.501961}%
\pgfsetstrokecolor{currentstroke}%
\pgfsetdash{{1.000000pt}{3.000000pt}}{0.000000pt}%
\pgfpathmoveto{\pgfqpoint{2.240536in}{0.467414in}}%
\pgfpathlineto{\pgfqpoint{2.240536in}{3.042715in}}%
\pgfusepath{stroke}%
\end{pgfscope}%
\begin{pgfscope}%
\pgfsetbuttcap%
\pgfsetroundjoin%
\definecolor{currentfill}{rgb}{0.000000,0.000000,0.000000}%
\pgfsetfillcolor{currentfill}%
\pgfsetlinewidth{0.501875pt}%
\definecolor{currentstroke}{rgb}{0.000000,0.000000,0.000000}%
\pgfsetstrokecolor{currentstroke}%
\pgfsetdash{}{0pt}%
\pgfsys@defobject{currentmarker}{\pgfqpoint{0.000000in}{0.000000in}}{\pgfqpoint{0.000000in}{0.055556in}}{%
\pgfpathmoveto{\pgfqpoint{0.000000in}{0.000000in}}%
\pgfpathlineto{\pgfqpoint{0.000000in}{0.055556in}}%
\pgfusepath{stroke,fill}%
}%
\begin{pgfscope}%
\pgfsys@transformshift{2.240536in}{0.467414in}%
\pgfsys@useobject{currentmarker}{}%
\end{pgfscope}%
\end{pgfscope}%
\begin{pgfscope}%
\pgfsetbuttcap%
\pgfsetroundjoin%
\definecolor{currentfill}{rgb}{0.000000,0.000000,0.000000}%
\pgfsetfillcolor{currentfill}%
\pgfsetlinewidth{0.501875pt}%
\definecolor{currentstroke}{rgb}{0.000000,0.000000,0.000000}%
\pgfsetstrokecolor{currentstroke}%
\pgfsetdash{}{0pt}%
\pgfsys@defobject{currentmarker}{\pgfqpoint{0.000000in}{-0.055556in}}{\pgfqpoint{0.000000in}{0.000000in}}{%
\pgfpathmoveto{\pgfqpoint{0.000000in}{0.000000in}}%
\pgfpathlineto{\pgfqpoint{0.000000in}{-0.055556in}}%
\pgfusepath{stroke,fill}%
}%
\begin{pgfscope}%
\pgfsys@transformshift{2.240536in}{3.042715in}%
\pgfsys@useobject{currentmarker}{}%
\end{pgfscope}%
\end{pgfscope}%
\begin{pgfscope}%
\pgftext[x=2.240536in,y=0.411859in,,top]{\rmfamily\fontsize{8.000000}{9.600000}\selectfont \(\displaystyle 0.4\)}%
\end{pgfscope}%
\begin{pgfscope}%
\pgfpathrectangle{\pgfqpoint{0.625853in}{0.467414in}}{\pgfqpoint{4.036708in}{2.575301in}} %
\pgfusepath{clip}%
\pgfsetbuttcap%
\pgfsetroundjoin%
\pgfsetlinewidth{0.501875pt}%
\definecolor{currentstroke}{rgb}{0.501961,0.501961,0.501961}%
\pgfsetstrokecolor{currentstroke}%
\pgfsetdash{{1.000000pt}{3.000000pt}}{0.000000pt}%
\pgfpathmoveto{\pgfqpoint{3.047878in}{0.467414in}}%
\pgfpathlineto{\pgfqpoint{3.047878in}{3.042715in}}%
\pgfusepath{stroke}%
\end{pgfscope}%
\begin{pgfscope}%
\pgfsetbuttcap%
\pgfsetroundjoin%
\definecolor{currentfill}{rgb}{0.000000,0.000000,0.000000}%
\pgfsetfillcolor{currentfill}%
\pgfsetlinewidth{0.501875pt}%
\definecolor{currentstroke}{rgb}{0.000000,0.000000,0.000000}%
\pgfsetstrokecolor{currentstroke}%
\pgfsetdash{}{0pt}%
\pgfsys@defobject{currentmarker}{\pgfqpoint{0.000000in}{0.000000in}}{\pgfqpoint{0.000000in}{0.055556in}}{%
\pgfpathmoveto{\pgfqpoint{0.000000in}{0.000000in}}%
\pgfpathlineto{\pgfqpoint{0.000000in}{0.055556in}}%
\pgfusepath{stroke,fill}%
}%
\begin{pgfscope}%
\pgfsys@transformshift{3.047878in}{0.467414in}%
\pgfsys@useobject{currentmarker}{}%
\end{pgfscope}%
\end{pgfscope}%
\begin{pgfscope}%
\pgfsetbuttcap%
\pgfsetroundjoin%
\definecolor{currentfill}{rgb}{0.000000,0.000000,0.000000}%
\pgfsetfillcolor{currentfill}%
\pgfsetlinewidth{0.501875pt}%
\definecolor{currentstroke}{rgb}{0.000000,0.000000,0.000000}%
\pgfsetstrokecolor{currentstroke}%
\pgfsetdash{}{0pt}%
\pgfsys@defobject{currentmarker}{\pgfqpoint{0.000000in}{-0.055556in}}{\pgfqpoint{0.000000in}{0.000000in}}{%
\pgfpathmoveto{\pgfqpoint{0.000000in}{0.000000in}}%
\pgfpathlineto{\pgfqpoint{0.000000in}{-0.055556in}}%
\pgfusepath{stroke,fill}%
}%
\begin{pgfscope}%
\pgfsys@transformshift{3.047878in}{3.042715in}%
\pgfsys@useobject{currentmarker}{}%
\end{pgfscope}%
\end{pgfscope}%
\begin{pgfscope}%
\pgftext[x=3.047878in,y=0.411859in,,top]{\rmfamily\fontsize{8.000000}{9.600000}\selectfont \(\displaystyle 0.6\)}%
\end{pgfscope}%
\begin{pgfscope}%
\pgfpathrectangle{\pgfqpoint{0.625853in}{0.467414in}}{\pgfqpoint{4.036708in}{2.575301in}} %
\pgfusepath{clip}%
\pgfsetbuttcap%
\pgfsetroundjoin%
\pgfsetlinewidth{0.501875pt}%
\definecolor{currentstroke}{rgb}{0.501961,0.501961,0.501961}%
\pgfsetstrokecolor{currentstroke}%
\pgfsetdash{{1.000000pt}{3.000000pt}}{0.000000pt}%
\pgfpathmoveto{\pgfqpoint{3.855219in}{0.467414in}}%
\pgfpathlineto{\pgfqpoint{3.855219in}{3.042715in}}%
\pgfusepath{stroke}%
\end{pgfscope}%
\begin{pgfscope}%
\pgfsetbuttcap%
\pgfsetroundjoin%
\definecolor{currentfill}{rgb}{0.000000,0.000000,0.000000}%
\pgfsetfillcolor{currentfill}%
\pgfsetlinewidth{0.501875pt}%
\definecolor{currentstroke}{rgb}{0.000000,0.000000,0.000000}%
\pgfsetstrokecolor{currentstroke}%
\pgfsetdash{}{0pt}%
\pgfsys@defobject{currentmarker}{\pgfqpoint{0.000000in}{0.000000in}}{\pgfqpoint{0.000000in}{0.055556in}}{%
\pgfpathmoveto{\pgfqpoint{0.000000in}{0.000000in}}%
\pgfpathlineto{\pgfqpoint{0.000000in}{0.055556in}}%
\pgfusepath{stroke,fill}%
}%
\begin{pgfscope}%
\pgfsys@transformshift{3.855219in}{0.467414in}%
\pgfsys@useobject{currentmarker}{}%
\end{pgfscope}%
\end{pgfscope}%
\begin{pgfscope}%
\pgfsetbuttcap%
\pgfsetroundjoin%
\definecolor{currentfill}{rgb}{0.000000,0.000000,0.000000}%
\pgfsetfillcolor{currentfill}%
\pgfsetlinewidth{0.501875pt}%
\definecolor{currentstroke}{rgb}{0.000000,0.000000,0.000000}%
\pgfsetstrokecolor{currentstroke}%
\pgfsetdash{}{0pt}%
\pgfsys@defobject{currentmarker}{\pgfqpoint{0.000000in}{-0.055556in}}{\pgfqpoint{0.000000in}{0.000000in}}{%
\pgfpathmoveto{\pgfqpoint{0.000000in}{0.000000in}}%
\pgfpathlineto{\pgfqpoint{0.000000in}{-0.055556in}}%
\pgfusepath{stroke,fill}%
}%
\begin{pgfscope}%
\pgfsys@transformshift{3.855219in}{3.042715in}%
\pgfsys@useobject{currentmarker}{}%
\end{pgfscope}%
\end{pgfscope}%
\begin{pgfscope}%
\pgftext[x=3.855219in,y=0.411859in,,top]{\rmfamily\fontsize{8.000000}{9.600000}\selectfont \(\displaystyle 0.8\)}%
\end{pgfscope}%
\begin{pgfscope}%
\pgfpathrectangle{\pgfqpoint{0.625853in}{0.467414in}}{\pgfqpoint{4.036708in}{2.575301in}} %
\pgfusepath{clip}%
\pgfsetbuttcap%
\pgfsetroundjoin%
\pgfsetlinewidth{0.501875pt}%
\definecolor{currentstroke}{rgb}{0.501961,0.501961,0.501961}%
\pgfsetstrokecolor{currentstroke}%
\pgfsetdash{{1.000000pt}{3.000000pt}}{0.000000pt}%
\pgfpathmoveto{\pgfqpoint{4.662561in}{0.467414in}}%
\pgfpathlineto{\pgfqpoint{4.662561in}{3.042715in}}%
\pgfusepath{stroke}%
\end{pgfscope}%
\begin{pgfscope}%
\pgfsetbuttcap%
\pgfsetroundjoin%
\definecolor{currentfill}{rgb}{0.000000,0.000000,0.000000}%
\pgfsetfillcolor{currentfill}%
\pgfsetlinewidth{0.501875pt}%
\definecolor{currentstroke}{rgb}{0.000000,0.000000,0.000000}%
\pgfsetstrokecolor{currentstroke}%
\pgfsetdash{}{0pt}%
\pgfsys@defobject{currentmarker}{\pgfqpoint{0.000000in}{0.000000in}}{\pgfqpoint{0.000000in}{0.055556in}}{%
\pgfpathmoveto{\pgfqpoint{0.000000in}{0.000000in}}%
\pgfpathlineto{\pgfqpoint{0.000000in}{0.055556in}}%
\pgfusepath{stroke,fill}%
}%
\begin{pgfscope}%
\pgfsys@transformshift{4.662561in}{0.467414in}%
\pgfsys@useobject{currentmarker}{}%
\end{pgfscope}%
\end{pgfscope}%
\begin{pgfscope}%
\pgfsetbuttcap%
\pgfsetroundjoin%
\definecolor{currentfill}{rgb}{0.000000,0.000000,0.000000}%
\pgfsetfillcolor{currentfill}%
\pgfsetlinewidth{0.501875pt}%
\definecolor{currentstroke}{rgb}{0.000000,0.000000,0.000000}%
\pgfsetstrokecolor{currentstroke}%
\pgfsetdash{}{0pt}%
\pgfsys@defobject{currentmarker}{\pgfqpoint{0.000000in}{-0.055556in}}{\pgfqpoint{0.000000in}{0.000000in}}{%
\pgfpathmoveto{\pgfqpoint{0.000000in}{0.000000in}}%
\pgfpathlineto{\pgfqpoint{0.000000in}{-0.055556in}}%
\pgfusepath{stroke,fill}%
}%
\begin{pgfscope}%
\pgfsys@transformshift{4.662561in}{3.042715in}%
\pgfsys@useobject{currentmarker}{}%
\end{pgfscope}%
\end{pgfscope}%
\begin{pgfscope}%
\pgftext[x=4.662561in,y=0.411859in,,top]{\rmfamily\fontsize{8.000000}{9.600000}\selectfont \(\displaystyle 1.0\)}%
\end{pgfscope}%
\begin{pgfscope}%
\pgftext[x=2.644207in,y=0.244290in,,top]{\rmfamily\fontsize{10.000000}{12.000000}\selectfont \(\displaystyle \mathcal{T}_{A|B}\)}%
\end{pgfscope}%
\begin{pgfscope}%
\pgfpathrectangle{\pgfqpoint{0.625853in}{0.467414in}}{\pgfqpoint{4.036708in}{2.575301in}} %
\pgfusepath{clip}%
\pgfsetbuttcap%
\pgfsetroundjoin%
\pgfsetlinewidth{0.501875pt}%
\definecolor{currentstroke}{rgb}{0.501961,0.501961,0.501961}%
\pgfsetstrokecolor{currentstroke}%
\pgfsetdash{{1.000000pt}{3.000000pt}}{0.000000pt}%
\pgfpathmoveto{\pgfqpoint{0.625853in}{0.467414in}}%
\pgfpathlineto{\pgfqpoint{4.662561in}{0.467414in}}%
\pgfusepath{stroke}%
\end{pgfscope}%
\begin{pgfscope}%
\pgfsetbuttcap%
\pgfsetroundjoin%
\definecolor{currentfill}{rgb}{0.000000,0.000000,0.000000}%
\pgfsetfillcolor{currentfill}%
\pgfsetlinewidth{0.501875pt}%
\definecolor{currentstroke}{rgb}{0.000000,0.000000,0.000000}%
\pgfsetstrokecolor{currentstroke}%
\pgfsetdash{}{0pt}%
\pgfsys@defobject{currentmarker}{\pgfqpoint{0.000000in}{0.000000in}}{\pgfqpoint{0.055556in}{0.000000in}}{%
\pgfpathmoveto{\pgfqpoint{0.000000in}{0.000000in}}%
\pgfpathlineto{\pgfqpoint{0.055556in}{0.000000in}}%
\pgfusepath{stroke,fill}%
}%
\begin{pgfscope}%
\pgfsys@transformshift{0.625853in}{0.467414in}%
\pgfsys@useobject{currentmarker}{}%
\end{pgfscope}%
\end{pgfscope}%
\begin{pgfscope}%
\pgfsetbuttcap%
\pgfsetroundjoin%
\definecolor{currentfill}{rgb}{0.000000,0.000000,0.000000}%
\pgfsetfillcolor{currentfill}%
\pgfsetlinewidth{0.501875pt}%
\definecolor{currentstroke}{rgb}{0.000000,0.000000,0.000000}%
\pgfsetstrokecolor{currentstroke}%
\pgfsetdash{}{0pt}%
\pgfsys@defobject{currentmarker}{\pgfqpoint{-0.055556in}{0.000000in}}{\pgfqpoint{0.000000in}{0.000000in}}{%
\pgfpathmoveto{\pgfqpoint{0.000000in}{0.000000in}}%
\pgfpathlineto{\pgfqpoint{-0.055556in}{0.000000in}}%
\pgfusepath{stroke,fill}%
}%
\begin{pgfscope}%
\pgfsys@transformshift{4.662561in}{0.467414in}%
\pgfsys@useobject{currentmarker}{}%
\end{pgfscope}%
\end{pgfscope}%
\begin{pgfscope}%
\pgftext[x=0.570298in,y=0.467414in,right,]{\rmfamily\fontsize{8.000000}{9.600000}\selectfont \(\displaystyle 0.0\)}%
\end{pgfscope}%
\begin{pgfscope}%
\pgfpathrectangle{\pgfqpoint{0.625853in}{0.467414in}}{\pgfqpoint{4.036708in}{2.575301in}} %
\pgfusepath{clip}%
\pgfsetbuttcap%
\pgfsetroundjoin%
\pgfsetlinewidth{0.501875pt}%
\definecolor{currentstroke}{rgb}{0.501961,0.501961,0.501961}%
\pgfsetstrokecolor{currentstroke}%
\pgfsetdash{{1.000000pt}{3.000000pt}}{0.000000pt}%
\pgfpathmoveto{\pgfqpoint{0.625853in}{0.896631in}}%
\pgfpathlineto{\pgfqpoint{4.662561in}{0.896631in}}%
\pgfusepath{stroke}%
\end{pgfscope}%
\begin{pgfscope}%
\pgfsetbuttcap%
\pgfsetroundjoin%
\definecolor{currentfill}{rgb}{0.000000,0.000000,0.000000}%
\pgfsetfillcolor{currentfill}%
\pgfsetlinewidth{0.501875pt}%
\definecolor{currentstroke}{rgb}{0.000000,0.000000,0.000000}%
\pgfsetstrokecolor{currentstroke}%
\pgfsetdash{}{0pt}%
\pgfsys@defobject{currentmarker}{\pgfqpoint{0.000000in}{0.000000in}}{\pgfqpoint{0.055556in}{0.000000in}}{%
\pgfpathmoveto{\pgfqpoint{0.000000in}{0.000000in}}%
\pgfpathlineto{\pgfqpoint{0.055556in}{0.000000in}}%
\pgfusepath{stroke,fill}%
}%
\begin{pgfscope}%
\pgfsys@transformshift{0.625853in}{0.896631in}%
\pgfsys@useobject{currentmarker}{}%
\end{pgfscope}%
\end{pgfscope}%
\begin{pgfscope}%
\pgfsetbuttcap%
\pgfsetroundjoin%
\definecolor{currentfill}{rgb}{0.000000,0.000000,0.000000}%
\pgfsetfillcolor{currentfill}%
\pgfsetlinewidth{0.501875pt}%
\definecolor{currentstroke}{rgb}{0.000000,0.000000,0.000000}%
\pgfsetstrokecolor{currentstroke}%
\pgfsetdash{}{0pt}%
\pgfsys@defobject{currentmarker}{\pgfqpoint{-0.055556in}{0.000000in}}{\pgfqpoint{0.000000in}{0.000000in}}{%
\pgfpathmoveto{\pgfqpoint{0.000000in}{0.000000in}}%
\pgfpathlineto{\pgfqpoint{-0.055556in}{0.000000in}}%
\pgfusepath{stroke,fill}%
}%
\begin{pgfscope}%
\pgfsys@transformshift{4.662561in}{0.896631in}%
\pgfsys@useobject{currentmarker}{}%
\end{pgfscope}%
\end{pgfscope}%
\begin{pgfscope}%
\pgftext[x=0.570298in,y=0.896631in,right,]{\rmfamily\fontsize{8.000000}{9.600000}\selectfont \(\displaystyle 0.5\)}%
\end{pgfscope}%
\begin{pgfscope}%
\pgfpathrectangle{\pgfqpoint{0.625853in}{0.467414in}}{\pgfqpoint{4.036708in}{2.575301in}} %
\pgfusepath{clip}%
\pgfsetbuttcap%
\pgfsetroundjoin%
\pgfsetlinewidth{0.501875pt}%
\definecolor{currentstroke}{rgb}{0.501961,0.501961,0.501961}%
\pgfsetstrokecolor{currentstroke}%
\pgfsetdash{{1.000000pt}{3.000000pt}}{0.000000pt}%
\pgfpathmoveto{\pgfqpoint{0.625853in}{1.325848in}}%
\pgfpathlineto{\pgfqpoint{4.662561in}{1.325848in}}%
\pgfusepath{stroke}%
\end{pgfscope}%
\begin{pgfscope}%
\pgfsetbuttcap%
\pgfsetroundjoin%
\definecolor{currentfill}{rgb}{0.000000,0.000000,0.000000}%
\pgfsetfillcolor{currentfill}%
\pgfsetlinewidth{0.501875pt}%
\definecolor{currentstroke}{rgb}{0.000000,0.000000,0.000000}%
\pgfsetstrokecolor{currentstroke}%
\pgfsetdash{}{0pt}%
\pgfsys@defobject{currentmarker}{\pgfqpoint{0.000000in}{0.000000in}}{\pgfqpoint{0.055556in}{0.000000in}}{%
\pgfpathmoveto{\pgfqpoint{0.000000in}{0.000000in}}%
\pgfpathlineto{\pgfqpoint{0.055556in}{0.000000in}}%
\pgfusepath{stroke,fill}%
}%
\begin{pgfscope}%
\pgfsys@transformshift{0.625853in}{1.325848in}%
\pgfsys@useobject{currentmarker}{}%
\end{pgfscope}%
\end{pgfscope}%
\begin{pgfscope}%
\pgfsetbuttcap%
\pgfsetroundjoin%
\definecolor{currentfill}{rgb}{0.000000,0.000000,0.000000}%
\pgfsetfillcolor{currentfill}%
\pgfsetlinewidth{0.501875pt}%
\definecolor{currentstroke}{rgb}{0.000000,0.000000,0.000000}%
\pgfsetstrokecolor{currentstroke}%
\pgfsetdash{}{0pt}%
\pgfsys@defobject{currentmarker}{\pgfqpoint{-0.055556in}{0.000000in}}{\pgfqpoint{0.000000in}{0.000000in}}{%
\pgfpathmoveto{\pgfqpoint{0.000000in}{0.000000in}}%
\pgfpathlineto{\pgfqpoint{-0.055556in}{0.000000in}}%
\pgfusepath{stroke,fill}%
}%
\begin{pgfscope}%
\pgfsys@transformshift{4.662561in}{1.325848in}%
\pgfsys@useobject{currentmarker}{}%
\end{pgfscope}%
\end{pgfscope}%
\begin{pgfscope}%
\pgftext[x=0.570298in,y=1.325848in,right,]{\rmfamily\fontsize{8.000000}{9.600000}\selectfont \(\displaystyle 1.0\)}%
\end{pgfscope}%
\begin{pgfscope}%
\pgfpathrectangle{\pgfqpoint{0.625853in}{0.467414in}}{\pgfqpoint{4.036708in}{2.575301in}} %
\pgfusepath{clip}%
\pgfsetbuttcap%
\pgfsetroundjoin%
\pgfsetlinewidth{0.501875pt}%
\definecolor{currentstroke}{rgb}{0.501961,0.501961,0.501961}%
\pgfsetstrokecolor{currentstroke}%
\pgfsetdash{{1.000000pt}{3.000000pt}}{0.000000pt}%
\pgfpathmoveto{\pgfqpoint{0.625853in}{1.755065in}}%
\pgfpathlineto{\pgfqpoint{4.662561in}{1.755065in}}%
\pgfusepath{stroke}%
\end{pgfscope}%
\begin{pgfscope}%
\pgfsetbuttcap%
\pgfsetroundjoin%
\definecolor{currentfill}{rgb}{0.000000,0.000000,0.000000}%
\pgfsetfillcolor{currentfill}%
\pgfsetlinewidth{0.501875pt}%
\definecolor{currentstroke}{rgb}{0.000000,0.000000,0.000000}%
\pgfsetstrokecolor{currentstroke}%
\pgfsetdash{}{0pt}%
\pgfsys@defobject{currentmarker}{\pgfqpoint{0.000000in}{0.000000in}}{\pgfqpoint{0.055556in}{0.000000in}}{%
\pgfpathmoveto{\pgfqpoint{0.000000in}{0.000000in}}%
\pgfpathlineto{\pgfqpoint{0.055556in}{0.000000in}}%
\pgfusepath{stroke,fill}%
}%
\begin{pgfscope}%
\pgfsys@transformshift{0.625853in}{1.755065in}%
\pgfsys@useobject{currentmarker}{}%
\end{pgfscope}%
\end{pgfscope}%
\begin{pgfscope}%
\pgfsetbuttcap%
\pgfsetroundjoin%
\definecolor{currentfill}{rgb}{0.000000,0.000000,0.000000}%
\pgfsetfillcolor{currentfill}%
\pgfsetlinewidth{0.501875pt}%
\definecolor{currentstroke}{rgb}{0.000000,0.000000,0.000000}%
\pgfsetstrokecolor{currentstroke}%
\pgfsetdash{}{0pt}%
\pgfsys@defobject{currentmarker}{\pgfqpoint{-0.055556in}{0.000000in}}{\pgfqpoint{0.000000in}{0.000000in}}{%
\pgfpathmoveto{\pgfqpoint{0.000000in}{0.000000in}}%
\pgfpathlineto{\pgfqpoint{-0.055556in}{0.000000in}}%
\pgfusepath{stroke,fill}%
}%
\begin{pgfscope}%
\pgfsys@transformshift{4.662561in}{1.755065in}%
\pgfsys@useobject{currentmarker}{}%
\end{pgfscope}%
\end{pgfscope}%
\begin{pgfscope}%
\pgftext[x=0.570298in,y=1.755065in,right,]{\rmfamily\fontsize{8.000000}{9.600000}\selectfont \(\displaystyle 1.5\)}%
\end{pgfscope}%
\begin{pgfscope}%
\pgfpathrectangle{\pgfqpoint{0.625853in}{0.467414in}}{\pgfqpoint{4.036708in}{2.575301in}} %
\pgfusepath{clip}%
\pgfsetbuttcap%
\pgfsetroundjoin%
\pgfsetlinewidth{0.501875pt}%
\definecolor{currentstroke}{rgb}{0.501961,0.501961,0.501961}%
\pgfsetstrokecolor{currentstroke}%
\pgfsetdash{{1.000000pt}{3.000000pt}}{0.000000pt}%
\pgfpathmoveto{\pgfqpoint{0.625853in}{2.184281in}}%
\pgfpathlineto{\pgfqpoint{4.662561in}{2.184281in}}%
\pgfusepath{stroke}%
\end{pgfscope}%
\begin{pgfscope}%
\pgfsetbuttcap%
\pgfsetroundjoin%
\definecolor{currentfill}{rgb}{0.000000,0.000000,0.000000}%
\pgfsetfillcolor{currentfill}%
\pgfsetlinewidth{0.501875pt}%
\definecolor{currentstroke}{rgb}{0.000000,0.000000,0.000000}%
\pgfsetstrokecolor{currentstroke}%
\pgfsetdash{}{0pt}%
\pgfsys@defobject{currentmarker}{\pgfqpoint{0.000000in}{0.000000in}}{\pgfqpoint{0.055556in}{0.000000in}}{%
\pgfpathmoveto{\pgfqpoint{0.000000in}{0.000000in}}%
\pgfpathlineto{\pgfqpoint{0.055556in}{0.000000in}}%
\pgfusepath{stroke,fill}%
}%
\begin{pgfscope}%
\pgfsys@transformshift{0.625853in}{2.184281in}%
\pgfsys@useobject{currentmarker}{}%
\end{pgfscope}%
\end{pgfscope}%
\begin{pgfscope}%
\pgfsetbuttcap%
\pgfsetroundjoin%
\definecolor{currentfill}{rgb}{0.000000,0.000000,0.000000}%
\pgfsetfillcolor{currentfill}%
\pgfsetlinewidth{0.501875pt}%
\definecolor{currentstroke}{rgb}{0.000000,0.000000,0.000000}%
\pgfsetstrokecolor{currentstroke}%
\pgfsetdash{}{0pt}%
\pgfsys@defobject{currentmarker}{\pgfqpoint{-0.055556in}{0.000000in}}{\pgfqpoint{0.000000in}{0.000000in}}{%
\pgfpathmoveto{\pgfqpoint{0.000000in}{0.000000in}}%
\pgfpathlineto{\pgfqpoint{-0.055556in}{0.000000in}}%
\pgfusepath{stroke,fill}%
}%
\begin{pgfscope}%
\pgfsys@transformshift{4.662561in}{2.184281in}%
\pgfsys@useobject{currentmarker}{}%
\end{pgfscope}%
\end{pgfscope}%
\begin{pgfscope}%
\pgftext[x=0.570298in,y=2.184281in,right,]{\rmfamily\fontsize{8.000000}{9.600000}\selectfont \(\displaystyle 2.0\)}%
\end{pgfscope}%
\begin{pgfscope}%
\pgfpathrectangle{\pgfqpoint{0.625853in}{0.467414in}}{\pgfqpoint{4.036708in}{2.575301in}} %
\pgfusepath{clip}%
\pgfsetbuttcap%
\pgfsetroundjoin%
\pgfsetlinewidth{0.501875pt}%
\definecolor{currentstroke}{rgb}{0.501961,0.501961,0.501961}%
\pgfsetstrokecolor{currentstroke}%
\pgfsetdash{{1.000000pt}{3.000000pt}}{0.000000pt}%
\pgfpathmoveto{\pgfqpoint{0.625853in}{2.613498in}}%
\pgfpathlineto{\pgfqpoint{4.662561in}{2.613498in}}%
\pgfusepath{stroke}%
\end{pgfscope}%
\begin{pgfscope}%
\pgfsetbuttcap%
\pgfsetroundjoin%
\definecolor{currentfill}{rgb}{0.000000,0.000000,0.000000}%
\pgfsetfillcolor{currentfill}%
\pgfsetlinewidth{0.501875pt}%
\definecolor{currentstroke}{rgb}{0.000000,0.000000,0.000000}%
\pgfsetstrokecolor{currentstroke}%
\pgfsetdash{}{0pt}%
\pgfsys@defobject{currentmarker}{\pgfqpoint{0.000000in}{0.000000in}}{\pgfqpoint{0.055556in}{0.000000in}}{%
\pgfpathmoveto{\pgfqpoint{0.000000in}{0.000000in}}%
\pgfpathlineto{\pgfqpoint{0.055556in}{0.000000in}}%
\pgfusepath{stroke,fill}%
}%
\begin{pgfscope}%
\pgfsys@transformshift{0.625853in}{2.613498in}%
\pgfsys@useobject{currentmarker}{}%
\end{pgfscope}%
\end{pgfscope}%
\begin{pgfscope}%
\pgfsetbuttcap%
\pgfsetroundjoin%
\definecolor{currentfill}{rgb}{0.000000,0.000000,0.000000}%
\pgfsetfillcolor{currentfill}%
\pgfsetlinewidth{0.501875pt}%
\definecolor{currentstroke}{rgb}{0.000000,0.000000,0.000000}%
\pgfsetstrokecolor{currentstroke}%
\pgfsetdash{}{0pt}%
\pgfsys@defobject{currentmarker}{\pgfqpoint{-0.055556in}{0.000000in}}{\pgfqpoint{0.000000in}{0.000000in}}{%
\pgfpathmoveto{\pgfqpoint{0.000000in}{0.000000in}}%
\pgfpathlineto{\pgfqpoint{-0.055556in}{0.000000in}}%
\pgfusepath{stroke,fill}%
}%
\begin{pgfscope}%
\pgfsys@transformshift{4.662561in}{2.613498in}%
\pgfsys@useobject{currentmarker}{}%
\end{pgfscope}%
\end{pgfscope}%
\begin{pgfscope}%
\pgftext[x=0.570298in,y=2.613498in,right,]{\rmfamily\fontsize{8.000000}{9.600000}\selectfont \(\displaystyle 2.5\)}%
\end{pgfscope}%
\begin{pgfscope}%
\pgfpathrectangle{\pgfqpoint{0.625853in}{0.467414in}}{\pgfqpoint{4.036708in}{2.575301in}} %
\pgfusepath{clip}%
\pgfsetbuttcap%
\pgfsetroundjoin%
\pgfsetlinewidth{0.501875pt}%
\definecolor{currentstroke}{rgb}{0.501961,0.501961,0.501961}%
\pgfsetstrokecolor{currentstroke}%
\pgfsetdash{{1.000000pt}{3.000000pt}}{0.000000pt}%
\pgfpathmoveto{\pgfqpoint{0.625853in}{3.042715in}}%
\pgfpathlineto{\pgfqpoint{4.662561in}{3.042715in}}%
\pgfusepath{stroke}%
\end{pgfscope}%
\begin{pgfscope}%
\pgfsetbuttcap%
\pgfsetroundjoin%
\definecolor{currentfill}{rgb}{0.000000,0.000000,0.000000}%
\pgfsetfillcolor{currentfill}%
\pgfsetlinewidth{0.501875pt}%
\definecolor{currentstroke}{rgb}{0.000000,0.000000,0.000000}%
\pgfsetstrokecolor{currentstroke}%
\pgfsetdash{}{0pt}%
\pgfsys@defobject{currentmarker}{\pgfqpoint{0.000000in}{0.000000in}}{\pgfqpoint{0.055556in}{0.000000in}}{%
\pgfpathmoveto{\pgfqpoint{0.000000in}{0.000000in}}%
\pgfpathlineto{\pgfqpoint{0.055556in}{0.000000in}}%
\pgfusepath{stroke,fill}%
}%
\begin{pgfscope}%
\pgfsys@transformshift{0.625853in}{3.042715in}%
\pgfsys@useobject{currentmarker}{}%
\end{pgfscope}%
\end{pgfscope}%
\begin{pgfscope}%
\pgfsetbuttcap%
\pgfsetroundjoin%
\definecolor{currentfill}{rgb}{0.000000,0.000000,0.000000}%
\pgfsetfillcolor{currentfill}%
\pgfsetlinewidth{0.501875pt}%
\definecolor{currentstroke}{rgb}{0.000000,0.000000,0.000000}%
\pgfsetstrokecolor{currentstroke}%
\pgfsetdash{}{0pt}%
\pgfsys@defobject{currentmarker}{\pgfqpoint{-0.055556in}{0.000000in}}{\pgfqpoint{0.000000in}{0.000000in}}{%
\pgfpathmoveto{\pgfqpoint{0.000000in}{0.000000in}}%
\pgfpathlineto{\pgfqpoint{-0.055556in}{0.000000in}}%
\pgfusepath{stroke,fill}%
}%
\begin{pgfscope}%
\pgfsys@transformshift{4.662561in}{3.042715in}%
\pgfsys@useobject{currentmarker}{}%
\end{pgfscope}%
\end{pgfscope}%
\begin{pgfscope}%
\pgftext[x=0.570298in,y=3.042715in,right,]{\rmfamily\fontsize{8.000000}{9.600000}\selectfont \(\displaystyle 3.0\)}%
\end{pgfscope}%
\begin{pgfscope}%
\pgftext[x=0.350002in,y=1.755065in,,bottom,rotate=90.000000]{\rmfamily\fontsize{10.000000}{12.000000}\selectfont \(\displaystyle \max\left(\mathcal{T}_{AB|E}^{excess}\right)\)}%
\end{pgfscope}%
\begin{pgfscope}%
\pgfsetbuttcap%
\pgfsetmiterjoin%
\definecolor{currentfill}{rgb}{1.000000,1.000000,1.000000}%
\pgfsetfillcolor{currentfill}%
\pgfsetlinewidth{1.003750pt}%
\definecolor{currentstroke}{rgb}{0.000000,0.000000,0.000000}%
\pgfsetstrokecolor{currentstroke}%
\pgfsetdash{}{0pt}%
\pgfpathmoveto{\pgfqpoint{3.811765in}{2.179161in}}%
\pgfpathlineto{\pgfqpoint{4.607005in}{2.179161in}}%
\pgfpathlineto{\pgfqpoint{4.607005in}{2.987159in}}%
\pgfpathlineto{\pgfqpoint{3.811765in}{2.987159in}}%
\pgfpathclose%
\pgfusepath{stroke,fill}%
\end{pgfscope}%
\begin{pgfscope}%
\pgfsetrectcap%
\pgfsetroundjoin%
\pgfsetlinewidth{1.003750pt}%
\definecolor{currentstroke}{rgb}{1.000000,0.498039,0.054902}%
\pgfsetstrokecolor{currentstroke}%
\pgfsetdash{}{0pt}%
\pgfpathmoveto{\pgfqpoint{3.889542in}{2.903826in}}%
\pgfpathlineto{\pgfqpoint{4.045098in}{2.903826in}}%
\pgfusepath{stroke}%
\end{pgfscope}%
\begin{pgfscope}%
\pgftext[x=4.167320in,y=2.864937in,left,base]{\rmfamily\fontsize{8.000000}{9.600000}\selectfont \(\displaystyle d=2\)}%
\end{pgfscope}%
\begin{pgfscope}%
\pgfsetrectcap%
\pgfsetroundjoin%
\pgfsetlinewidth{1.003750pt}%
\definecolor{currentstroke}{rgb}{0.839216,0.152941,0.156863}%
\pgfsetstrokecolor{currentstroke}%
\pgfsetdash{}{0pt}%
\pgfpathmoveto{\pgfqpoint{3.889542in}{2.748893in}}%
\pgfpathlineto{\pgfqpoint{4.045098in}{2.748893in}}%
\pgfusepath{stroke}%
\end{pgfscope}%
\begin{pgfscope}%
\pgftext[x=4.167320in,y=2.710004in,left,base]{\rmfamily\fontsize{8.000000}{9.600000}\selectfont \(\displaystyle d=3\)}%
\end{pgfscope}%
\begin{pgfscope}%
\pgfsetrectcap%
\pgfsetroundjoin%
\pgfsetlinewidth{1.003750pt}%
\definecolor{currentstroke}{rgb}{0.580392,0.403922,0.741176}%
\pgfsetstrokecolor{currentstroke}%
\pgfsetdash{}{0pt}%
\pgfpathmoveto{\pgfqpoint{3.889542in}{2.593960in}}%
\pgfpathlineto{\pgfqpoint{4.045098in}{2.593960in}}%
\pgfusepath{stroke}%
\end{pgfscope}%
\begin{pgfscope}%
\pgftext[x=4.167320in,y=2.555071in,left,base]{\rmfamily\fontsize{8.000000}{9.600000}\selectfont \(\displaystyle d=5\)}%
\end{pgfscope}%
\begin{pgfscope}%
\pgfsetrectcap%
\pgfsetroundjoin%
\pgfsetlinewidth{1.003750pt}%
\definecolor{currentstroke}{rgb}{0.172549,0.627451,0.172549}%
\pgfsetstrokecolor{currentstroke}%
\pgfsetdash{}{0pt}%
\pgfpathmoveto{\pgfqpoint{3.889542in}{2.439027in}}%
\pgfpathlineto{\pgfqpoint{4.045098in}{2.439027in}}%
\pgfusepath{stroke}%
\end{pgfscope}%
\begin{pgfscope}%
\pgftext[x=4.167320in,y=2.400138in,left,base]{\rmfamily\fontsize{8.000000}{9.600000}\selectfont \(\displaystyle d=10\)}%
\end{pgfscope}%
\begin{pgfscope}%
\pgfsetrectcap%
\pgfsetroundjoin%
\pgfsetlinewidth{1.003750pt}%
\definecolor{currentstroke}{rgb}{0.121569,0.466667,0.705882}%
\pgfsetstrokecolor{currentstroke}%
\pgfsetdash{}{0pt}%
\pgfpathmoveto{\pgfqpoint{3.889542in}{2.284094in}}%
\pgfpathlineto{\pgfqpoint{4.045098in}{2.284094in}}%
\pgfusepath{stroke}%
\end{pgfscope}%
\begin{pgfscope}%
\pgftext[x=4.167320in,y=2.245205in,left,base]{\rmfamily\fontsize{8.000000}{9.600000}\selectfont \(\displaystyle d=100\)}%
\end{pgfscope}%
\end{pgfpicture}%
\makeatother%
\endgroup%

%% file: best_case.pgf
%% Creator: Matplotlib, PGF backend
%%
%% To include the figure in your LaTeX document, write
%%   \input{<filename>.pgf}
%%
%% Make sure the required packages are loaded in your preamble
%%   \usepackage{pgf}
%%
%% Figures using additional raster images can only be included by \input if
%% they are in the same directory as the main LaTeX file. For loading figures
%% from other directories you can use the `import` package
%%   \usepackage{import}
%% and then include the figures with
%%   \import{<path to file>}{<filename>.pgf}
%%
%% Matplotlib used the following preamble
%%   \usepackage[utf8x]{inputenc}
%%   \usepackage[T1]{fontenc}
%%
\begingroup%
\makeatletter%
\begin{pgfpicture}%
\pgfpathrectangle{\pgfpointorigin}{\pgfqpoint{4.929809in}{3.191777in}}%
\pgfusepath{use as bounding box, clip}%
\begin{pgfscope}%
\pgfsetbuttcap%
\pgfsetmiterjoin%
\definecolor{currentfill}{rgb}{1.000000,1.000000,1.000000}%
\pgfsetfillcolor{currentfill}%
\pgfsetlinewidth{0.000000pt}%
\definecolor{currentstroke}{rgb}{1.000000,1.000000,1.000000}%
\pgfsetstrokecolor{currentstroke}%
\pgfsetdash{}{0pt}%
\pgfpathmoveto{\pgfqpoint{0.000000in}{-0.000000in}}%
\pgfpathlineto{\pgfqpoint{4.929809in}{-0.000000in}}%
\pgfpathlineto{\pgfqpoint{4.929809in}{3.191777in}}%
\pgfpathlineto{\pgfqpoint{0.000000in}{3.191777in}}%
\pgfpathclose%
\pgfusepath{fill}%
\end{pgfscope}%
\begin{pgfscope}%
\pgfsetbuttcap%
\pgfsetmiterjoin%
\definecolor{currentfill}{rgb}{1.000000,1.000000,1.000000}%
\pgfsetfillcolor{currentfill}%
\pgfsetlinewidth{0.000000pt}%
\definecolor{currentstroke}{rgb}{0.000000,0.000000,0.000000}%
\pgfsetstrokecolor{currentstroke}%
\pgfsetstrokeopacity{0.000000}%
\pgfsetdash{}{0pt}%
\pgfpathmoveto{\pgfqpoint{0.717676in}{0.467414in}}%
\pgfpathlineto{\pgfqpoint{4.754383in}{0.467414in}}%
\pgfpathlineto{\pgfqpoint{4.754383in}{3.042715in}}%
\pgfpathlineto{\pgfqpoint{0.717676in}{3.042715in}}%
\pgfpathclose%
\pgfusepath{fill}%
\end{pgfscope}%
\begin{pgfscope}%
\pgfpathrectangle{\pgfqpoint{0.717676in}{0.467414in}}{\pgfqpoint{4.036708in}{2.575301in}} %
\pgfusepath{clip}%
\pgfsetrectcap%
\pgfsetroundjoin%
\pgfsetlinewidth{1.003750pt}%
\definecolor{currentstroke}{rgb}{1.000000,0.498039,0.054902}%
\pgfsetstrokecolor{currentstroke}%
\pgfsetdash{}{0pt}%
\pgfpathmoveto{\pgfqpoint{0.717676in}{2.674815in}}%
\pgfpathlineto{\pgfqpoint{2.061899in}{0.837154in}}%
\pgfpathlineto{\pgfqpoint{2.065936in}{0.835314in}}%
\pgfpathlineto{\pgfqpoint{4.754383in}{0.835314in}}%
\pgfpathlineto{\pgfqpoint{4.754383in}{0.835314in}}%
\pgfusepath{stroke}%
\end{pgfscope}%
\begin{pgfscope}%
\pgfpathrectangle{\pgfqpoint{0.717676in}{0.467414in}}{\pgfqpoint{4.036708in}{2.575301in}} %
\pgfusepath{clip}%
\pgfsetrectcap%
\pgfsetroundjoin%
\pgfsetlinewidth{1.003750pt}%
\definecolor{currentstroke}{rgb}{0.839216,0.152941,0.156863}%
\pgfsetstrokecolor{currentstroke}%
\pgfsetdash{}{0pt}%
\pgfpathmoveto{\pgfqpoint{0.717676in}{2.674815in}}%
\pgfpathlineto{\pgfqpoint{2.736029in}{0.835314in}}%
\pgfpathlineto{\pgfqpoint{4.754383in}{0.835314in}}%
\pgfpathlineto{\pgfqpoint{4.754383in}{0.835314in}}%
\pgfusepath{stroke}%
\end{pgfscope}%
\begin{pgfscope}%
\pgfpathrectangle{\pgfqpoint{0.717676in}{0.467414in}}{\pgfqpoint{4.036708in}{2.575301in}} %
\pgfusepath{clip}%
\pgfsetrectcap%
\pgfsetroundjoin%
\pgfsetlinewidth{1.003750pt}%
\definecolor{currentstroke}{rgb}{0.580392,0.403922,0.741176}%
\pgfsetstrokecolor{currentstroke}%
\pgfsetdash{}{0pt}%
\pgfpathmoveto{\pgfqpoint{0.717676in}{2.674815in}}%
\pgfpathlineto{\pgfqpoint{3.410160in}{0.835314in}}%
\pgfpathlineto{\pgfqpoint{4.754383in}{0.835314in}}%
\pgfpathlineto{\pgfqpoint{4.754383in}{0.835314in}}%
\pgfusepath{stroke}%
\end{pgfscope}%
\begin{pgfscope}%
\pgfpathrectangle{\pgfqpoint{0.717676in}{0.467414in}}{\pgfqpoint{4.036708in}{2.575301in}} %
\pgfusepath{clip}%
\pgfsetrectcap%
\pgfsetroundjoin%
\pgfsetlinewidth{1.003750pt}%
\definecolor{currentstroke}{rgb}{0.172549,0.627451,0.172549}%
\pgfsetstrokecolor{currentstroke}%
\pgfsetdash{}{0pt}%
\pgfpathmoveto{\pgfqpoint{0.717676in}{2.674815in}}%
\pgfpathlineto{\pgfqpoint{4.019702in}{0.835723in}}%
\pgfpathlineto{\pgfqpoint{4.031813in}{0.835314in}}%
\pgfpathlineto{\pgfqpoint{4.754383in}{0.835314in}}%
\pgfpathlineto{\pgfqpoint{4.754383in}{0.835314in}}%
\pgfusepath{stroke}%
\end{pgfscope}%
\begin{pgfscope}%
\pgfpathrectangle{\pgfqpoint{0.717676in}{0.467414in}}{\pgfqpoint{4.036708in}{2.575301in}} %
\pgfusepath{clip}%
\pgfsetrectcap%
\pgfsetroundjoin%
\pgfsetlinewidth{1.003750pt}%
\definecolor{currentstroke}{rgb}{0.121569,0.466667,0.705882}%
\pgfsetstrokecolor{currentstroke}%
\pgfsetdash{}{0pt}%
\pgfpathmoveto{\pgfqpoint{0.717676in}{2.674815in}}%
\pgfpathlineto{\pgfqpoint{4.673649in}{0.835686in}}%
\pgfpathlineto{\pgfqpoint{4.689796in}{0.835314in}}%
\pgfpathlineto{\pgfqpoint{4.754383in}{0.835314in}}%
\pgfpathlineto{\pgfqpoint{4.754383in}{0.835314in}}%
\pgfusepath{stroke}%
\end{pgfscope}%
\begin{pgfscope}%
\pgfsetrectcap%
\pgfsetmiterjoin%
\pgfsetlinewidth{1.003750pt}%
\definecolor{currentstroke}{rgb}{0.000000,0.000000,0.000000}%
\pgfsetstrokecolor{currentstroke}%
\pgfsetdash{}{0pt}%
\pgfpathmoveto{\pgfqpoint{0.717676in}{3.042715in}}%
\pgfpathlineto{\pgfqpoint{4.754383in}{3.042715in}}%
\pgfusepath{stroke}%
\end{pgfscope}%
\begin{pgfscope}%
\pgfsetrectcap%
\pgfsetmiterjoin%
\pgfsetlinewidth{1.003750pt}%
\definecolor{currentstroke}{rgb}{0.000000,0.000000,0.000000}%
\pgfsetstrokecolor{currentstroke}%
\pgfsetdash{}{0pt}%
\pgfpathmoveto{\pgfqpoint{4.754383in}{0.467414in}}%
\pgfpathlineto{\pgfqpoint{4.754383in}{3.042715in}}%
\pgfusepath{stroke}%
\end{pgfscope}%
\begin{pgfscope}%
\pgfsetrectcap%
\pgfsetmiterjoin%
\pgfsetlinewidth{1.003750pt}%
\definecolor{currentstroke}{rgb}{0.000000,0.000000,0.000000}%
\pgfsetstrokecolor{currentstroke}%
\pgfsetdash{}{0pt}%
\pgfpathmoveto{\pgfqpoint{0.717676in}{0.467414in}}%
\pgfpathlineto{\pgfqpoint{4.754383in}{0.467414in}}%
\pgfusepath{stroke}%
\end{pgfscope}%
\begin{pgfscope}%
\pgfsetrectcap%
\pgfsetmiterjoin%
\pgfsetlinewidth{1.003750pt}%
\definecolor{currentstroke}{rgb}{0.000000,0.000000,0.000000}%
\pgfsetstrokecolor{currentstroke}%
\pgfsetdash{}{0pt}%
\pgfpathmoveto{\pgfqpoint{0.717676in}{0.467414in}}%
\pgfpathlineto{\pgfqpoint{0.717676in}{3.042715in}}%
\pgfusepath{stroke}%
\end{pgfscope}%
\begin{pgfscope}%
\pgfpathrectangle{\pgfqpoint{0.717676in}{0.467414in}}{\pgfqpoint{4.036708in}{2.575301in}} %
\pgfusepath{clip}%
\pgfsetbuttcap%
\pgfsetroundjoin%
\pgfsetlinewidth{0.501875pt}%
\definecolor{currentstroke}{rgb}{0.501961,0.501961,0.501961}%
\pgfsetstrokecolor{currentstroke}%
\pgfsetdash{{1.000000pt}{3.000000pt}}{0.000000pt}%
\pgfpathmoveto{\pgfqpoint{0.717676in}{0.467414in}}%
\pgfpathlineto{\pgfqpoint{0.717676in}{3.042715in}}%
\pgfusepath{stroke}%
\end{pgfscope}%
\begin{pgfscope}%
\pgfsetbuttcap%
\pgfsetroundjoin%
\definecolor{currentfill}{rgb}{0.000000,0.000000,0.000000}%
\pgfsetfillcolor{currentfill}%
\pgfsetlinewidth{0.501875pt}%
\definecolor{currentstroke}{rgb}{0.000000,0.000000,0.000000}%
\pgfsetstrokecolor{currentstroke}%
\pgfsetdash{}{0pt}%
\pgfsys@defobject{currentmarker}{\pgfqpoint{0.000000in}{0.000000in}}{\pgfqpoint{0.000000in}{0.055556in}}{%
\pgfpathmoveto{\pgfqpoint{0.000000in}{0.000000in}}%
\pgfpathlineto{\pgfqpoint{0.000000in}{0.055556in}}%
\pgfusepath{stroke,fill}%
}%
\begin{pgfscope}%
\pgfsys@transformshift{0.717676in}{0.467414in}%
\pgfsys@useobject{currentmarker}{}%
\end{pgfscope}%
\end{pgfscope}%
\begin{pgfscope}%
\pgfsetbuttcap%
\pgfsetroundjoin%
\definecolor{currentfill}{rgb}{0.000000,0.000000,0.000000}%
\pgfsetfillcolor{currentfill}%
\pgfsetlinewidth{0.501875pt}%
\definecolor{currentstroke}{rgb}{0.000000,0.000000,0.000000}%
\pgfsetstrokecolor{currentstroke}%
\pgfsetdash{}{0pt}%
\pgfsys@defobject{currentmarker}{\pgfqpoint{0.000000in}{-0.055556in}}{\pgfqpoint{0.000000in}{0.000000in}}{%
\pgfpathmoveto{\pgfqpoint{0.000000in}{0.000000in}}%
\pgfpathlineto{\pgfqpoint{0.000000in}{-0.055556in}}%
\pgfusepath{stroke,fill}%
}%
\begin{pgfscope}%
\pgfsys@transformshift{0.717676in}{3.042715in}%
\pgfsys@useobject{currentmarker}{}%
\end{pgfscope}%
\end{pgfscope}%
\begin{pgfscope}%
\pgftext[x=0.717676in,y=0.411859in,,top]{\rmfamily\fontsize{8.000000}{9.600000}\selectfont \(\displaystyle 0.0\)}%
\end{pgfscope}%
\begin{pgfscope}%
\pgfpathrectangle{\pgfqpoint{0.717676in}{0.467414in}}{\pgfqpoint{4.036708in}{2.575301in}} %
\pgfusepath{clip}%
\pgfsetbuttcap%
\pgfsetroundjoin%
\pgfsetlinewidth{0.501875pt}%
\definecolor{currentstroke}{rgb}{0.501961,0.501961,0.501961}%
\pgfsetstrokecolor{currentstroke}%
\pgfsetdash{{1.000000pt}{3.000000pt}}{0.000000pt}%
\pgfpathmoveto{\pgfqpoint{1.525017in}{0.467414in}}%
\pgfpathlineto{\pgfqpoint{1.525017in}{3.042715in}}%
\pgfusepath{stroke}%
\end{pgfscope}%
\begin{pgfscope}%
\pgfsetbuttcap%
\pgfsetroundjoin%
\definecolor{currentfill}{rgb}{0.000000,0.000000,0.000000}%
\pgfsetfillcolor{currentfill}%
\pgfsetlinewidth{0.501875pt}%
\definecolor{currentstroke}{rgb}{0.000000,0.000000,0.000000}%
\pgfsetstrokecolor{currentstroke}%
\pgfsetdash{}{0pt}%
\pgfsys@defobject{currentmarker}{\pgfqpoint{0.000000in}{0.000000in}}{\pgfqpoint{0.000000in}{0.055556in}}{%
\pgfpathmoveto{\pgfqpoint{0.000000in}{0.000000in}}%
\pgfpathlineto{\pgfqpoint{0.000000in}{0.055556in}}%
\pgfusepath{stroke,fill}%
}%
\begin{pgfscope}%
\pgfsys@transformshift{1.525017in}{0.467414in}%
\pgfsys@useobject{currentmarker}{}%
\end{pgfscope}%
\end{pgfscope}%
\begin{pgfscope}%
\pgfsetbuttcap%
\pgfsetroundjoin%
\definecolor{currentfill}{rgb}{0.000000,0.000000,0.000000}%
\pgfsetfillcolor{currentfill}%
\pgfsetlinewidth{0.501875pt}%
\definecolor{currentstroke}{rgb}{0.000000,0.000000,0.000000}%
\pgfsetstrokecolor{currentstroke}%
\pgfsetdash{}{0pt}%
\pgfsys@defobject{currentmarker}{\pgfqpoint{0.000000in}{-0.055556in}}{\pgfqpoint{0.000000in}{0.000000in}}{%
\pgfpathmoveto{\pgfqpoint{0.000000in}{0.000000in}}%
\pgfpathlineto{\pgfqpoint{0.000000in}{-0.055556in}}%
\pgfusepath{stroke,fill}%
}%
\begin{pgfscope}%
\pgfsys@transformshift{1.525017in}{3.042715in}%
\pgfsys@useobject{currentmarker}{}%
\end{pgfscope}%
\end{pgfscope}%
\begin{pgfscope}%
\pgftext[x=1.525017in,y=0.411859in,,top]{\rmfamily\fontsize{8.000000}{9.600000}\selectfont \(\displaystyle 0.2\)}%
\end{pgfscope}%
\begin{pgfscope}%
\pgfpathrectangle{\pgfqpoint{0.717676in}{0.467414in}}{\pgfqpoint{4.036708in}{2.575301in}} %
\pgfusepath{clip}%
\pgfsetbuttcap%
\pgfsetroundjoin%
\pgfsetlinewidth{0.501875pt}%
\definecolor{currentstroke}{rgb}{0.501961,0.501961,0.501961}%
\pgfsetstrokecolor{currentstroke}%
\pgfsetdash{{1.000000pt}{3.000000pt}}{0.000000pt}%
\pgfpathmoveto{\pgfqpoint{2.332359in}{0.467414in}}%
\pgfpathlineto{\pgfqpoint{2.332359in}{3.042715in}}%
\pgfusepath{stroke}%
\end{pgfscope}%
\begin{pgfscope}%
\pgfsetbuttcap%
\pgfsetroundjoin%
\definecolor{currentfill}{rgb}{0.000000,0.000000,0.000000}%
\pgfsetfillcolor{currentfill}%
\pgfsetlinewidth{0.501875pt}%
\definecolor{currentstroke}{rgb}{0.000000,0.000000,0.000000}%
\pgfsetstrokecolor{currentstroke}%
\pgfsetdash{}{0pt}%
\pgfsys@defobject{currentmarker}{\pgfqpoint{0.000000in}{0.000000in}}{\pgfqpoint{0.000000in}{0.055556in}}{%
\pgfpathmoveto{\pgfqpoint{0.000000in}{0.000000in}}%
\pgfpathlineto{\pgfqpoint{0.000000in}{0.055556in}}%
\pgfusepath{stroke,fill}%
}%
\begin{pgfscope}%
\pgfsys@transformshift{2.332359in}{0.467414in}%
\pgfsys@useobject{currentmarker}{}%
\end{pgfscope}%
\end{pgfscope}%
\begin{pgfscope}%
\pgfsetbuttcap%
\pgfsetroundjoin%
\definecolor{currentfill}{rgb}{0.000000,0.000000,0.000000}%
\pgfsetfillcolor{currentfill}%
\pgfsetlinewidth{0.501875pt}%
\definecolor{currentstroke}{rgb}{0.000000,0.000000,0.000000}%
\pgfsetstrokecolor{currentstroke}%
\pgfsetdash{}{0pt}%
\pgfsys@defobject{currentmarker}{\pgfqpoint{0.000000in}{-0.055556in}}{\pgfqpoint{0.000000in}{0.000000in}}{%
\pgfpathmoveto{\pgfqpoint{0.000000in}{0.000000in}}%
\pgfpathlineto{\pgfqpoint{0.000000in}{-0.055556in}}%
\pgfusepath{stroke,fill}%
}%
\begin{pgfscope}%
\pgfsys@transformshift{2.332359in}{3.042715in}%
\pgfsys@useobject{currentmarker}{}%
\end{pgfscope}%
\end{pgfscope}%
\begin{pgfscope}%
\pgftext[x=2.332359in,y=0.411859in,,top]{\rmfamily\fontsize{8.000000}{9.600000}\selectfont \(\displaystyle 0.4\)}%
\end{pgfscope}%
\begin{pgfscope}%
\pgfpathrectangle{\pgfqpoint{0.717676in}{0.467414in}}{\pgfqpoint{4.036708in}{2.575301in}} %
\pgfusepath{clip}%
\pgfsetbuttcap%
\pgfsetroundjoin%
\pgfsetlinewidth{0.501875pt}%
\definecolor{currentstroke}{rgb}{0.501961,0.501961,0.501961}%
\pgfsetstrokecolor{currentstroke}%
\pgfsetdash{{1.000000pt}{3.000000pt}}{0.000000pt}%
\pgfpathmoveto{\pgfqpoint{3.139700in}{0.467414in}}%
\pgfpathlineto{\pgfqpoint{3.139700in}{3.042715in}}%
\pgfusepath{stroke}%
\end{pgfscope}%
\begin{pgfscope}%
\pgfsetbuttcap%
\pgfsetroundjoin%
\definecolor{currentfill}{rgb}{0.000000,0.000000,0.000000}%
\pgfsetfillcolor{currentfill}%
\pgfsetlinewidth{0.501875pt}%
\definecolor{currentstroke}{rgb}{0.000000,0.000000,0.000000}%
\pgfsetstrokecolor{currentstroke}%
\pgfsetdash{}{0pt}%
\pgfsys@defobject{currentmarker}{\pgfqpoint{0.000000in}{0.000000in}}{\pgfqpoint{0.000000in}{0.055556in}}{%
\pgfpathmoveto{\pgfqpoint{0.000000in}{0.000000in}}%
\pgfpathlineto{\pgfqpoint{0.000000in}{0.055556in}}%
\pgfusepath{stroke,fill}%
}%
\begin{pgfscope}%
\pgfsys@transformshift{3.139700in}{0.467414in}%
\pgfsys@useobject{currentmarker}{}%
\end{pgfscope}%
\end{pgfscope}%
\begin{pgfscope}%
\pgfsetbuttcap%
\pgfsetroundjoin%
\definecolor{currentfill}{rgb}{0.000000,0.000000,0.000000}%
\pgfsetfillcolor{currentfill}%
\pgfsetlinewidth{0.501875pt}%
\definecolor{currentstroke}{rgb}{0.000000,0.000000,0.000000}%
\pgfsetstrokecolor{currentstroke}%
\pgfsetdash{}{0pt}%
\pgfsys@defobject{currentmarker}{\pgfqpoint{0.000000in}{-0.055556in}}{\pgfqpoint{0.000000in}{0.000000in}}{%
\pgfpathmoveto{\pgfqpoint{0.000000in}{0.000000in}}%
\pgfpathlineto{\pgfqpoint{0.000000in}{-0.055556in}}%
\pgfusepath{stroke,fill}%
}%
\begin{pgfscope}%
\pgfsys@transformshift{3.139700in}{3.042715in}%
\pgfsys@useobject{currentmarker}{}%
\end{pgfscope}%
\end{pgfscope}%
\begin{pgfscope}%
\pgftext[x=3.139700in,y=0.411859in,,top]{\rmfamily\fontsize{8.000000}{9.600000}\selectfont \(\displaystyle 0.6\)}%
\end{pgfscope}%
\begin{pgfscope}%
\pgfpathrectangle{\pgfqpoint{0.717676in}{0.467414in}}{\pgfqpoint{4.036708in}{2.575301in}} %
\pgfusepath{clip}%
\pgfsetbuttcap%
\pgfsetroundjoin%
\pgfsetlinewidth{0.501875pt}%
\definecolor{currentstroke}{rgb}{0.501961,0.501961,0.501961}%
\pgfsetstrokecolor{currentstroke}%
\pgfsetdash{{1.000000pt}{3.000000pt}}{0.000000pt}%
\pgfpathmoveto{\pgfqpoint{3.947042in}{0.467414in}}%
\pgfpathlineto{\pgfqpoint{3.947042in}{3.042715in}}%
\pgfusepath{stroke}%
\end{pgfscope}%
\begin{pgfscope}%
\pgfsetbuttcap%
\pgfsetroundjoin%
\definecolor{currentfill}{rgb}{0.000000,0.000000,0.000000}%
\pgfsetfillcolor{currentfill}%
\pgfsetlinewidth{0.501875pt}%
\definecolor{currentstroke}{rgb}{0.000000,0.000000,0.000000}%
\pgfsetstrokecolor{currentstroke}%
\pgfsetdash{}{0pt}%
\pgfsys@defobject{currentmarker}{\pgfqpoint{0.000000in}{0.000000in}}{\pgfqpoint{0.000000in}{0.055556in}}{%
\pgfpathmoveto{\pgfqpoint{0.000000in}{0.000000in}}%
\pgfpathlineto{\pgfqpoint{0.000000in}{0.055556in}}%
\pgfusepath{stroke,fill}%
}%
\begin{pgfscope}%
\pgfsys@transformshift{3.947042in}{0.467414in}%
\pgfsys@useobject{currentmarker}{}%
\end{pgfscope}%
\end{pgfscope}%
\begin{pgfscope}%
\pgfsetbuttcap%
\pgfsetroundjoin%
\definecolor{currentfill}{rgb}{0.000000,0.000000,0.000000}%
\pgfsetfillcolor{currentfill}%
\pgfsetlinewidth{0.501875pt}%
\definecolor{currentstroke}{rgb}{0.000000,0.000000,0.000000}%
\pgfsetstrokecolor{currentstroke}%
\pgfsetdash{}{0pt}%
\pgfsys@defobject{currentmarker}{\pgfqpoint{0.000000in}{-0.055556in}}{\pgfqpoint{0.000000in}{0.000000in}}{%
\pgfpathmoveto{\pgfqpoint{0.000000in}{0.000000in}}%
\pgfpathlineto{\pgfqpoint{0.000000in}{-0.055556in}}%
\pgfusepath{stroke,fill}%
}%
\begin{pgfscope}%
\pgfsys@transformshift{3.947042in}{3.042715in}%
\pgfsys@useobject{currentmarker}{}%
\end{pgfscope}%
\end{pgfscope}%
\begin{pgfscope}%
\pgftext[x=3.947042in,y=0.411859in,,top]{\rmfamily\fontsize{8.000000}{9.600000}\selectfont \(\displaystyle 0.8\)}%
\end{pgfscope}%
\begin{pgfscope}%
\pgfpathrectangle{\pgfqpoint{0.717676in}{0.467414in}}{\pgfqpoint{4.036708in}{2.575301in}} %
\pgfusepath{clip}%
\pgfsetbuttcap%
\pgfsetroundjoin%
\pgfsetlinewidth{0.501875pt}%
\definecolor{currentstroke}{rgb}{0.501961,0.501961,0.501961}%
\pgfsetstrokecolor{currentstroke}%
\pgfsetdash{{1.000000pt}{3.000000pt}}{0.000000pt}%
\pgfpathmoveto{\pgfqpoint{4.754383in}{0.467414in}}%
\pgfpathlineto{\pgfqpoint{4.754383in}{3.042715in}}%
\pgfusepath{stroke}%
\end{pgfscope}%
\begin{pgfscope}%
\pgfsetbuttcap%
\pgfsetroundjoin%
\definecolor{currentfill}{rgb}{0.000000,0.000000,0.000000}%
\pgfsetfillcolor{currentfill}%
\pgfsetlinewidth{0.501875pt}%
\definecolor{currentstroke}{rgb}{0.000000,0.000000,0.000000}%
\pgfsetstrokecolor{currentstroke}%
\pgfsetdash{}{0pt}%
\pgfsys@defobject{currentmarker}{\pgfqpoint{0.000000in}{0.000000in}}{\pgfqpoint{0.000000in}{0.055556in}}{%
\pgfpathmoveto{\pgfqpoint{0.000000in}{0.000000in}}%
\pgfpathlineto{\pgfqpoint{0.000000in}{0.055556in}}%
\pgfusepath{stroke,fill}%
}%
\begin{pgfscope}%
\pgfsys@transformshift{4.754383in}{0.467414in}%
\pgfsys@useobject{currentmarker}{}%
\end{pgfscope}%
\end{pgfscope}%
\begin{pgfscope}%
\pgfsetbuttcap%
\pgfsetroundjoin%
\definecolor{currentfill}{rgb}{0.000000,0.000000,0.000000}%
\pgfsetfillcolor{currentfill}%
\pgfsetlinewidth{0.501875pt}%
\definecolor{currentstroke}{rgb}{0.000000,0.000000,0.000000}%
\pgfsetstrokecolor{currentstroke}%
\pgfsetdash{}{0pt}%
\pgfsys@defobject{currentmarker}{\pgfqpoint{0.000000in}{-0.055556in}}{\pgfqpoint{0.000000in}{0.000000in}}{%
\pgfpathmoveto{\pgfqpoint{0.000000in}{0.000000in}}%
\pgfpathlineto{\pgfqpoint{0.000000in}{-0.055556in}}%
\pgfusepath{stroke,fill}%
}%
\begin{pgfscope}%
\pgfsys@transformshift{4.754383in}{3.042715in}%
\pgfsys@useobject{currentmarker}{}%
\end{pgfscope}%
\end{pgfscope}%
\begin{pgfscope}%
\pgftext[x=4.754383in,y=0.411859in,,top]{\rmfamily\fontsize{8.000000}{9.600000}\selectfont \(\displaystyle 1.0\)}%
\end{pgfscope}%
\begin{pgfscope}%
\pgftext[x=2.736029in,y=0.244290in,,top]{\rmfamily\fontsize{10.000000}{12.000000}\selectfont \(\displaystyle \mathcal{T}_{A|B}\)}%
\end{pgfscope}%
\begin{pgfscope}%
\pgfpathrectangle{\pgfqpoint{0.717676in}{0.467414in}}{\pgfqpoint{4.036708in}{2.575301in}} %
\pgfusepath{clip}%
\pgfsetbuttcap%
\pgfsetroundjoin%
\pgfsetlinewidth{0.501875pt}%
\definecolor{currentstroke}{rgb}{0.501961,0.501961,0.501961}%
\pgfsetstrokecolor{currentstroke}%
\pgfsetdash{{1.000000pt}{3.000000pt}}{0.000000pt}%
\pgfpathmoveto{\pgfqpoint{0.717676in}{0.467414in}}%
\pgfpathlineto{\pgfqpoint{4.754383in}{0.467414in}}%
\pgfusepath{stroke}%
\end{pgfscope}%
\begin{pgfscope}%
\pgfsetbuttcap%
\pgfsetroundjoin%
\definecolor{currentfill}{rgb}{0.000000,0.000000,0.000000}%
\pgfsetfillcolor{currentfill}%
\pgfsetlinewidth{0.501875pt}%
\definecolor{currentstroke}{rgb}{0.000000,0.000000,0.000000}%
\pgfsetstrokecolor{currentstroke}%
\pgfsetdash{}{0pt}%
\pgfsys@defobject{currentmarker}{\pgfqpoint{0.000000in}{0.000000in}}{\pgfqpoint{0.055556in}{0.000000in}}{%
\pgfpathmoveto{\pgfqpoint{0.000000in}{0.000000in}}%
\pgfpathlineto{\pgfqpoint{0.055556in}{0.000000in}}%
\pgfusepath{stroke,fill}%
}%
\begin{pgfscope}%
\pgfsys@transformshift{0.717676in}{0.467414in}%
\pgfsys@useobject{currentmarker}{}%
\end{pgfscope}%
\end{pgfscope}%
\begin{pgfscope}%
\pgfsetbuttcap%
\pgfsetroundjoin%
\definecolor{currentfill}{rgb}{0.000000,0.000000,0.000000}%
\pgfsetfillcolor{currentfill}%
\pgfsetlinewidth{0.501875pt}%
\definecolor{currentstroke}{rgb}{0.000000,0.000000,0.000000}%
\pgfsetstrokecolor{currentstroke}%
\pgfsetdash{}{0pt}%
\pgfsys@defobject{currentmarker}{\pgfqpoint{-0.055556in}{0.000000in}}{\pgfqpoint{0.000000in}{0.000000in}}{%
\pgfpathmoveto{\pgfqpoint{0.000000in}{0.000000in}}%
\pgfpathlineto{\pgfqpoint{-0.055556in}{0.000000in}}%
\pgfusepath{stroke,fill}%
}%
\begin{pgfscope}%
\pgfsys@transformshift{4.754383in}{0.467414in}%
\pgfsys@useobject{currentmarker}{}%
\end{pgfscope}%
\end{pgfscope}%
\begin{pgfscope}%
\pgftext[x=0.662120in,y=0.467414in,right,]{\rmfamily\fontsize{8.000000}{9.600000}\selectfont \(\displaystyle -0.2\)}%
\end{pgfscope}%
\begin{pgfscope}%
\pgfpathrectangle{\pgfqpoint{0.717676in}{0.467414in}}{\pgfqpoint{4.036708in}{2.575301in}} %
\pgfusepath{clip}%
\pgfsetbuttcap%
\pgfsetroundjoin%
\pgfsetlinewidth{0.501875pt}%
\definecolor{currentstroke}{rgb}{0.501961,0.501961,0.501961}%
\pgfsetstrokecolor{currentstroke}%
\pgfsetdash{{1.000000pt}{3.000000pt}}{0.000000pt}%
\pgfpathmoveto{\pgfqpoint{0.717676in}{0.835314in}}%
\pgfpathlineto{\pgfqpoint{4.754383in}{0.835314in}}%
\pgfusepath{stroke}%
\end{pgfscope}%
\begin{pgfscope}%
\pgfsetbuttcap%
\pgfsetroundjoin%
\definecolor{currentfill}{rgb}{0.000000,0.000000,0.000000}%
\pgfsetfillcolor{currentfill}%
\pgfsetlinewidth{0.501875pt}%
\definecolor{currentstroke}{rgb}{0.000000,0.000000,0.000000}%
\pgfsetstrokecolor{currentstroke}%
\pgfsetdash{}{0pt}%
\pgfsys@defobject{currentmarker}{\pgfqpoint{0.000000in}{0.000000in}}{\pgfqpoint{0.055556in}{0.000000in}}{%
\pgfpathmoveto{\pgfqpoint{0.000000in}{0.000000in}}%
\pgfpathlineto{\pgfqpoint{0.055556in}{0.000000in}}%
\pgfusepath{stroke,fill}%
}%
\begin{pgfscope}%
\pgfsys@transformshift{0.717676in}{0.835314in}%
\pgfsys@useobject{currentmarker}{}%
\end{pgfscope}%
\end{pgfscope}%
\begin{pgfscope}%
\pgfsetbuttcap%
\pgfsetroundjoin%
\definecolor{currentfill}{rgb}{0.000000,0.000000,0.000000}%
\pgfsetfillcolor{currentfill}%
\pgfsetlinewidth{0.501875pt}%
\definecolor{currentstroke}{rgb}{0.000000,0.000000,0.000000}%
\pgfsetstrokecolor{currentstroke}%
\pgfsetdash{}{0pt}%
\pgfsys@defobject{currentmarker}{\pgfqpoint{-0.055556in}{0.000000in}}{\pgfqpoint{0.000000in}{0.000000in}}{%
\pgfpathmoveto{\pgfqpoint{0.000000in}{0.000000in}}%
\pgfpathlineto{\pgfqpoint{-0.055556in}{0.000000in}}%
\pgfusepath{stroke,fill}%
}%
\begin{pgfscope}%
\pgfsys@transformshift{4.754383in}{0.835314in}%
\pgfsys@useobject{currentmarker}{}%
\end{pgfscope}%
\end{pgfscope}%
\begin{pgfscope}%
\pgftext[x=0.662120in,y=0.835314in,right,]{\rmfamily\fontsize{8.000000}{9.600000}\selectfont \(\displaystyle 0.0\)}%
\end{pgfscope}%
\begin{pgfscope}%
\pgfpathrectangle{\pgfqpoint{0.717676in}{0.467414in}}{\pgfqpoint{4.036708in}{2.575301in}} %
\pgfusepath{clip}%
\pgfsetbuttcap%
\pgfsetroundjoin%
\pgfsetlinewidth{0.501875pt}%
\definecolor{currentstroke}{rgb}{0.501961,0.501961,0.501961}%
\pgfsetstrokecolor{currentstroke}%
\pgfsetdash{{1.000000pt}{3.000000pt}}{0.000000pt}%
\pgfpathmoveto{\pgfqpoint{0.717676in}{1.203215in}}%
\pgfpathlineto{\pgfqpoint{4.754383in}{1.203215in}}%
\pgfusepath{stroke}%
\end{pgfscope}%
\begin{pgfscope}%
\pgfsetbuttcap%
\pgfsetroundjoin%
\definecolor{currentfill}{rgb}{0.000000,0.000000,0.000000}%
\pgfsetfillcolor{currentfill}%
\pgfsetlinewidth{0.501875pt}%
\definecolor{currentstroke}{rgb}{0.000000,0.000000,0.000000}%
\pgfsetstrokecolor{currentstroke}%
\pgfsetdash{}{0pt}%
\pgfsys@defobject{currentmarker}{\pgfqpoint{0.000000in}{0.000000in}}{\pgfqpoint{0.055556in}{0.000000in}}{%
\pgfpathmoveto{\pgfqpoint{0.000000in}{0.000000in}}%
\pgfpathlineto{\pgfqpoint{0.055556in}{0.000000in}}%
\pgfusepath{stroke,fill}%
}%
\begin{pgfscope}%
\pgfsys@transformshift{0.717676in}{1.203215in}%
\pgfsys@useobject{currentmarker}{}%
\end{pgfscope}%
\end{pgfscope}%
\begin{pgfscope}%
\pgfsetbuttcap%
\pgfsetroundjoin%
\definecolor{currentfill}{rgb}{0.000000,0.000000,0.000000}%
\pgfsetfillcolor{currentfill}%
\pgfsetlinewidth{0.501875pt}%
\definecolor{currentstroke}{rgb}{0.000000,0.000000,0.000000}%
\pgfsetstrokecolor{currentstroke}%
\pgfsetdash{}{0pt}%
\pgfsys@defobject{currentmarker}{\pgfqpoint{-0.055556in}{0.000000in}}{\pgfqpoint{0.000000in}{0.000000in}}{%
\pgfpathmoveto{\pgfqpoint{0.000000in}{0.000000in}}%
\pgfpathlineto{\pgfqpoint{-0.055556in}{0.000000in}}%
\pgfusepath{stroke,fill}%
}%
\begin{pgfscope}%
\pgfsys@transformshift{4.754383in}{1.203215in}%
\pgfsys@useobject{currentmarker}{}%
\end{pgfscope}%
\end{pgfscope}%
\begin{pgfscope}%
\pgftext[x=0.662120in,y=1.203215in,right,]{\rmfamily\fontsize{8.000000}{9.600000}\selectfont \(\displaystyle 0.2\)}%
\end{pgfscope}%
\begin{pgfscope}%
\pgfpathrectangle{\pgfqpoint{0.717676in}{0.467414in}}{\pgfqpoint{4.036708in}{2.575301in}} %
\pgfusepath{clip}%
\pgfsetbuttcap%
\pgfsetroundjoin%
\pgfsetlinewidth{0.501875pt}%
\definecolor{currentstroke}{rgb}{0.501961,0.501961,0.501961}%
\pgfsetstrokecolor{currentstroke}%
\pgfsetdash{{1.000000pt}{3.000000pt}}{0.000000pt}%
\pgfpathmoveto{\pgfqpoint{0.717676in}{1.571115in}}%
\pgfpathlineto{\pgfqpoint{4.754383in}{1.571115in}}%
\pgfusepath{stroke}%
\end{pgfscope}%
\begin{pgfscope}%
\pgfsetbuttcap%
\pgfsetroundjoin%
\definecolor{currentfill}{rgb}{0.000000,0.000000,0.000000}%
\pgfsetfillcolor{currentfill}%
\pgfsetlinewidth{0.501875pt}%
\definecolor{currentstroke}{rgb}{0.000000,0.000000,0.000000}%
\pgfsetstrokecolor{currentstroke}%
\pgfsetdash{}{0pt}%
\pgfsys@defobject{currentmarker}{\pgfqpoint{0.000000in}{0.000000in}}{\pgfqpoint{0.055556in}{0.000000in}}{%
\pgfpathmoveto{\pgfqpoint{0.000000in}{0.000000in}}%
\pgfpathlineto{\pgfqpoint{0.055556in}{0.000000in}}%
\pgfusepath{stroke,fill}%
}%
\begin{pgfscope}%
\pgfsys@transformshift{0.717676in}{1.571115in}%
\pgfsys@useobject{currentmarker}{}%
\end{pgfscope}%
\end{pgfscope}%
\begin{pgfscope}%
\pgfsetbuttcap%
\pgfsetroundjoin%
\definecolor{currentfill}{rgb}{0.000000,0.000000,0.000000}%
\pgfsetfillcolor{currentfill}%
\pgfsetlinewidth{0.501875pt}%
\definecolor{currentstroke}{rgb}{0.000000,0.000000,0.000000}%
\pgfsetstrokecolor{currentstroke}%
\pgfsetdash{}{0pt}%
\pgfsys@defobject{currentmarker}{\pgfqpoint{-0.055556in}{0.000000in}}{\pgfqpoint{0.000000in}{0.000000in}}{%
\pgfpathmoveto{\pgfqpoint{0.000000in}{0.000000in}}%
\pgfpathlineto{\pgfqpoint{-0.055556in}{0.000000in}}%
\pgfusepath{stroke,fill}%
}%
\begin{pgfscope}%
\pgfsys@transformshift{4.754383in}{1.571115in}%
\pgfsys@useobject{currentmarker}{}%
\end{pgfscope}%
\end{pgfscope}%
\begin{pgfscope}%
\pgftext[x=0.662120in,y=1.571115in,right,]{\rmfamily\fontsize{8.000000}{9.600000}\selectfont \(\displaystyle 0.4\)}%
\end{pgfscope}%
\begin{pgfscope}%
\pgfpathrectangle{\pgfqpoint{0.717676in}{0.467414in}}{\pgfqpoint{4.036708in}{2.575301in}} %
\pgfusepath{clip}%
\pgfsetbuttcap%
\pgfsetroundjoin%
\pgfsetlinewidth{0.501875pt}%
\definecolor{currentstroke}{rgb}{0.501961,0.501961,0.501961}%
\pgfsetstrokecolor{currentstroke}%
\pgfsetdash{{1.000000pt}{3.000000pt}}{0.000000pt}%
\pgfpathmoveto{\pgfqpoint{0.717676in}{1.939015in}}%
\pgfpathlineto{\pgfqpoint{4.754383in}{1.939015in}}%
\pgfusepath{stroke}%
\end{pgfscope}%
\begin{pgfscope}%
\pgfsetbuttcap%
\pgfsetroundjoin%
\definecolor{currentfill}{rgb}{0.000000,0.000000,0.000000}%
\pgfsetfillcolor{currentfill}%
\pgfsetlinewidth{0.501875pt}%
\definecolor{currentstroke}{rgb}{0.000000,0.000000,0.000000}%
\pgfsetstrokecolor{currentstroke}%
\pgfsetdash{}{0pt}%
\pgfsys@defobject{currentmarker}{\pgfqpoint{0.000000in}{0.000000in}}{\pgfqpoint{0.055556in}{0.000000in}}{%
\pgfpathmoveto{\pgfqpoint{0.000000in}{0.000000in}}%
\pgfpathlineto{\pgfqpoint{0.055556in}{0.000000in}}%
\pgfusepath{stroke,fill}%
}%
\begin{pgfscope}%
\pgfsys@transformshift{0.717676in}{1.939015in}%
\pgfsys@useobject{currentmarker}{}%
\end{pgfscope}%
\end{pgfscope}%
\begin{pgfscope}%
\pgfsetbuttcap%
\pgfsetroundjoin%
\definecolor{currentfill}{rgb}{0.000000,0.000000,0.000000}%
\pgfsetfillcolor{currentfill}%
\pgfsetlinewidth{0.501875pt}%
\definecolor{currentstroke}{rgb}{0.000000,0.000000,0.000000}%
\pgfsetstrokecolor{currentstroke}%
\pgfsetdash{}{0pt}%
\pgfsys@defobject{currentmarker}{\pgfqpoint{-0.055556in}{0.000000in}}{\pgfqpoint{0.000000in}{0.000000in}}{%
\pgfpathmoveto{\pgfqpoint{0.000000in}{0.000000in}}%
\pgfpathlineto{\pgfqpoint{-0.055556in}{0.000000in}}%
\pgfusepath{stroke,fill}%
}%
\begin{pgfscope}%
\pgfsys@transformshift{4.754383in}{1.939015in}%
\pgfsys@useobject{currentmarker}{}%
\end{pgfscope}%
\end{pgfscope}%
\begin{pgfscope}%
\pgftext[x=0.662120in,y=1.939015in,right,]{\rmfamily\fontsize{8.000000}{9.600000}\selectfont \(\displaystyle 0.6\)}%
\end{pgfscope}%
\begin{pgfscope}%
\pgfpathrectangle{\pgfqpoint{0.717676in}{0.467414in}}{\pgfqpoint{4.036708in}{2.575301in}} %
\pgfusepath{clip}%
\pgfsetbuttcap%
\pgfsetroundjoin%
\pgfsetlinewidth{0.501875pt}%
\definecolor{currentstroke}{rgb}{0.501961,0.501961,0.501961}%
\pgfsetstrokecolor{currentstroke}%
\pgfsetdash{{1.000000pt}{3.000000pt}}{0.000000pt}%
\pgfpathmoveto{\pgfqpoint{0.717676in}{2.306915in}}%
\pgfpathlineto{\pgfqpoint{4.754383in}{2.306915in}}%
\pgfusepath{stroke}%
\end{pgfscope}%
\begin{pgfscope}%
\pgfsetbuttcap%
\pgfsetroundjoin%
\definecolor{currentfill}{rgb}{0.000000,0.000000,0.000000}%
\pgfsetfillcolor{currentfill}%
\pgfsetlinewidth{0.501875pt}%
\definecolor{currentstroke}{rgb}{0.000000,0.000000,0.000000}%
\pgfsetstrokecolor{currentstroke}%
\pgfsetdash{}{0pt}%
\pgfsys@defobject{currentmarker}{\pgfqpoint{0.000000in}{0.000000in}}{\pgfqpoint{0.055556in}{0.000000in}}{%
\pgfpathmoveto{\pgfqpoint{0.000000in}{0.000000in}}%
\pgfpathlineto{\pgfqpoint{0.055556in}{0.000000in}}%
\pgfusepath{stroke,fill}%
}%
\begin{pgfscope}%
\pgfsys@transformshift{0.717676in}{2.306915in}%
\pgfsys@useobject{currentmarker}{}%
\end{pgfscope}%
\end{pgfscope}%
\begin{pgfscope}%
\pgfsetbuttcap%
\pgfsetroundjoin%
\definecolor{currentfill}{rgb}{0.000000,0.000000,0.000000}%
\pgfsetfillcolor{currentfill}%
\pgfsetlinewidth{0.501875pt}%
\definecolor{currentstroke}{rgb}{0.000000,0.000000,0.000000}%
\pgfsetstrokecolor{currentstroke}%
\pgfsetdash{}{0pt}%
\pgfsys@defobject{currentmarker}{\pgfqpoint{-0.055556in}{0.000000in}}{\pgfqpoint{0.000000in}{0.000000in}}{%
\pgfpathmoveto{\pgfqpoint{0.000000in}{0.000000in}}%
\pgfpathlineto{\pgfqpoint{-0.055556in}{0.000000in}}%
\pgfusepath{stroke,fill}%
}%
\begin{pgfscope}%
\pgfsys@transformshift{4.754383in}{2.306915in}%
\pgfsys@useobject{currentmarker}{}%
\end{pgfscope}%
\end{pgfscope}%
\begin{pgfscope}%
\pgftext[x=0.662120in,y=2.306915in,right,]{\rmfamily\fontsize{8.000000}{9.600000}\selectfont \(\displaystyle 0.8\)}%
\end{pgfscope}%
\begin{pgfscope}%
\pgfpathrectangle{\pgfqpoint{0.717676in}{0.467414in}}{\pgfqpoint{4.036708in}{2.575301in}} %
\pgfusepath{clip}%
\pgfsetbuttcap%
\pgfsetroundjoin%
\pgfsetlinewidth{0.501875pt}%
\definecolor{currentstroke}{rgb}{0.501961,0.501961,0.501961}%
\pgfsetstrokecolor{currentstroke}%
\pgfsetdash{{1.000000pt}{3.000000pt}}{0.000000pt}%
\pgfpathmoveto{\pgfqpoint{0.717676in}{2.674815in}}%
\pgfpathlineto{\pgfqpoint{4.754383in}{2.674815in}}%
\pgfusepath{stroke}%
\end{pgfscope}%
\begin{pgfscope}%
\pgfsetbuttcap%
\pgfsetroundjoin%
\definecolor{currentfill}{rgb}{0.000000,0.000000,0.000000}%
\pgfsetfillcolor{currentfill}%
\pgfsetlinewidth{0.501875pt}%
\definecolor{currentstroke}{rgb}{0.000000,0.000000,0.000000}%
\pgfsetstrokecolor{currentstroke}%
\pgfsetdash{}{0pt}%
\pgfsys@defobject{currentmarker}{\pgfqpoint{0.000000in}{0.000000in}}{\pgfqpoint{0.055556in}{0.000000in}}{%
\pgfpathmoveto{\pgfqpoint{0.000000in}{0.000000in}}%
\pgfpathlineto{\pgfqpoint{0.055556in}{0.000000in}}%
\pgfusepath{stroke,fill}%
}%
\begin{pgfscope}%
\pgfsys@transformshift{0.717676in}{2.674815in}%
\pgfsys@useobject{currentmarker}{}%
\end{pgfscope}%
\end{pgfscope}%
\begin{pgfscope}%
\pgfsetbuttcap%
\pgfsetroundjoin%
\definecolor{currentfill}{rgb}{0.000000,0.000000,0.000000}%
\pgfsetfillcolor{currentfill}%
\pgfsetlinewidth{0.501875pt}%
\definecolor{currentstroke}{rgb}{0.000000,0.000000,0.000000}%
\pgfsetstrokecolor{currentstroke}%
\pgfsetdash{}{0pt}%
\pgfsys@defobject{currentmarker}{\pgfqpoint{-0.055556in}{0.000000in}}{\pgfqpoint{0.000000in}{0.000000in}}{%
\pgfpathmoveto{\pgfqpoint{0.000000in}{0.000000in}}%
\pgfpathlineto{\pgfqpoint{-0.055556in}{0.000000in}}%
\pgfusepath{stroke,fill}%
}%
\begin{pgfscope}%
\pgfsys@transformshift{4.754383in}{2.674815in}%
\pgfsys@useobject{currentmarker}{}%
\end{pgfscope}%
\end{pgfscope}%
\begin{pgfscope}%
\pgftext[x=0.662120in,y=2.674815in,right,]{\rmfamily\fontsize{8.000000}{9.600000}\selectfont \(\displaystyle 1.0\)}%
\end{pgfscope}%
\begin{pgfscope}%
\pgfpathrectangle{\pgfqpoint{0.717676in}{0.467414in}}{\pgfqpoint{4.036708in}{2.575301in}} %
\pgfusepath{clip}%
\pgfsetbuttcap%
\pgfsetroundjoin%
\pgfsetlinewidth{0.501875pt}%
\definecolor{currentstroke}{rgb}{0.501961,0.501961,0.501961}%
\pgfsetstrokecolor{currentstroke}%
\pgfsetdash{{1.000000pt}{3.000000pt}}{0.000000pt}%
\pgfpathmoveto{\pgfqpoint{0.717676in}{3.042715in}}%
\pgfpathlineto{\pgfqpoint{4.754383in}{3.042715in}}%
\pgfusepath{stroke}%
\end{pgfscope}%
\begin{pgfscope}%
\pgfsetbuttcap%
\pgfsetroundjoin%
\definecolor{currentfill}{rgb}{0.000000,0.000000,0.000000}%
\pgfsetfillcolor{currentfill}%
\pgfsetlinewidth{0.501875pt}%
\definecolor{currentstroke}{rgb}{0.000000,0.000000,0.000000}%
\pgfsetstrokecolor{currentstroke}%
\pgfsetdash{}{0pt}%
\pgfsys@defobject{currentmarker}{\pgfqpoint{0.000000in}{0.000000in}}{\pgfqpoint{0.055556in}{0.000000in}}{%
\pgfpathmoveto{\pgfqpoint{0.000000in}{0.000000in}}%
\pgfpathlineto{\pgfqpoint{0.055556in}{0.000000in}}%
\pgfusepath{stroke,fill}%
}%
\begin{pgfscope}%
\pgfsys@transformshift{0.717676in}{3.042715in}%
\pgfsys@useobject{currentmarker}{}%
\end{pgfscope}%
\end{pgfscope}%
\begin{pgfscope}%
\pgfsetbuttcap%
\pgfsetroundjoin%
\definecolor{currentfill}{rgb}{0.000000,0.000000,0.000000}%
\pgfsetfillcolor{currentfill}%
\pgfsetlinewidth{0.501875pt}%
\definecolor{currentstroke}{rgb}{0.000000,0.000000,0.000000}%
\pgfsetstrokecolor{currentstroke}%
\pgfsetdash{}{0pt}%
\pgfsys@defobject{currentmarker}{\pgfqpoint{-0.055556in}{0.000000in}}{\pgfqpoint{0.000000in}{0.000000in}}{%
\pgfpathmoveto{\pgfqpoint{0.000000in}{0.000000in}}%
\pgfpathlineto{\pgfqpoint{-0.055556in}{0.000000in}}%
\pgfusepath{stroke,fill}%
}%
\begin{pgfscope}%
\pgfsys@transformshift{4.754383in}{3.042715in}%
\pgfsys@useobject{currentmarker}{}%
\end{pgfscope}%
\end{pgfscope}%
\begin{pgfscope}%
\pgftext[x=0.662120in,y=3.042715in,right,]{\rmfamily\fontsize{8.000000}{9.600000}\selectfont \(\displaystyle 1.2\)}%
\end{pgfscope}%
\begin{pgfscope}%
\pgftext[x=0.350002in,y=1.755065in,,bottom,rotate=90.000000]{\rmfamily\fontsize{10.000000}{12.000000}\selectfont \(\displaystyle \max\left(\mathcal{T}_{AB|E}^{excess}\right)\)}%
\end{pgfscope}%
\begin{pgfscope}%
\pgfsetbuttcap%
\pgfsetmiterjoin%
\definecolor{currentfill}{rgb}{1.000000,1.000000,1.000000}%
\pgfsetfillcolor{currentfill}%
\pgfsetlinewidth{1.003750pt}%
\definecolor{currentstroke}{rgb}{0.000000,0.000000,0.000000}%
\pgfsetstrokecolor{currentstroke}%
\pgfsetdash{}{0pt}%
\pgfpathmoveto{\pgfqpoint{3.903587in}{2.179161in}}%
\pgfpathlineto{\pgfqpoint{4.698828in}{2.179161in}}%
\pgfpathlineto{\pgfqpoint{4.698828in}{2.987159in}}%
\pgfpathlineto{\pgfqpoint{3.903587in}{2.987159in}}%
\pgfpathclose%
\pgfusepath{stroke,fill}%
\end{pgfscope}%
\begin{pgfscope}%
\pgfsetrectcap%
\pgfsetroundjoin%
\pgfsetlinewidth{1.003750pt}%
\definecolor{currentstroke}{rgb}{1.000000,0.498039,0.054902}%
\pgfsetstrokecolor{currentstroke}%
\pgfsetdash{}{0pt}%
\pgfpathmoveto{\pgfqpoint{3.981365in}{2.903826in}}%
\pgfpathlineto{\pgfqpoint{4.136920in}{2.903826in}}%
\pgfusepath{stroke}%
\end{pgfscope}%
\begin{pgfscope}%
\pgftext[x=4.259142in,y=2.864937in,left,base]{\rmfamily\fontsize{8.000000}{9.600000}\selectfont \(\displaystyle d=2\)}%
\end{pgfscope}%
\begin{pgfscope}%
\pgfsetrectcap%
\pgfsetroundjoin%
\pgfsetlinewidth{1.003750pt}%
\definecolor{currentstroke}{rgb}{0.839216,0.152941,0.156863}%
\pgfsetstrokecolor{currentstroke}%
\pgfsetdash{}{0pt}%
\pgfpathmoveto{\pgfqpoint{3.981365in}{2.748893in}}%
\pgfpathlineto{\pgfqpoint{4.136920in}{2.748893in}}%
\pgfusepath{stroke}%
\end{pgfscope}%
\begin{pgfscope}%
\pgftext[x=4.259142in,y=2.710004in,left,base]{\rmfamily\fontsize{8.000000}{9.600000}\selectfont \(\displaystyle d=3\)}%
\end{pgfscope}%
\begin{pgfscope}%
\pgfsetrectcap%
\pgfsetroundjoin%
\pgfsetlinewidth{1.003750pt}%
\definecolor{currentstroke}{rgb}{0.580392,0.403922,0.741176}%
\pgfsetstrokecolor{currentstroke}%
\pgfsetdash{}{0pt}%
\pgfpathmoveto{\pgfqpoint{3.981365in}{2.593960in}}%
\pgfpathlineto{\pgfqpoint{4.136920in}{2.593960in}}%
\pgfusepath{stroke}%
\end{pgfscope}%
\begin{pgfscope}%
\pgftext[x=4.259142in,y=2.555071in,left,base]{\rmfamily\fontsize{8.000000}{9.600000}\selectfont \(\displaystyle d=5\)}%
\end{pgfscope}%
\begin{pgfscope}%
\pgfsetrectcap%
\pgfsetroundjoin%
\pgfsetlinewidth{1.003750pt}%
\definecolor{currentstroke}{rgb}{0.172549,0.627451,0.172549}%
\pgfsetstrokecolor{currentstroke}%
\pgfsetdash{}{0pt}%
\pgfpathmoveto{\pgfqpoint{3.981365in}{2.439027in}}%
\pgfpathlineto{\pgfqpoint{4.136920in}{2.439027in}}%
\pgfusepath{stroke}%
\end{pgfscope}%
\begin{pgfscope}%
\pgftext[x=4.259142in,y=2.400138in,left,base]{\rmfamily\fontsize{8.000000}{9.600000}\selectfont \(\displaystyle d=10\)}%
\end{pgfscope}%
\begin{pgfscope}%
\pgfsetrectcap%
\pgfsetroundjoin%
\pgfsetlinewidth{1.003750pt}%
\definecolor{currentstroke}{rgb}{0.121569,0.466667,0.705882}%
\pgfsetstrokecolor{currentstroke}%
\pgfsetdash{}{0pt}%
\pgfpathmoveto{\pgfqpoint{3.981365in}{2.284094in}}%
\pgfpathlineto{\pgfqpoint{4.136920in}{2.284094in}}%
\pgfusepath{stroke}%
\end{pgfscope}%
\begin{pgfscope}%
\pgftext[x=4.259142in,y=2.245205in,left,base]{\rmfamily\fontsize{8.000000}{9.600000}\selectfont \(\displaystyle d=100\)}%
\end{pgfscope}%
\end{pgfpicture}%
\makeatother%
\endgroup%